\documentclass[10pt]{article}
\usepackage{amssymb,amsbsy,latexsym,amsmath,bbm,epsfig,psfrag,amsthm,mathrsfs,mathabx}
\usepackage[swedish,english]{babel}
\usepackage{graphicx,tikz,esint}
\usepackage[T1]{fontenc}
\usepackage[latin1]{inputenc}
\usepackage{amsfonts}
\usepackage{amsmath}
\usepackage{amssymb}
\usepackage{epsf}
\usepackage{graphics}
\usepackage{graphicx}
\usepackage{latexsym}
\usepackage{psfrag}
\usepackage{relsize}
\usepackage{verbatim}
\usepackage{hyperref}
\usepackage{pgfplots}
\usepackage{float}
\usetikzlibrary{positioning}
\usetikzlibrary{arrows}
\usetikzlibrary{decorations.pathreplacing}

\theoremstyle{plain}

\newtheorem{Lem}{Lemma}
\numberwithin{Lem}{section}
\newtheorem{Prop}{Proposition}
\numberwithin{Prop}{section}
\newtheorem{Thm}{Theorem}
\numberwithin{Thm}{section}
\newtheorem{Cor}{Corollary}
\numberwithin{Cor}{section}
\newtheorem{Con}{Conjecture}
\numberwithin{Con}{section}

\theoremstyle{definition}
\newtheorem{Def}{Definition}
\numberwithin{Def}{section}
\newtheorem{hyp}{Hypothesis}
\numberwithin{hyp}{section}

\numberwithin{conj}{section}
\newtheorem{ex}{Example}
\numberwithin{ex}{section}

\theoremstyle{remark}
\newtheorem{rem}{\bf{Remark}}
\numberwithin{rem}{section}

\numberwithin{equation}{section}

\newcommand{\dv}{\partial}

\newcommand{\eps}{\varepsilon}

\newcommand{\R}{{\mathbb R}}
\newcommand{\C}{{\mathbb C}}
\newcommand{\Z}{{\mathbb Z}}

\newcommand{\N}{{\mathbb N}}
\newcommand{\G}{\Gamma}

\newcommand{\x}{\chi}

\newcommand{\Hp}{\mathbb{H}}

\newcommand{\Sn}{\mathcal{S}_{nt}(\mu)}
\newcommand{\supp}{\text{supp}}
\newcommand{\Siso}{\dv \mathcal{S}_{nt}^{iso}(\mu)}
\newcommand{\Ssing}{\mathcal{S}_{nt}^{sing}(\mu)}
\newcommand{\Sgen}{\mathcal{S}_{nt}^{gen}(\mu)}
\newcommand{\Sreg}{\mathcal{S}_{nt}^{reg}(\mu)}
\newcommand{\SsingI}{\mathcal{S}_{nt}^{sing,I}(\mu)}
\newcommand{\SsingII}{\mathcal{S}_{nt}^{sing,II}(\mu)}
\newcommand{\SsingIII}{\mathcal{S}_{nt}^{sing,III}(\mu)}
\newcommand{\SsingIV}{\mathcal{S}_{nt}^{sing,IV}(\mu)}
\newcommand{\Hc}{\text{\sffamily{C}}_{c,1}^{\lambda,\alpha}(\mathbb{R})}

\newcommand{\Dl}{\partial\mathcal{L}}
\newcommand{\Hil}{\mathcal{H}}
\newcommand{\EE}{\mathcal{E}}
\newcommand{\LL}{\mathcal{L}}

\begin{document}

\vspace{.4cm}

\title{Asymptotic Geometry of Discrete Interlaced Patterns: Part II}
\author{Erik Duse \and Anthony Metcalfe}
\maketitle

\begin{abstract}
We study the boundary of the \emph{liquid region} $\LL$ in large random lozenge tiling models defined by 
uniform random interlacing particle systems with general initial configuration, which lies on the line $(x,1)$, $x\in\R\equiv \dv \Hp$.
We assume that the initial particle configuration converges weakly to a limiting density $\phi(x)$,
$0\le \phi\le 1$. The liquid region is given by a homeomorphism $W_{\LL}:\LL\to \Hp$, the
upper half plane, and we consider the extension of $W_{\LL}^{-1}$ to $\overline{\Hp}$. Part of 
$\dv \LL$ is given by a curve, the \emph{edge}  $\EE$,  parametrized by intervals in $\dv \Hp$,
and this corresponds to points where $\phi$ is identical to $0$ or $1$. If $0<\phi<1$, the non-trivial support,
there are two cases. Either $W_{\LL}^{-1}(w)$ has the limit $(x,1)$ as $w\to x$ non-tangentially and we have 
a \emph{regular point}, or we have what we call a \emph{singular point}. In this case $W_{\LL}^{-1}$ 
does not extend continuously to $\overline{\Hp}$. Singular points give rise to parts of $\dv \LL$
not given by $\EE$ and which can border a frozen region, or be ``inside'' the liquid region. This shows that
in general the boundary of $\dv \LL$ can be very complicated. We expect that on the singular parts of 
$\dv \LL$ we do not get a universal point process like the Airy or the extended Sine kernel point processes. Furthermore, $\EE$ and the singular parts of $\dv \LL$ are shocks of the \emph{complex Burgers equation}. 
 \end{abstract}
\tableofcontents
\clearpage
\section{Introduction}
\subsection{Discrete Interlacing Sequences}
We begin by briefly recalling the underlying probabilistic model described in \cite{Duse14a}.
A {\em discrete Gelfand-Tsetlin pattern} of depth $n$ is an $n$-tuple, denoted
$(y^{(1)},y^{(2)},\ldots,y^{(n)}) \in  \Z \times \Z^2 \times \cdots \times \Z^n$,
which satisfies the interlacing constraint
\begin{equation*}
y_1^{(r+1)} \; \ge \; y_1^{(r)} \; > \; y_2^{(r+1)} \; \ge \; y_2^{(r)}
\; > \cdots \ge \; y_r^{(r)} \; > \; y_{r+1}^{(r+1)},
\end{equation*}
for all $r \in \{1,\ldots,n-1\}$, denoted $y^{(r+1)} \succ y^{(r)}$. 
For each $n\ge1$, fix $x^{(n)} \in \Z^n$ with $x_1^{(n)} > x_2^{(n)} > \cdots > x_n^{(n)}$,
and consider the following probability measure on the set of patterns of depth $n$:
\begin{equation*}
\nu_n[(y^{(1)},\ldots,y^{(n)})]
:= \frac1{Z_n} \cdot \left\{
\begin{array}{rcl}
1 & ; &
\text{when} \; x^{(n)} = y^{(n)} \succ y^{(n-1)} \succ \cdots \succ y^{(1)}, \\
0 & ; & \text{otherwise},
\end{array}
\right.
\end{equation*}
where $Z_n > 0$ is a normalisation constant. This can equivalently
be considered as a measure on configurations of interlaced particles in
$\Z \times \{1,\ldots,n\}$ by placing a particle at position
$(u,r) \in \Z \times \{1,\ldots,n\}$ whenever $u$ is an element of $y^{(r)}$.
$\nu_n$ is then the uniform probability measure on the set of all such interlaced
configurations with the particles on the top row in the deterministic positions
defined by $x^{(n)}$. This measure also arises naturally from certain tiling
models (see \cite{Duse14a} and \cite{Pet12} for further details).
In \cite{Duse14a} and \cite{Pet12} it was independently shown that this process is
determinantal. The correlation kernel, $K_n : (\Z \times \{1,\ldots,n\})^2  \to \C$,
acts on pairs of particle positions. Note that the deterministic top row and the
interlacing constraint implies that it is sufficient to restrict to those positions,
$(u,r), (v,s) \in \Z \times \{1,\ldots,n-1\}$, with $u \ge x_n^{(n)}+n-r$ and
$v \ge x_n^{(n)}+n-s$. For all such $(u,r)$ and $(v,s)$,
\begin{equation}
\label{eqKnrusvFixTopLine}
K_n((u,r),(v,s)) = \widetilde{K}_n ((u,r),(v,s)) - \phi_{r,s}(u,v),
\end{equation}
where
\begin{equation*}
\widetilde{K}_n ((u,r),(v,s))
:=\frac{1}{(2\pi i)^2} \frac{(n-s)!}{(n-r-1)!} \oint_{\gamma_n}dw \oint_{\Gamma_n}dz
\frac{\prod_{k=u+r-n+1}^{u-1}(z-k)}{\prod_{\substack{k=v+s-n}}^{v}(w-k)}
\frac{1}{w-z}\prod_{i=1}^n\bigg(\frac{w-x_i^{(n)}}{z-x_i^{(n)}}\bigg),
\end{equation*}
and
\begin{equation*}
\phi_{r,s} (u,v)
:= 1_{(v \ge u)} \cdot \left\{
\begin{array}{lll}
0 & ; & \text{when} \; s \le r, \\
1 & ; & \text{when} \; s = r+1, \\
\frac1{(s-r-1)!} \prod_{j=1}^{s-r-1} (v-u+s-r-j) & ; & \text{when} \; s > r+1.
\end{array}
\right.
\end{equation*}
Above $\Gamma_n$ and $\gamma_n$ are counter-clockwise, $\Gamma_n$ contains
$\{x_i^{(n)} : x_i^{(n)} \ge u\}$ and none of $\{x_i^{(n)} \le u+r-n \}$,
and $\gamma_n$ contains $\Gamma_n$ and $\{v+s-n,...,v\}$.

\subsection{Asymptotic Assumptions and Geometric Behaviour of the Liquid Region}

It is natural to consider the asymptotic behaviour of the determinantal system
introduced in the previous section as $n \to \infty$, under the assumption
that the (rescaled) empirical distribution of the deterministic particles
on the top row converges weakly to a measure with compact support. More exactly,
assume that
\begin{equation*}
\frac1n \sum_{i=1}^n \delta_{x_i^{(n)}/n} \to \mu
\end{equation*}
as $n \to \infty$, in the sense of weak convergence of measures, where $\mu$
is a probability measure with compact support, $\supp(\mu)$. We
additionally assume that the convex hull of $\supp(\mu)$ is of length
strictly greater than $1$.

\begin{Def}
\label{remmu}
For clarity we explicitly state the class of measures in which $\mu$ lies:
$\mu\in\mathcal{B}(\R)$, where $\mathcal{B}(\R)$ is the set of Borel measures on $\R$. Moreover, $\mu\leq \lambda$ where $\lambda$ is Lebesgue measure
(recall $x^{(n)} \in \Z^n$), $\Vert \mu\Vert=1$, $\mu$ has compact support. We will denote this set of measures by
$\mu \in \mathcal{M}_{c,1}^{\lambda}(\R)$. Additionally we note that $\mu$
admits a density w.r.t. $\lambda$, which is uniquely defined up to a set of
zero Lebesgue measure. Denoting the density by $f$, and $[a,b]$ the convex hull of $\supp(\mu)$, ($b-a>1$), it satisfies
$f \in L^{\infty}(\R)$, $f(x) = 0$ for all $x \in \R \setminus [a,b]$, $\int_{\R}f(x)dx=1$,
and $0 \le f(x) \le 1$ for all $x \in [a,b]$. We write
$f \in \rho_{c,1}^{\lambda}(\R)$. Note that
$\R \setminus \supp(\mu)$ is the largest open set on which $f=0$ almost
everywhere, and $\R \setminus \supp(\lambda-\mu)$ is the largest open
set on which $f=1$ almost everywhere. Finally we note that the set $ \mathcal{M}_{c,1}^{\lambda}(\R)$ is convex, i.e., if $\sigma,\nu\in  \mathcal{M}_{c,1}^{\lambda}(\R)$, then for all $t\in[0,1]$, $t\sigma+(1-t)\nu\in \mathcal{M}_{c,1}^{\lambda}(\R)$.
\end{Def}

\begin{Def}
Define the set of functions $\Hc$ to be all $f\in \rho_{c,1}^{\lambda}(\R)$ such that:
\begin{itemize}
\item
There exists a finite family of open disjoint interval $\{I_k\}_k$ such that $\text{supp}(\mu)=\bigcup_{k}\overline{I}_k$.
\item
$f\in C^{\alpha}(\overline{I}_k)$ for all $k$ and some $0<\alpha<1$.
\item
The set $(\{t:f(t)=0\}\cup\{t:f(t)=1\})\bigcap(\cup_{k}\overline{I}_k)$ is isolated. 
\end{itemize}
We note that if $f\in \Hc$, then $f$ is continuous everywhere except at possibly the set $\bigcup_k \dv I_k$.
\end{Def}

Note, rescaling the vertical and horizontal positions of the particles
of the Gelfand-Tsetlin patterns by $\frac1n$, that the weak convergence
and the interlacing constraint imply that the rescaled particles almost
surely lie asymptotically in the the following set:
\begin{align*}
\mathcal{P}=\{(\chi,\eta)\in\R^2:a\leq \chi+\eta-1\leq \chi\leq b,0\leq\eta\leq 1\}
\end{align*} 
Fixing $(\chi,\eta) \in \mathcal{P}$, the local asymptotic behaviour of particles
near $(\chi,\eta)$ can be examined by considering the asymptotic behaviour of
$K_n((u_n,r_n),(v_n,s_n))$ as $n \to \infty$, where $\{(u_n,r_n)\}_{n\ge1} \subset \Z^2$
and $\{(v_n,s_n)\}_{n\ge1} \subset \Z^2$ satisfy
\begin{align*} 
\frac1n (u_n,r_n) &\to (\chi,\eta), \quad \frac1n (v_n,s_n) \to (\chi,\eta)
\end{align*}
as $n \to \infty$. Assume this additional asymptotic behaviour, substitute $(u_n,r_n)$ and
$(v_n,s_n)$ into equation (\ref{eqKnrusvFixTopLine}), and rescale the contours by $\frac1n$
to get,
\begin{equation}
\label{eqJnrnunsnvn1}
\widetilde{K}_n ((u_n,r_n),(v_n,s_n))
= \frac{A_n}{(2\pi i)^2} \oint_{\gamma_n} dw \oint_{\Gamma_n} dz \;
\frac{\exp(n f_n(w) - n \tilde{f}_n(z))}{w-z},
\end{equation}
for all $n \in \N$. Now $\Gamma_n$ contains $\{\frac1n x_i^{(n)} : x_i^{(n)} \ge u_n\}$
and none of $\{\frac1n x_i^{(n)} \le u_n+r_n-n \}$, and $\gamma_n$ contains $\Gamma_n$
and $\{\frac1n (v_n+s_n-n),...,\frac1n v_n\}$. Also
$A_n := \frac{(n-s_n)!}{(n-r_n-1)!} \; n^{s_n-r_n-1}$,
\begin{align*}
f_n(w)
& := \frac1n \sum_{i=1}^n \log \bigg( w - \frac{x_i^{(n)}}n \bigg) -
\frac1n \sum_{j=v_n+s_n-n}^{v_n} \log \left( w - \frac{j}n \right), \\
\tilde{f}_n(z)
& := \frac1n \sum_{i=1}^n \log \bigg( z - \frac{x_i^{(n)}}n \bigg) -
\frac1n \sum_{j=u_n+r_n-n+1}^{u_n-1} \log \left( z - \frac{j}n \right).
\end{align*}
Finally, inspired by the asymptotic assumptions and the forms of $f_n$
and $\tilde{f}_n$, we define
\begin{equation}
\label{AsymptoticFunction}
f_{(\chi,\eta)} (w) := \int_{\R}\log(w-t)d\mu(t)-\int_{\chi+\eta-1}^{\chi}\log(w-t)dt,
\end{equation}
for all $w \in \C \setminus \R$.
\begin{rem}
Do not confuse the asymptotic function $f_{(\chi,\eta)}(w)$ with the density $f$ of the the measure $\mu$. The authors apologize for this unfortunate notation and hope that it will not cause any confusion. Furthermore, the asymptotic function will only be mentioned in the introduction, and in all other sections of this paper, $f$ will always denote the density of the measure.
\end{rem}

Steepest descent analysis and equations (\ref{eqKnrusvFixTopLine}) and
(\ref{eqJnrnunsnvn1}) suggest that, as $n \to \infty$,
the asymptotic behaviour of $K_n ((u_n,r_n),(v_n,s_n))$ depends on the
behaviour of the roots of $f_{(\chi,\eta)}'$:
\begin{equation}
\label{eqf'}
f_{(\chi,\eta)}' (w)
= \int_{\R}\frac{d\mu(t)}{w-t}-\int_{\chi+\eta-1}^{\chi}\frac{dt}{w-t},
\end{equation}
for all $w \in \C \setminus \R$. In \cite{Duse14a}, we define the {\em liquid region},
$\LL$, as the set of all $(\chi,\eta) \in \mathcal{P}$ for which
$f_{(\chi,\eta)}'$ has a unique root in the upper-half plane,
$\Hp := \{w \in \C: \text{Im}(w) > 0 \}$. Whenever $(\chi,\eta) \in \LL$, one
expects universal bulk asymptotic behaviour, i.e., that the local
asymptotic behaviour of the particles near $(\chi,\eta)$ are governed by
the extended discrete \emph{Sine kernel} as $n\to +\infty$. Also, one expects
that the particles are not asymptotically densely packed. Moreover, when
considering the corresponding tiling model and its associated height function,
one would expect to see the Gaussian Free Field asymptotically.
See for example \cite{Pet12},\cite{Pet15} for a special case.

Let $W_\LL:\LL\to\Hp$ map $(\chi,\eta) \in \LL$ to the corresponding unique root
of $f_{(\chi,\eta)}'$ in $\Hp$. In \cite{Duse14a}, we show that $W_\LL$ is a
homeomorphism with inverse $W_\LL^{-1} (w) = (\chi_\LL(w), \eta_\LL(w))$ for all
$w \in \Hp$, where
\begin{align}
\label{InverseReal1}
\chi_\LL(w) &:= w + \frac{(w - \bar{w}) (e^{C(\bar{w})}-1)}{e^{C(w)} - e^{C(\bar{w})}}, \\
\label{InverseReal2}
\eta_\LL(w) &:=
1 + \frac{(w - \bar{w}) (e^{C(w)}-1) (e^{C(\bar{w})}-1) }{e^{C(w)} - e^{C(\bar{w})}},
\end{align}
and $C : \C \setminus \supp(\mu) \to \C$ is the {\em Cauchy} transform of $\mu$:
\begin{align}
\label{eqCauTrans}
C(w):=\int_{\R}\frac{d\mu(t)}{w-t}.
\end{align}
Thus $\LL$ is a non-empty, open (with respect to $\R^2$), simply connected subset of $\mathcal{P}$.
\newline \newline \indent
Define the complex slope $\Omega=\Omega(\chi,\eta)\in \C$ by 
\begin{align}
\label{ComplexSlope}
\Omega(\chi,\eta)=\frac{W_\LL(\chi,\eta)-\chi}{W_\LL(\chi,\eta)-\chi-\eta+1}.
\end{align}
The equation $f'{(\chi,\eta)}(w)\big\vert_{w=W_\LL(\chi,\eta)}=0$ implies that the \emph{complex slope} $\Omega$ satisfies the equation
\begin{align}
\label{EqComplexSlope}
\frac{1}{\Omega}=\exp\int_{\R}\bigg(\chi+\frac{(1-\eta)\Omega}{1-\Omega}-t\bigg)^{-1}d\mu(t).
\end{align}
Note that since
\begin{align}
\Omega=\exp\int_{\R}\frac{d\mu(t)}{t-W_\LL(\chi,\eta)}
\end{align}
and $W_\LL(\chi,\eta)\in \Hp$, it follows that Im$[\Omega]>0$ for all $(\chi,\eta)\in \LL$. Moreover, by differentiating (\ref{EqComplexSlope}) with respect to $\chi$ and $\eta$ respectively, one see that $\Omega$ satisfies the \emph{complex Burgers equation}
\begin{align}
\label{ComplexBurger}
\Omega\frac{\dv \Omega}{\dv \chi}=-(1-\Omega)\frac{\dv \Omega}{\dv \eta}.
\end{align}
For a connection to lozenge tiling problems see \cite{Ken07}.

Using the complex slope $\Omega$ one define the Beta kernel $\mathcal{B}_\Omega:\Z^2\to \C$, according to: 
\begin{align}
\mathcal{B}_\Omega(m,l)=\frac{1}{2\pi i}\int_{\overline{\Omega}}^{\Omega}(1-z)^{m}z^{-l-1}dz,
\end{align}
where the integration contours are such that they cross $(0,1)\subset \R$ when $m\geq 0$, and $(-\infty,0)\subset \R$ when $m<0$. 
It was shown in \cite{Pet12}, that if one let $\mu=\lambda\big\vert_{\cup_{k=1}^mI_k}$, where $I_k=[a_k,b_k]$, and $\cup_{k=1}^mI_k$ is a disjoint union of intervals, then if one assumes that 
\begin{align*}
\lim_{n\to \infty}\frac{1}{n}(x_i^{(n)},y_i^{(n)})=(\chi,\eta), \quad \text{for $i=1,2,..,r$}
\end{align*}
and,
\begin{align*}
x_i^{(n)}-x_j^{(n)}&=l_{ij}\in \Z \quad \text{and} \quad  y_i^{(n)}-y_j^{(n)}=m_{ij} \in\Z 
\end{align*}
are fixed whenever $n$ is sufficiently large, then
\begin{align*}
\lim_{n\to \infty}\rho_r((x_1^{(n)},y_1^{(n)}),(x_2^{(n)},y_2^{(n)}),...,(x_r^{(n)},y_r^{(n)})=\det[\mathcal{B}_\Omega(m_{ij},l_{ij})]_{i,j=1}^r
\end{align*}
Though it is not done in this paper, this result can be easily extended to the case when $\mu\in \mathcal{M}_{c,1}^{\lambda}(\R)$. In particular note that this implies that the macroscopic density of particles are given by
\begin{align*}
\rho(\chi,\eta)=\frac{1}{2\pi i}\int_{\overline{\Omega}}^{\Omega}\frac{dz}{z}=\frac{1}{\pi}\arg \Omega(\chi,\eta).
\end{align*}

In \cite{Duse14a}, we also study $\dv\LL$. Our motivation for doing this is
that edge-type behavior is expected at $\dv\LL$ for appropriate scaling limits.
It is therefore necessary to understand the geometry of $\dv\mathcal{L}$. We
study $\dv\LL$ using the above homeomorphism: $\dv\LL$ is the set of all
$(\chi,\eta) \in \mathcal{P}$ for which there exists a sequence,
$\{w_n\}_{n\ge1} \subset \mathbb{H}$, with
$W_\LL^{-1}(w_n) = (\chi_\LL(w_n),\eta_\LL(w_n)) \to (\chi,\eta)$ as $n \to \infty$,
and either $|w_n| \to \infty$ or $w_n \to x \in \R = \dv\Hp$ as $n \to \infty$.

The situation when $|w_n| \to \infty$ is trivial:
$(\chi_\LL(w_n),\eta_\LL(w_n)) \to (\frac12 + \int t d\mu(t), 0)$ as $n \to \infty$.
In order to consider the situation when $w_n \to x \in \R = \dv\Hp$, recall
that $\mu \le \lambda$. In \cite{Duse14a}, we consider the case where
$w_n \to x \in R$, where $R \subset \R$ is the open set,
\begin{equation}
\label{eqR}
R
:= R_{\mu} \cup R_{\lambda-\mu} \cup R_0\cup R_1 \cup R_2,
\end{equation}
and
\begin{itemize}
\item
$R_{\mu}:=\R \setminus \supp(\mu)\cap \{t\in \R:C(t)\neq 0\}$.
\item
$R_{\lambda-\mu}:=\R \setminus \supp(\lambda-\mu)$.
\item
$R_0:=\R \setminus \supp(\mu)\cap \{t\in \R:C(t)=0\}$
\item
$R_1$ is the set of all
$t \in \partial (\R \setminus \supp(\mu)) \cap \partial (\R \setminus \supp(\lambda-\mu))$
for which there exists an interval, $I := (t_2,t_1)$, with $t \in I$,
$(t,t_1) \subset \R \setminus \supp(\mu)$ and $(t_2,t) \subset \R \setminus \supp(\lambda-\mu)$.
\item
$R_2$ is the set of all
$t \in \partial (\R \setminus \supp(\mu)) \cap \partial (\R \setminus \supp(\lambda-\mu))$
for which there exists an interval, $I := (t_2,t_1)$, with $t \in I$,
$(t,t_1) \subset \R \setminus \supp(\lambda-\mu)$ and $(t_2,t) \subset \R \setminus \supp(\mu)$.
\end{itemize}
We show that $(\chi_\LL(w_n),\eta_\LL(w_n)) \to (\chi_\EE(t),\eta_\EE(t))$ as $n \to \infty$,
where $\chi_\EE, \eta_\EE : R \to \R$ are the real-analytic functions defined by,
\begin{align}
\label{eqchiEEetaEE}
(\chi_{\EE}(t),\eta_{\EE}(t))=\left\{
 \begin{array}{ll}
  \displaystyle\bigg( t + \frac{1-e^{-C(t)}}{C'(t)},1 + \frac{e^{C(t)}+e^{-C(t)}-2}{C'(t)}\bigg) & \text{if }  t \in R_{\mu}\cup R_0\\
    \displaystyle \bigg( t + \frac{1-(\frac{t-t_1}{t-t_2})e^{-C(t)} - 1}{C_I'(t) + \frac{1}{t-t_2} - \frac{1}{t-t_1}},1 + \frac{(\frac{t-t_2}{t-t_1})e^{C_I(t)}+(\frac{t-t_1}{t-t_2})e^{-C_I(t)} - 2}{C_I'(t) +  \frac{1}{t-t_2} - \frac{1}{t-t_1}}\bigg) & \text{if } t \in R_{\lambda-\mu}\\
    \displaystyle(t,1 - e^{C_I(t)} (t-t_2)) & \text{if } t\in R_1\\
 \displaystyle ( t - e^{-C_I(t)} (t-t_1),1 + e^{-C_I(t)} (t-t_1)) & \text{if } t\in R_2
 \end{array} \right.
\end{align}
Above $I := (t_2,t_1)$ is any interval which satisfies
$t \in I \subset \R \setminus \supp(\lambda-\mu)$ whenever
$t \in \R \setminus \supp(\lambda-\mu)$, and the requirements of
equation (\ref{eqR}) whenever $x \in R_1 \cup R_2$. Also, $C$ is
the Cauchy transform of equation (\ref{eqCauTrans}), and
$C_I(t) := \int_{\R \setminus I} \frac{d\mu(x)}{t-x}$ for all $t \in I$.
It follows from above that $(\chi_\EE(\cdot),\eta_\EE(\cdot)) : R \to \dv\LL$
is the unique continuous extension, to $R$, of
$(\chi_\LL(\cdot),\eta_\LL(\cdot)) : \Hp \to \LL$. In \cite{Duse14a}
we show that the extension is injective, and we define the {\em edge},
$\EE \subset \dv\LL$, as the image space of the extension. We argue that
$\EE$ is a natural subset of $\dv\LL$ on which to expect edge asymptotic
behaviour. This will be examined in the upcoming papers, \cite{Duse15c} and
\cite{Duse15d}. In these papers we will show, for example, as $n \to \infty$ and
choosing the parameters $(u_n,r_n)$ and $(v_n,s_n)$ appropriately, that
$K_n((u_n,r_n),(v_n,s_n))$ converges to the {\em Airy} or {\em Pearcey} kernel
when $x \in \R \setminus \supp(\mu)$ and
$(\chi,\eta) = (\chi_\EE(t),\eta_\EE(t))$. Similarly when
$t \in \R \setminus \supp(\lambda-\mu)$, except now the asymptotic behaviour
of the correlation kernel of the `holes' is examined. Thus $\EE$ is a subset of $\dv \LL$ where we expect standard, universal type edge behavior. Furthermore, in \cite{Duse14a}, we defined the sets $\EE_\mu=W_\EE^{-1}(R_{\mu})$,  $\EE_{\lambda-\mu}=W_\EE^{-1}(R_{\mu})$, $\EE_0=W_\EE^{-1}(R_{0})$, $\EE_1=W_\EE^{-1}(R_{1})$, and $\EE_2=W_\EE^{-1}(R_{2})$. One can show that for any sequence $\{(\chi_n,\eta_n)\}_n\subset \LL$, such that $\lim_{n\to \infty}(\chi_n,\eta_n)=(\chi_\EE,\eta_\EE)\in \EE$, the boundary value of the complex slope $\Omega$ exists and equals
\begin{align}
\label{bdEdgeComplexSlope}
\lim_{n\to\infty}\Omega(\chi_n,\eta_n)=\left\{
 \begin{array}{ll}
   e^{-C(t)}\in \R& \text{if }  (\chi_\EE,\eta_\EE)\in \EE_{\mu}\\
    \frac{t-t_2}{t-t_1}e^{-C_I(t)}\in\R& \text{if }  (\chi_\EE,\eta_\EE) \in \EE_{\lambda-\mu}\\
    1& \text{if }  (\chi_\EE,\eta_\EE)\in \EE_0\\
  0& \text{if }  (\chi_\EE,\eta_\EE)\in \EE_1\\
  \infty & \text{if } (\chi_\EE,\eta_\EE)\in \EE_2
 \end{array} \right.
\end{align} 
where $t=W_{\EE}(\chi_\EE,\eta_\EE)$, and where $\lim_{n\to \infty}\Omega(\chi_n,\eta_n)=\infty$ is viewed as a limit on the Riemann sphere $\C\cup \{\infty\}$. Hence, we may view $\EE$ as a \emph{shock} of the complex Burgers equation (\ref{ComplexBurger}).

\begin{rem}
\label{ConvRem}
In principle the convergence of $K_n((u_n,r_n),(v_n,s_n))$ could depend on how the empirical measure $\mu_n$ converges to $\mu$. However, such questions will be considered in an upcoming paper \cite{Duse15c}.
\end{rem}
\begin{rem}
\label{remR}
Note that $R_1 \cap R_2 = \emptyset$. Also $R_1 \cup R_2 = \partial (\R \setminus \supp(\mu))
\cap \partial (\R \setminus \supp(\lambda-\mu))$, the set of all common boundary points of
the disjoint open sets $\R \setminus \supp(\mu)$ and $\R \setminus \supp(\lambda-\mu)$.
Therefore we can alternatively write,
$R = (\; \overline{(\R \setminus \supp(\mu)) \cup (\R \setminus \supp(\lambda-\mu))} \;)^\circ$.
\end{rem}
Note that $R = \R = \dv\Hp$ in the special case when $\mu$ is Lebesgue measure
restricted to a finite number of disjoint intervals. This case was examined
by Petrov, \cite{Pet12}. For general $\mu$, however, $\R \setminus R$ is non-empty.
It therefore remains to consider sequences, $\{w_n\}_{n\ge1} \subset \mathbb{H}$,
with $w_n \to x \in \R \setminus R$ as $n \to \infty$. In \cite{Duse14a},
letting $f$ denotes the density of $\mu$ (see Definition \ref{remmu}), we
show that:
\begin{Lem}
\label{lemBddEdge2}
$(x,1) \in \partial \LL$ for $x \in \R \setminus R =
(\supp(\mu) \cap \supp(\lambda-\mu)) \setminus (R_1 \cup R_2)$
whenever there exists an $\epsilon>0$ for which one of the following
cases is satisfied:
\begin{enumerate}
\item
$\sup_{t \in (x-\epsilon,x+\epsilon)} f(t) < 1$ and
$\inf_{t \in (x-\epsilon,x+\epsilon)} f(t) > 0$.
\item
$\sup_{t \in (x-\epsilon,x)} f(t) < 1$, $\inf_{t \in (x-\epsilon,x)} f(t) > 0$
and $f(t) = 0$ for all $t \in (x,x+\epsilon)$.
\item
$\sup_{t \in (x-\epsilon,x)} f(t) < 1$, $\inf_{t \in (x-\epsilon,x)} f(t) > 0$
and $f(t) = 1$ for all $t \in (x,x+\epsilon)$.
\item
$\sup_{t \in (x,x+\epsilon)} f(t) < 1$, $\inf_{t \in (x,x+\epsilon)} f(t) > 0$
and $f(t) = 0$ for all $t \in (x-\epsilon,x)$.
\item
$\sup_{t \in (x,x+\epsilon)} f(t) < 1$, $\inf_{t \in (x,x+\epsilon)} f(t) > 0$
and $f(t) = 1$ for all $t \in (x-\epsilon,x)$.
\end{enumerate}
Moreover $(\chi_\LL(w_n),\eta_\LL(w_n)) \to (x,1)$ as $n \to \infty$ for all
$\{w_n\}_{n\ge1} \subset \mathbb{H}$ with $w_n \to x$.
\end{Lem}

Recall that for a general $f \in \rho_{c,1}^{\lambda}(\R)$ the assumptions of Lemma \ref{lemBddEdge2} need not be satisfied, and so the
above lemma gives an incomplete picture.
The main goal of this paper is to extend this result. In particular, we will
examine the novel and subtle geometric behaviour of $\dv\LL$ when the conditions
of the above lemma are violated. This analysis is surprisingly difficult, and
naturally leads to questions in harmonic analysis. Points in $\dv \LL\backslash \EE$ will be either of the form $(x,1)$, or be points where we expect to have non-standard, or non-universal "edge" behaviour for the correlation kernel. The detailed local asymptotics will not be investigated in the present paper.

\subsection{Introduction to The Geometry of $\dv\mathcal{L}\backslash \mathcal{E}$ and the Non-Trivial Support of $\mu$}

As explained in the previous section, fixing
$\mu \in \mathcal{M}_{c,1}^{\lambda}(\R)$ (see remark \ref{remmu})
and defining $\chi_\LL$ and $\eta_\LL$ as in equations
(\ref{InverseReal1}) and (\ref{InverseReal2}),
we wish to examine the boundary behaviour of the homeomorphism
$(\chi_\LL(\cdot),\eta_\LL(\cdot)) : \mathbb{H} \to \LL$ in the
neighbourhood of the following set:
\begin{Def}
Given $\mu \in \mathcal{M}_{c,1}^{\lambda}(\R)$, the \emph{non-trivial support} of $\mu$,
denoted $\Sn \subset \R$, is the complement of the open set defined
in equation (\ref{eqR}). More exactly,
\begin{align*}
\Sn
:= \text{supp}(\mu)\cap \text{supp}(\lambda-\mu) \setminus (R_1\cup R_2),
\end{align*}
where $\lambda$ is Lebesgue measure and $R_1 \cup R_2 = \partial (\R \setminus \supp(\mu))
\cap \partial (\R \setminus \supp(\lambda-\mu))$ (see remark \ref{remR}).
\end{Def}

Throughout the remainder of this paper we therefore make the following
assumptions:
\begin{hyp}
\label{hypConv1}
Fix $\mu \in \mathcal{M}_{c,1}^{\lambda}(\R)$ for which $\Sn^{\circ}$
is non-empty. 
\end{hyp}
\begin{rem}
Hypothesis \ref{hypConv1} excludes densities of the form $f(t)=\phi(t)\chi_{\mathfrak{K}}(t)$, where $\phi\in\rho_{c,1}^{\lambda}(\R)$, and $\mathfrak{K}$ is a measurable closed set such that $\mathfrak{K}^{\circ}=\varnothing$. Then $\Sn\subset \mathfrak{K}$. In particular, we will not consider examples of the form $f(t)=\chi_{\mathfrak{C}}(t)$, where $\mathfrak{C}$ is a fat Cantor set, that is a nowhere dense set such that $\lambda(\mathfrak{C})>0$.
\end{rem}
\begin{hyp}
\label{hypConv2}
Let $X:=\{t: 0<f(t)<1, d\mu(t)=f(t)dt \}$. Assume that for any open interval $I\subset \Sn^{\circ}$, $\lambda(X\bigcap I)>0$.
\end{hyp}
\begin{rem}
\label{hypRem}
This assumption is non-trivial. In \cite{Rudin}, it is shown that there exists a Borel set $A\subset [0,1]$ such that for any interval $I\subset [0,1]$ one has
\begin{align}
\label{Nonfull}
0<\lambda(A\bigcap I)<\lambda(I).
\end{align}
Taking $f\big\vert_{[0,1]}(t)=\chi_A(t)$, (\ref{Nonfull}) shows that $[0,1]\subset \Sn$. However, $\lambda(\{t:0<f(t)<1\}\bigcap[0,1])=0$. 
\end{rem}
Fix $x \in \Sn$ and a sequence
$\{w_n\}_{n\ge1} \subset \mathbb{H}$ with $w_n \to x$ as $n \to \infty$. Assuming these hypothesises, we wish to examine the behaviour of
$\{(\chi_\LL(w_n),\eta_\LL(w_n))\}_{n\ge1}$ as $n \to \infty$ for the
various possibilities of the point $x \in \Sn$ and the sequence
$\{w_n\}_{n\ge1} \subset \mathbb{H}$. More precisely, we introduce the following equivalence relation:
\begin{Def}
\label{defEquiRel}
To sequences $\omega_x=\{w_n\}_{n=1}$ and $\omega'_x=\{w'_m\}_{m=1}$ are said to be equivalent if the following holds:
\begin{itemize}
\item
$\lim_{n\to\infty}w_n=\lim_{k\to\infty}w'_m=x.$
\item
There exist $N>0$ and $M>0$, depending on $\omega_x$ and $\omega_x'$ such that $w_{N+n}=w'_{M+n}$ whenever $n>0$. 
\end{itemize}
This is easily seen to be an equivalence relation. We denote this by $\omega_x\sim \omega_x'$ and denote $[\omega]$ by its equivalence class. Furthermore, for each $x\in \R$, let $S_x$ denote the set of equivalence classes of sequences converging to $x$.  
\end{Def}
Now let 
\begin{align}
\dv \mathcal{L}_{\omega}(x)&:=\overline{\{(\chi_\LL(w_n),\eta_\LL(w_n)): n\geq1\}}\backslash\{(\chi_\LL(w_n),\eta_\LL(w_n)): n\geq 1\}\\&=\{(\chi',\eta')\in \mathcal{P}: \{w_{n_k}\}_k\subset \{w_n\}_{n=1},\lim_{k\to \infty}(\chi_\LL(w_{n_k}),\eta_\LL(w_{n_k})=(\chi',\eta')\}.
\end{align}
Then clearly $\dv \LL_{\omega}(x)=\dv \LL_{\omega'}(x)=\dv \LL_{[\omega]}(x)$ whenever $\omega\sim \omega'$. Finally let
\begin{align}
\dv \mathcal{L}(x)=\bigcup_{[\omega]\in S_x}\dv \LL_{[\omega]}(x).
\end{align}
We now note that by Lemma \ref{lemLBoundary} in the appendix, $\dv \LL=\dv \LL(\infty)\bigcup\big(\bigcup_{x\in \R}\dv \LL(x)\big)$.
In Lemma \ref{TangTop}, we show for every $x\in \Sn^{\circ}$
that we can always choose $\{w_n\}_{n\ge1}$ such that
$(\chi_\LL(w_n),\eta_\LL(w_n)) \to (x,1)$. In other words,
$\Sn^{\circ} \times \{1\} \subset \dv\LL$. We define the generic
case as that in which this limit is observed for arbitrary sequences:
\begin{Def}
\label{defGeneric}
$x \in\Sn$ is said to be {\em generic}
whenever $\dv \LL(x)=\{(x,1)\}$. In particular, this is equivalent to $(\chi_\LL(w_n),\eta_\LL(w_n)) \to (x,1)$ as $n \to \infty$ for arbitrary
sequences $\{w_n\}_{n\ge1}\subset \Hp$ converging to $x$. The set of generic points will be denoted by $\Sgen$.
\end{Def}
The homeomorphism, $(\chi_\LL(\cdot),\eta_\LL(\cdot)) : \mathbb{H} \to \LL$,
therefore has a unique continuous extension to $x \in \mathcal{S}_\mu$ whenever
$x$ is generic. Lemma \ref{lemBddEdge2}, above, gives sufficient conditions
for $x$ to be generic. We generalise these conditions in Proposition \ref{GenericProp}. Moreover we prove in Theorem \ref{GenericThm} that for a typical set $G\subset \Sn^{\circ}$, where $G$ is defined in Proposition \ref{Topology1}, $G\subset \Sgen$ is dense in $\Sn^{\circ}$.

We are particularly interested in those parts of $\dv\LL$ that arise from
non-generic points. Recall in the previous section, we defined
the edge, $\EE \subset \dv\LL$, by extending
$(\chi_\LL(\cdot),\eta_\LL(\cdot))$ uniquely and continuously to
$\R \setminus \Sn$. In particular $\EE=\bigcup_{x\in R}\dv\LL(x)$. Also the point
$\dv \LL(\infty)=(\frac12 + \int t d\mu(t), 0)$ is obtained by extending the homeomorphism
uniquely and continuously to `infinity'. Finally, as observed above,
$\Sn^{\circ} \times \{1\} \subset \dv\LL$.
We therefore define the {\em singular part of $\dv\LL$}, denoted
$\dv\mathcal{L}_{\text{sing}} \subset \dv\LL$, as:
\begin{align}
\label{SingBound}
\dv \mathcal{L}_{\text{sing}}
:= \dv\LL \setminus \bigg(\mathcal{E} \bigcup \bigg\{ \bigg( \frac12 + \int t d\mu(t), 0 \bigg) \bigg\}
\bigcup \bigg(\Sgen \times \{1\} \bigg) \bigg).
\end{align}
In view of Lemma \ref{lemLBoundary}, this leads to the natural decomposition of the boundary $\dv \LL$ according to
\begin{align}
\label{BoundDecomp}
\dv \mathcal{L}=\bigg\{\bigg(\frac12 + \int t d\mu(t), 0\bigg)\bigg\}\bigcup \EE\bigcup(\Sgen \times \{1\})\bigcup \dv \mathcal{L}_{\text{sing}}.
\end{align}
In particular we have 
\begin{align}
\label{SingBoundDecomp}
\dv \mathcal{L}_{\text{sing}}=\bigcup_{x\in \R\backslash (R\cup\Sgen)}\dv\LL(x).
\end{align}

We begin our analysis by expressing $((\chi_\LL(w),\eta_\LL(w))=((\chi_\LL(u,v),\eta_\LL(u,v))$ in real and imaginary parts of $C(w)$, where $w=u+iv$. Using that 
\begin{align}
\label{eqRn}
\text{Re}(C(w))= \int_\R \frac{(u-t) f(t) dt}{(u-t)^2 + v^2}:=\pi H_vf(u) \\
\label{eqIn}
- \text{Im}(C(w))= \int_\R \frac{v f(t) dt}{(u-t)^2 + v^2}=\pi P_v f(u),
\end{align}
equations (\ref{InverseReal1}) and (\ref{InverseReal2}) then become
\begin{align}
\label{eqchin}
\chi_{\mathcal{L}}(u,v)
&= u +v\frac{  e^{-\pi H_vf(u)}-\cos(\pi P_v f(u))}{\sin(\pi P_v f(u))}, \\
\label{eqetan}
\eta_{\mathcal{L}}(u,v)
&= 1 - v\frac{ e^{\pi H_vf(u)} +  e^{-\pi H_vf(u)}-2 \cos(\pi P_v f(u)))}{\sin(\pi P_v f(u)))}.
\end{align}
\begin{rem}
Recall that $ P_v f(u)$ is the Poisson kernel of $f$ and $H_vf(u)$ is the harmonic conjugate of $P_v f(u)$. Also note that by Lemma \ref{Poisson1}, $0<\pi P_v f(u)<\pi$ for all $(u,v)\in \Hp$. It is a well-known fact from harmonic analysis that 
\begin{align}
&\lim_{v\to 0^+}P_v f(u) = f(u) \quad \text{for a.e $u$}\\
&\lim_{v\to 0^+}H_v f(u) = \mathcal{H}f(u) \quad \text{for a.e $u$},
\end{align}
where $\mathcal{H}f$ denotes the Hilbert transform of $f$. In fact, the limits exist for every $u$ in the Lebesgue set of $f$ and the Lebesgue set of $\mathcal{H}f$ respectively.
\end{rem}

We now distinguish between different types of sequences that will be of use:
\begin{Def}
\label{defNonTang}
$\{w_n\}_{n\ge1}=\{u_n+iv_n\}_{n\ge1}$ is said to converge {\em non-tangentially} to $x$ whenever there exists
a constant $k>0$ for which $| \frac{u_n-x}{v_n} | < k$ for all $n$
sufficiently large and such that $\lim_{n\to \infty}w_n=x$. $\{w_n\}_{n\ge1}$ is said to converge {\em tangentially} to $x$ whenever
$| \frac{u_n-x}{v_n} | \to \infty$ as $n \to \infty$ and $\lim_{n\to \infty}w_n=x$.
\end{Def}

Note, we can alternatively define the above sequences by considering
the following {\em truncated cones}: For all $k > 0$ and $h > 0$,
define $\Gamma_k^h(x) \subset \Gamma_k(x) \subset \mathbb{H}$ by,
\begin{align*}
\Gamma_k^h(x) &:= \{(u,v) \in \mathbb{H}: 0<v<h \text{ and } |u-x| < kv \}, \\
\Gamma_k(x) &:= \{(u,v) \in \mathbb{H}: v>0 \text{ and } |u-x| < kv \}.
\end{align*}
These are shown in figure (\ref{Cone}). Note that $\{w_n\}_{n\ge1}$
converges non-tangentially to $x$ iff $w_n\to x$ and there exists a $k>0$ for which
$w_n \in \Gamma_k(x)$ for all $n$ sufficiently large. Also,
$\{w_n\}_{n\ge1}$ converges tangentially to $x$ iff $w_n\to x$ and there exists an $n(k)$ for which $w_n \not\in \Gamma_k(x)$ for all $n>n(k)$.

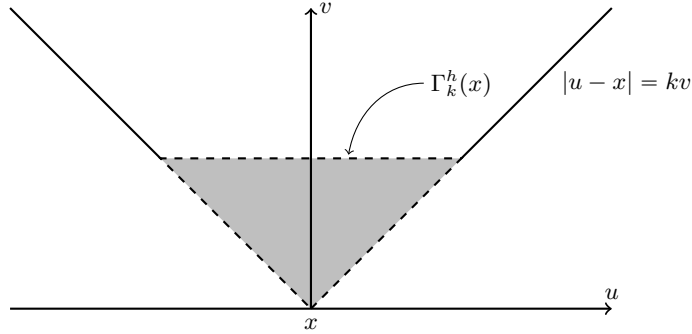
\begin{figure}[H]
\centering
\begin{tikzpicture}
\filldraw[fill=lightgray, draw=black,thick,dashed] (0,0) -- (2,2) -- (-2,2) -- (0,0);
\draw[thick,->] (0,0) -- (0,4);
\draw[thick,->] (-4,0) -- (4,0);
\draw[thick] (-2,2) -- (-4,4);
\draw[thick] (2,2) -- (4,4);
\draw (4.2,3) node {\small$\vert u-x \vert= kv$};
\draw (4,0.2) node {\small$ u$};
\draw (0.2,4) node {\small$v$};
\draw (0,-0.2) node {\small$x$};
\draw[<-] (0.5,2.05) to [out=80,in=180] (1.5,3);
\draw (2,3) node {\small$\Gamma_k^h(x)$};
\end{tikzpicture}

\caption{\emph{Truncated Cone}}
\label{Cone}
\end{figure}

Of course, arbitrary sequences $\{w_n\}_{n\ge1}\subset \Hp$ such that $\lim_{n\to \infty}w_n=x$
 are not-necessarily tangential nor non-tangential.
However the following result trivially follows from definitions
\ref{defGeneric} and \ref{defNonTang} by considering sub-sequences:
\begin{Lem}
\label{lemGeneric}
$x$ is generic if and only if both of the following occur:
\begin{itemize}
\item
$(\chi_{\mathcal{L}}(w_n),\eta_{\mathcal{L}}(w_n)) \to (x,1)$ as $n \to \infty$ whenever $\{w_n\}_{n\ge1}$ converges non-tangentially to $x$.
\item
$(\chi_{\mathcal{L}}(w_n),\eta_{\mathcal{L}}(w_n)) \to (x,1)$ as $n \to \infty$ whenever $\{w_n\}_{n\ge1}$  converges tangentially to $x$.
\end{itemize}
\end{Lem}

Generic situations are considered in section 4. We begin by
considering non-tangential sequences:
\begin{Def}
\label{DefReg}
$x \in \Sn$ is said to be {\em regular} if and only if
$(\chi_{\mathcal{L}}(w_n),\eta_{\mathcal{L}}(w_n)) \to (x,1)$ as $n \to \infty$ whenever $\{w_n\}_{n\ge1}$ converges non-tangentially to $x$. The set of regular points is denoted by $\Sreg$.
\end{Def}

In sections 2-4 we provide
sufficient conditions for a point to be regular. For example, in Proposition \ref{Generic1}
we show that $x \in \Sn$ is regular whenever $x$ belongs to the Lebesgue set of $f$ and $0<f(x)<1$.

Lemma \ref{lemGeneric} and Definition \ref{DefReg} imply that all generic points
are regular. The converse question, however, is non-trivial. In
Proposition \ref{GenericProp} we give sufficient conditions for a regular point to be
generic. In order to prove that the tangential limits converge correctly, we
assume a uniform convergence condition in a neighborhood of $x$. This condition
holds, for example, whenever the measure $\mu$ is such that $f \in \Hc$ (see for example
Proposition \ref{HoldeReg}).

In section 6 we consider non-generic situations:
\begin{Def}
\label{DefSing}
$x\in \Sn$ is said to be \emph{singular} if it is not regular, and the set of all singular points will be denoted by $\Ssing$.
We identify four classes singular points:
\begin{itemize}
\item
$x \in \SsingI$ if and only if there exists a $\delta>0$ and a function
$\varphi : \R \to \R$ for which $x$ is in the Lebesgue set of $\varphi$, $\varphi(x) = 0$,
$|\mathcal{H} \varphi (x)| <+\infty$, and $f(t) = \chi_{[x-\delta,x]}(t) + \varphi(t)$
for almost all $t$.
\item
$x \in \SsingII$ if and only if there exists a $\delta>0$ and a function
$\varphi : \R \to \R$ for which $x$ is in the Lebesgue set of $\varphi$, $\varphi(x) = 0$,
$|\mathcal{H} \varphi (x)| <+\infty$, and $f(t) = \chi_{[x,x+\delta]}(t) + \varphi(t)$
for almost all $t$.
\item
$x \in  \SsingIII$ if and only if $\int_{\R}\frac{f(t)dt}{(x-t)^2}<+\infty$
and $\mathcal{H}f(x)\neq 0$.
\item
$x \in  \SsingIV$ if and only if $\int_{\R}\frac{1-f(t)dt}{(x-t)^2}<+\infty$
and $\mathcal{H} (1-f)(x)\neq 0$.
\end{itemize}
\end{Def}
The fact that $x \in \Sn$ is singular whenever
$x \in \SsingI \cup \SsingII\cup  \SsingIII \cup  \SsingIV$
is shown in Propositions \ref{Sing1}-\ref{Sing2} and \ref{Sing3}. Indeed we show, whenever
$x \in \SsingI \cup \SsingII\cup  \SsingIII \cup  \SsingIV$
and $\{w_n\}_{n\ge1}$ is non-tangential, that $(\chi_{\mathcal{L}}(w_n),\eta_{\mathcal{L}}(w_n))$ converges to a 
point which is different from $(x,1)$. We give expressions for the position of this
in each of the 4 cases, noting in particular that the position is independent of the
choice of the constant $\delta$ whenever $x \in \SsingI \cup \SsingII$.
Also, it follows from the definition of $\Hc$ that
$\SsingI \cup \SsingII\cup \SsingIII \cup  \SsingIV$
is the set of all singular points whenever the measure $\mu$ is such that $f \in \Hc$.

\indent In particular the set  $\mathcal{S}_{nt}^{sing}(\mu)$ can be seen as an obstruction to extending the map $W_{\mathcal{L}}^{-1}(\mu):\Hp\rightarrow \mathcal{L}$ to a homeomorphism of the boundary. More precisely, in Theorem \ref{ExtendedMap} it is proven that $W_{\mathcal{L}}^{-1}(\mu):\Hp\rightarrow \mathcal{L}$ extends to a homeomorphism $\overline{W}_{\mathcal{L}}^{-1}(\mu):\overline{\Hp}\rightarrow \overline{\mathcal{L}}$ if $\mathcal{S}_{nt}^{sing}(\mu)=\varnothing.$ In particular, when $\mathcal{S}_{nt}^{sing}(\mu)\neq\varnothing$, then $\dv \LL$ is not homeomorphic to $S^1$.
Furthermore it will also be shown that these points need not be isolated. When considering the boundary behavior of the map $W_{\mathcal{L}}^{-1}$ for sequence $\{w_n\}_n\in \Hp$ such that $\lim_{n\to \infty}w_n=x\in \mathcal{S}_{nt}^{sing}(\mu)$ we will almost exclusively consider the case of isolated singular points. Furthermore, it will be shown that to study boundary behavior at such points one will be forced to consider particular classes of tangential sequences converging to $x$. More precisely, under an additional technical assumption on the density $f$, we prove in Propositions \ref{GeoSing1}-\ref{GeoSing4} and Proposition \ref{GeoSing2} and Theorems \ref{GeoThm1}-\ref{GeoThm2} that:
\begin{itemize}
\item
If $x \in \SsingI$ and there exists an $\eps>0$ such that $\int_{\R}\frac{\vert \varphi(t)\vert dt}{(y-t)^2}<\infty$ for all $y\in(x-\eps,x)\cup(x,x+\eps)$, then 
\begin{align*}
\dv\mathcal{L}(x)=\overline{\bigg\{\bigg(x,1-\frac{\delta e^{\pi \mathcal{H}\varphi(x)}}{1+\xi}\bigg):\xi\in(0,+\infty)\bigg\}}.
\end{align*}
In particular $x$ is isolated on the right from points in $\SsingIII$ and the left from points in $\SsingIV$.
\item
If $x \in \SsingII$ and there exists an $\eps>0$ such that $\int_{\R}\frac{\vert \varphi(t)\vert dt}{(y-t)^2}<\infty$ for all $y\in(x-\eps,x)\cup(x,x+\eps)$, then 
\begin{align*}
\dv\mathcal{L}(x)=\overline{\bigg\{\bigg(x+\frac{\delta e^{-\pi \mathcal{H}\varphi(x)}}{1+\xi},1-\frac{\delta e^{-\pi \mathcal{H}\varphi(x)}}{1+\xi}\bigg):\xi\in(0,+\infty)\bigg\}}.
\end{align*}
In particular $x$ is isolated on the right from points in $\SsingIV$ and the left from points in $\SsingIII$.
\item
If $x \in \SsingIII$ and there exists an $\eps>0$ such that $\int_{\R}\frac{f(t)dt}{(y-t)^2}<\infty$ for all $y\in(x-\eps,x)\cup(x,x+\eps)$, then 
\begin{align*}
\dv\mathcal{L}(x)=\overline{\bigg\{\bigg(x+\frac{1-e^{-\pi\mathcal{H}f(x)}}{\xi-\pi(\mathcal{H}f)'(x)},1-\frac{e^{\pi\mathcal{H}f(x)}+e^{-\pi\mathcal{H}f(x)}-2}{\xi-\pi(\mathcal{H}f)'(x)}\bigg):\xi\in(0,+\infty)\bigg\}}.
\end{align*}
In particular $x$ is isolated from other points in $\SsingIII$.
\item
If $x \in \SsingI$ and that there exists an $\eps>0$ such that $\int_{\R}\frac{(1-f(t)) dt}{(y-t)^2}<\infty$ for all $y\in(x-\eps,x)\cup(x,x+\eps)$, then 
\begin{align*}
\dv\mathcal{L}(x)=\overline{\bigg\{\bigg(x+\frac{1+e^{\pi\mathcal{H}(1-f)(x)}}{\xi-\pi(\mathcal{H}(1-f)'(x)},1-\frac{e^{\pi\mathcal{H}(1-f)(x)}+e^{-\pi\mathcal{H}(1-f)(x)}+2}{\xi-\pi(\mathcal{H}(1-f)'(x)}\bigg):\xi\in(0,+\infty)\bigg\}}.
\end{align*}
In particular $x$ is isolated from other points in $\SsingIV$.
\end{itemize}
Note that the geometry of $\dv\mathcal{L}(x)$ in these cases is entirely characterized by either $\mathcal{H}\varphi(x)$ or the numbers $\mathcal{H}f(x)$ and $(\mathcal{H}f)'(x)$. An additional reason why we choose to only consider those singular points which satisfied some additional criteria for isolatedness, is that we do not believe that the same type of simple characterization of $\dv\mathcal{L}(x)$ is possible in the case dense singular points, or in the case when the assumptions of Propositions \ref{GeoSing1}-\ref{GeoSing2} are violated.
Finally, if one is to apply Definition \ref{DefSing} to points $x\in R$ one would find that every $x$ in $R_{\mu}\cup R_{\lambda-\mu}\cup R_1\cup R_2$ were singular points. Therefore the case of considering boundary behavior of non-isolated boundary points of $\dv\Sn$ is similar to the case of non-isolated singular points. We will therefore restrict ourselves to consider only isolated points of $\dv\Sn$.
It will therefore prove useful to define the following subsets of $\dv \mathcal{S}_{nt}(\mu)$:
\begin{Def}
Let $\dv\mathcal{S}_{nt}^{iso}(\mu):=\dv \mathcal{S}_R^0\cup \dv\mathcal{S}_L^0\cup \dv\mathcal{S}_R^1\cup\dv\mathcal{S}_R^1$ be the isolated boundary points of $\dv \Sn$, where 
\begin{itemize}
\item
$\dv\mathcal{S}_R^0$ is the set of all $x\in\dv \mathcal{S}_{nt}(\mu)$ for which there exists intervals $I = (x,x+\delta)$ and $J=(x-\delta,x)$, such that $I\subset \supp(\mu)^c$ and $J\subset \Sn$ for some $\delta>0$.
\item
$\dv\mathcal{S}_L^0$ is the set of all $x\in\dv \mathcal{S}_{nt}(\mu)$ for which there exists intervals $I = (x-\delta,x)$ and $J=(x,x+\delta)$, such that $I\subset \supp(\mu)^c$ and $J\subset \Sn$ for some $\delta>0$.
\item
$\dv\mathcal{S}_R^1$ is the set of all $x\in\dv \mathcal{S}_{nt}(\mu)$ for which there exists intervals, $I = (x,x+\delta)$ and $J=(x-\delta,x)$, such that $I\subset \supp(\lambda-\mu)^c$ and $J\subset \Sn$ for some $\delta>0$.
\item
$\dv\mathcal{S}_L^1$ is the set of all $x\in\dv \mathcal{S}_{nt}(\mu)$ for which there exists interval, $I = (x-\delta,x_1)$ and $J=(x,x+\delta)$, such that $I\subset \supp(\lambda-\mu)^c$ and $J\subset \Sn$ for some $\delta>0$.
\end{itemize}
\end{Def}
We notice that if the density $f\in \Hc$ and $x\in \Sn$ is a singular point, then it is isolated. Furthermore, $\dv \mathcal{S}_{nt}(\mu)=\dv\mathcal{S}_{nt}^{iso}(\mu)$, in this case. In Proposition \ref{Edge1} and Proposition \ref{Edge2} we provide sufficient conditions on the density $f$ for when $x\in \dv\mathcal{S}_{nt}^{iso}(\mu)$ is generic. In particular we show that $\lim_{\substack{t\to x\\ t\in R}}(\chi_{\EE}(t),\eta_{\EE}(t))=(x,1)$. If on the other hand $x\in  \SsingI \cup \SsingII\cup  \SsingIII \cup  \SsingIV$ and $(\chi^*,\eta^*)=\lim_{n\to \infty}((\chi_{\EE}(w_n),\eta_{\EE}(w_n))$, where $\{w_n\}_n$ converges non-tangentially to $x$, then $\lim_{\substack{t\to x\\ t\in R}}(\chi_{\EE}(t),\eta_{\EE}(t))=(\chi^*,\eta^*)$.
\newline \newline \indent
We now consider the boundary behaviour of the complex slope $\Omega$ for sequences $\{(\chi_n,\eta_n)\}_n\subset \LL$ such that $\lim_{n\to \infty}(\chi_n,\eta_n)=(\chi,\eta)\in \dv \LL\backslash \R$. First consider the case when $(\chi,\eta)\in \{(x,1): x\in \Sgen\}$. One can show that almost all \emph{non-tangential} limits exists and 
\begin{align*}
\lim_{n\to \infty}\Omega(\chi_n,\eta_n)=e^{-\pi\mathcal{H}f(x)+i\pi f(x)}\in \C.
\end{align*}
Thus, such limit is thus not in general real, which should be contrasted to the case when $(\chi,\eta)\in \EE$. On the other hand, if we assume that $x\in \Ssing$, and that in addition $x$ is an isolated singular point, then for \emph{all} sequences $\{(\chi_n,\eta_n)\}_n\subset \LL$ such that $\lim_{n\to \infty}(\chi_n,\eta_n)=(\chi,\eta)\in \dv \LL(x)$ we get
\begin{align}
\label{bdComplexSlope}
\lim_{n\to\infty}\Omega(\chi_n,\eta_n)=\left\{
 \begin{array}{ll}
   e^{-\pi\mathcal{H}f(x)}\in \R& \text{if }  x \in \SsingIII\\
    -e^{\pi\mathcal{H}(1-f)(x)}\in\R& \text{if } x\in \SsingIV\\
    0& \text{if }  x\in \SsingI\\
  \infty& \text{if }  x\in \SsingII \\
 \end{array} \right.
\end{align} 
This shows that at least a subset of $\dv \LL_{sing}$ are shocks of the complex Burgers equation in the same way as $\EE$.
\newline \newline \indent We will conclude this introduction by discussing open problems not solved in this paper. 
\begin{Con}
\label{Con1}
$\Sreg=\Sgen.$
\end{Con}
\begin{Con}
\label{Con2}
$\Ssing$ is meagre set in $\R$.
\end{Con}
However, note that $\Ssing$ is not necessarily negligible from a measure theoretic point of view. This is proven in Lemma \ref{ExSingset1}, where we show that there exists a $\mu\in \mathcal{M}_{c,1}^{\lambda}(\R)$ such that $\lambda(\Ssing)>0$.  Moreover, in Lemma \ref{ExSingset2}, we show that the set $\Ssing$ may be dense in $\Sn^{\circ}$. Finally, in Proposition \ref{propHausdorff}, we show that there exits a $\mu\in \mathcal{M}_{c,1}^{\lambda}(\R)$ such that $\mathcal{H}^1(\dv \LL)=+\infty$, where $\mathcal{H}^1$ denotes the one dimensional Hausdorff measure.

\begin{ex}
Consider the density $f(t)=\frac{15}{16}(t+1)^2(t-1)^2\chi_{[-1,1]}(t)$. Here $\supp(\mu)=[-1,1]$ and $\Sreg=\Sgen=(-1,1)$ and $\Ssing=\{-1\}\cup\{1\}$. The boundary of the \emph{liquid} region is shown in figure \ref{figLiquidregion}.
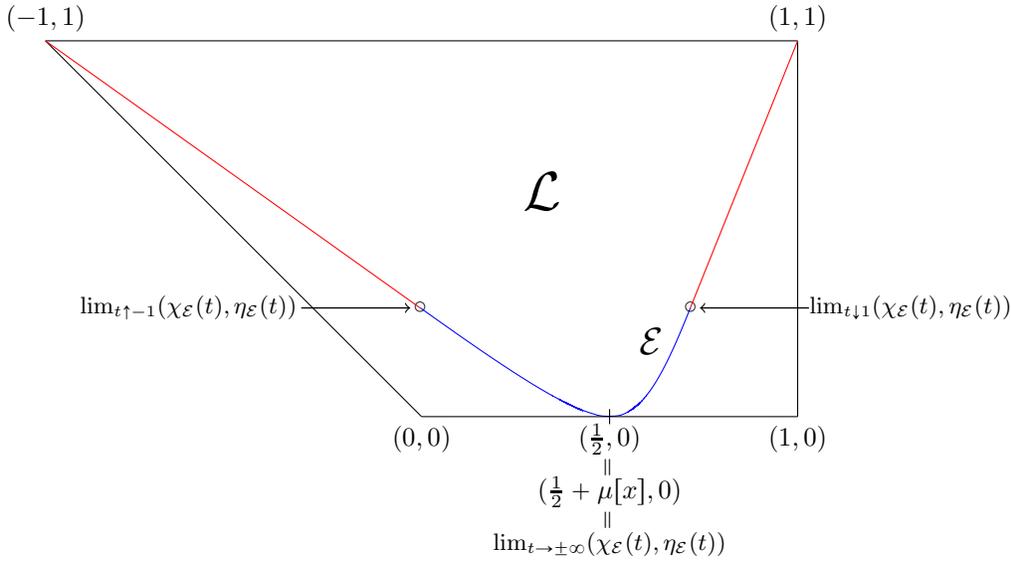
\begin{figure}[H]
\begin{tikzpicture}
\begin{axis}[hide axis, xmin=-1,xmax=1, ymin=0,ymax=1, samples=100, smooth, x=5cm, y=5cm]
\addplot [domain=1:1.5, blue, thin]
({x + (exp((15/16)*((10*x/3) - 2*(x^3) + ((x^2-1)^2)*(ln(abs((x+1)/(x-1))))))-1)/((exp((15/16)*((10*x/3) - 2*(x^3) + ((x^2-1)^2)*(ln(abs((x+1)/(x-1)))))))*((15/16)*((16/3) - 8*(x^2) + 4*(x^3-x)*(ln(abs((x+1)/(x-1)))))))},
{1 +((exp((15/16)*((10*x/3) - 2*(x^3) + ((x^2-1)^2)*(ln(abs((x+1)/(x-1))))))-1)^2)/((exp((15/16)*((10*x/3) - 2*(x^3) + ((x^2-1)^2)*(ln(abs((x+1)/(x-1)))))))*((15/16)*((16/3) - 8*(x^2) + 4*(x^3-x)*(ln(abs((x+1)/(x-1)))))))});
\addplot [domain=-1.5:-1, blue, thin]
({x + (exp((15/16)*((10*x/3) - 2*(x^3) + ((x^2-1)^2)*(ln(abs((x+1)/(x-1))))))-1)/((exp((15/16)*((10*x/3) - 2*(x^3) + ((x^2-1)^2)*(ln(abs((x+1)/(x-1)))))))*((15/16)*((16/3) - 8*(x^2) + 4*(x^3-x)*(ln(abs((x+1)/(x-1)))))))},
{1 +((exp((15/16)*((10*x/3) - 2*(x^3) + ((x^2-1)^2)*(ln(abs((x+1)/(x-1))))))-1)^2)/((exp((15/16)*((10*x/3) - 2*(x^3) + ((x^2-1)^2)*(ln(abs((x+1)/(x-1)))))))*((15/16)*((16/3) - 8*(x^2) + 4*(x^3-x)*(ln(abs((x+1)/(x-1)))))))});
\end{axis}

\begin{axis}[hide axis, xmin=-1,xmax=1, ymin=0,ymax=1, samples=100, smooth, x=5cm, y=5cm]
\addplot [domain=1.5:2.3, blue, thin]
({x + (exp((15/16)*((10*x/3) - 2*(x^3) + ((x^2-1)^2)*(ln(abs((x+1)/(x-1))))))-1)/((exp((15/16)*((10*x/3) - 2*(x^3) + ((x^2-1)^2)*(ln(abs((x+1)/(x-1)))))))*((15/16)*((16/3) - 8*(x^2) + 4*(x^3-x)*(ln(abs((x+1)/(x-1)))))))},
{1 +((exp((15/16)*((10*x/3) - 2*(x^3) + ((x^2-1)^2)*(ln(abs((x+1)/(x-1))))))-1)^2)/((exp((15/16)*((10*x/3) - 2*(x^3) + ((x^2-1)^2)*(ln(abs((x+1)/(x-1)))))))*((15/16)*((16/3) - 8*(x^2) + 4*(x^3-x)*(ln(abs((x+1)/(x-1)))))))});
\addplot [domain=-1.9:-1.5, blue, thin]
({x + (exp((15/16)*((10*x/3) - 2*(x^3) + ((x^2-1)^2)*(ln(abs((x+1)/(x-1))))))-1)/((exp((15/16)*((10*x/3) - 2*(x^3) + ((x^2-1)^2)*(ln(abs((x+1)/(x-1)))))))*((15/16)*((16/3) - 8*(x^2) + 4*(x^3-x)*(ln(abs((x+1)/(x-1)))))))},
{1 +((exp((15/16)*((10*x/3) - 2*(x^3) + ((x^2-1)^2)*(ln(abs((x+1)/(x-1))))))-1)^2)/((exp((15/16)*((10*x/3) - 2*(x^3) + ((x^2-1)^2)*(ln(abs((x+1)/(x-1)))))))*((15/16)*((16/3) - 8*(x^2) + 4*(x^3-x)*(ln(abs((x+1)/(x-1)))))))});
\end{axis}

\draw [blue] plot [smooth] coordinates {(7.75359872475353,0.0658251069572175) (7.61535036044783,0.0122876337932226) (7.37236200578712,0.0122876337874533) (7.08205590554426,0.104539016530695)};

\draw (5,0) --++ (5,0);
\draw (10,0) --++ (0,5);
\draw (10,5) --++ (-10,0);
\draw (0,5) --++ (5,-5);

\draw (5,-.3) node {$(0,0)$};
\draw (7.5,-.3) node {$(\frac12,0)$};
\draw (7.5,.1) --++ (0,-.2);
\draw (7.5,-.7) node {\rotatebox{90}{$\,=$}};
\draw (7.5,-1) node {$(\frac12 + \mu[x],0)$};
\draw (7.5,-1.4) node {\rotatebox{90}{$\,=$}};
\draw (7.5,-1.7) node {\small $\lim_{t \to \pm \infty} (\chi_\EE(t),\eta_\EE(t))$};
\draw (10,-.3) node {$(1,0)$};
\draw (10,5.3) node {$(1,1)$};
\draw (0,5.3) node {$(-1,1)$};

\draw (8.57300959372038,1.44630449135594) node {$\circ$};
\draw (11.5,1.44630449135594) node {\small $\lim_{t \downarrow 1} (\chi_\EE(t),\eta_\EE(t))$};
\draw[arrows=->,line width=.5pt](10.15,1.44630449135594)--(8.7,1.44630449135594);
\draw (4.98068591492368,1.44630449135594) node {$\circ$};
\draw (1.9,1.44630449135594) node {\small $\lim_{t \uparrow -1} (\chi_\EE(t),\eta_\EE(t))$};
\draw[arrows=->,line width=.5pt](3.4,1.44630449135594)--(4.85,1.44630449135594);
\draw[red] (8.57300959372038,1.44630449135594) --++ (1.42699040627962,3.55369550864406);
\draw[red] (4.98068591492368,1.44630449135594) --++ (-4.98068591492368,3.55369550864406);

\draw (6.6,3) node {\huge $\LL$};
\draw (8.05,1) node {\Large $\EE$};
\end{tikzpicture}
\caption{\label{figLiquidregion}. The boundary of the \emph{liquid} region. The blue curve is the \emph{edge} $\EE$. Here we expect to see the Airy process. The remaining part of the boundary is the red lines and the top line $\{(x,1):-1\leq x\leq 1\}$. Here we do not expect to see any universal edge fluctuations.}
\end{figure}

\end{ex}

\textbf{Acknowledgements:} This research was carried out at the Royal Institute
of Technology (KTH), Stockholm and Uppsala University, and was partially supported by grant
KAW 2010.0063 from the Knut and Alice Wallenberg Foundation. The authors would like to thank Kurt Johansson for helpful discussions and useful suggestions. Finally, we are indebted to Samuel Holmin for the proof idea of Lemma \ref{ExSingset2} and for pointing out the reference \cite{Rudin} in relation to Hypothesis \ref{hypConv2}. 

\section{Preliminaries}
\subsection{Integral Means and the Boundary Behavior of $e^{H_{v_n}f(u_n)}$ and $P_{v_n} f(u_n)$}
When studying the asymptotic behaviour of $(\chi_{\LL}(w_n),\eta_{\LL}(w_n))$ for non-tangential sequences $\{w_n\}_n$ such that $\lim_{n\to \infty}w_n=x$, it is natural to first try to estimate $e^{H_{v_n}f(u_n)}$ and $P_{v_n}f(u_n)$ separately. We will not attempt to classify all possible situations for which a point is regular, but contend ourselves with providing sufficient conditions which cover many interesting cases.
In particular, we provide sufficient conditions for $v_ne^{\pi \vert H_{v_n}f(u_n)\vert}\rightarrow 0$ for a non-tangential sequence $u_n+iv_n\rightarrow x\in \Sn$ as $n\rightarrow +\infty$. To achieve this it will be natural to consider certain means of the function $f$.

Recall that the Lebesgue set $\mathscr{L}_f$ of an $L_{loc}^{1}(\R)$ function $f$ is the set of all $x\in\R$ such that
\begin{align}
\label{LPoint}
\lim_{h\to0^+}\frac{1}{2h}\int_{x-h}^{x+h}\vert f(t)-f(x)\vert dt=0.
\end{align}
It is a well known result that the set of points which fails to be Lebesgue points has Lebesgue measure zero, see \cite{SteinW1}. If $x$ does not belong to the Lebesgue set of $f$ one may try to redefine the value of $f(x)$ at $x$ such that (\ref{LPoint}) holds. If this is not possible then $x$ does not belong to the Lebesgue set of $f$ for any $f\in[f]\in L^1(\R)$, where $[f]$ denotes the equivalence class of $f$ in $L^1(\R)$. 
In particular we note that if $f\in \rho_{c,1}^{\lambda}(\R)$ and (\ref{LPoint}) holds then 
\begin{align*}
\lim_{h\rightarrow 0^+} \frac{1}{2h}\int_{x-h}^{x+h}f(t)dt=\lim_{h\rightarrow 0^+} \frac{1}{h}\int_{x}^{x+h}f(t)dt=\lim_{h\rightarrow 0^+} \frac{1}{h}\int_{x-h}^{x}f(t)dt=f(x).
\end{align*}
Of course the converse of this is not true in general. However, if $f(x)=0$ or $f(x)=1$ then in the first case we have
\begin{align*}
\lim_{h\rightarrow 0^+} \frac{1}{2h}\int_{x-h}^{x+h}\vert f(t)-f(x)\vert dt=\lim_{h\rightarrow 0^+} \frac{1}{2h}\int_{x-h}^{x+h}f(t)dt=0
\end{align*}
or in the second case 
\begin{align*}
\lim_{h\rightarrow 0^+} \frac{1}{2h}\int_{x-h}^{x+h}\vert f(t)-f(x)\vert dt= \lim_{h\rightarrow 0^+} \frac{1}{2h}\int_{x-h}^{x+h}dt-\frac{1}{2h}\int_{x-h}^{x+h}f(t)dt=0.
\end{align*}
Therefore let
\begin{align}
\label{Density}
f(x)=F'(x):=\lim_{h\to 0^+}\frac{F(x+h)-F(x-h)}{2h},
\end{align}
where
\begin{align*}
F(x)=\int_{-\infty}^{x}d\mu(t),
\end{align*}
and note that the limit (\ref{Density}) exists for almost every $x$, in particular for every $x$ in the Lebesgue set of $f$.  Functions $f\in L^1(\R)$ defined through (\ref{Density}) are said to be \emph{strictly defined}, (see page 192 in \cite{Stein}). We will therefore always assume that the density $f$ in the equivalence class of densities of the measure $\mu$ is defined by (\ref{Density}), and the Lebesgue set of $f$ will always be with respect to this density. Moreover, it will be important to study not only the properties of the density $f$ but also of its Hilbert transform $\mathcal{H}f$, where
\begin{align*}
\mathcal{H}f(x):=\lim_{\eps\to 0^+}\frac{1}{\pi}\int_{\vert x-t\vert>\eps}\frac{f(t)dt}{x-t},
\end{align*}
and where this limit exists for almost every $x$. It is a well-known fact in the theory of singular integrals that the Hilbert transform is a bounded operator on $L^p(\R)$ for every $1<p<\infty$, see for example Theorem 4.1.7 in \cite{Garf}. Since $f\in \rho_{c,1}^{\lambda}(\R)$ it follows that $f\in L^p(\R)$ for every $1\leq p\leq \infty$, and hence that $\mathcal{H}f\in L^{p}(\R)$ for every $1<p<\infty$,  \newline \newline As was remarked before, we will be interested in considering non-tangential limits. That is, if $u:\Hp\to \R$, we will say that $u$ has a non-tangential limit $l$ at $x_0\in \R=\dv\Hp$, if for each $\alpha>0$,
\begin{align*}
\lim_{\substack{(x,y)\to(x_0,0)\\ (x,y)\in\Gamma_{\alpha}(x_0)}}u(x,y)=l.
\end{align*}
Similarly, we will say that a function $u:\Hp\to \R$ is \emph{non-tangentially bounded} at $x_0$ if for every $\alpha>0$ we have that
\begin{align*}
\sup_{(x,y)\in\Gamma_{\alpha}^1(x_0)}\vert u(x,y)\vert <\infty.
\end{align*}
For many estimates it will prove useful to introduce the Hardy-Littlewood maximal function $m_f$, defined at $x\in \R$ for $f\in L^p(\R)$ for $1\leq p\leq \infty$ by
\begin{align*}
m_f(x):=\sup_{h>0}\frac{1}{2h}\int_{x-h}^{x+h}\vert f(t)\vert dt.
\end{align*}
Recall that, $P_v f(u)$ is the Poisson integral of the density $f$. However, by Lemma 1.5 in chapter VI in \cite{Stein}, $H_vf(u)=P_v (\mathcal{H}f)(u)$, that is  $H_vf(u)$ is the Poisson integral of the Hilbert transform of $f$. Now Theorem 3.16 in chapter II of \cite{Stein} implies that $\lim_{n\to\infty}P_{v_n}f(u_n)=x$ for non-tangential limits at each $x\in\mathscr{L}_f$, thus in particular almost everywhere. Similarly, $H_vf(u)$ has the non-tangential limit $\Hil f(x)$ at every $x\in \mathscr{L}_{\Hil f}$, thus in particular, almost everywhere.  Finally, Theorem 1.4 in chapter VI of \cite{Stein} shows that $m_{\mathcal{H}f}$ dominates $H_vf$ in the following sense:
\begin{align}
\label{nontangMax}
\sup_{(u,v)\in\Gamma_{\alpha}^1(x)}\vert H_vf(u)\vert\leq d_{\alpha}m_{\Hil f}(x), 
\end{align}
where the constant $d_{\alpha}$ does not depend on $x$. Moreover, Lemma 1.2 in chapter VI in \cite{Stein} states that
\begin{align}
\label{Hconv}
\lim_{v\to 0^+}\bigg\{ H_vf(x)-\int_{0<v\leq \vert t\vert}\frac{f(x-t)dt}{t}\bigg\}=0
\end{align}
at each point $x$ in $\mathscr{L}_f$. 
\begin{rem}
\label{R1}
Note that in Lemma 1.2 in \cite{Stein}, $\mathscr{L}_f$ is the Lebesgue set of $f$ and not of $\Hil f$. It should be noted that (\ref{Hconv}) does not apply for arbitrary non-tangential limits, as can be seen by considering the function $f(t)=(\log(\vert t\vert^{-1}))^{-1}\chi_{[-a,a]}(t)$ at $0$, for some $a>0$. However, if $f$ satisfies the following Dini-type condition:
\begin{align}
\label{Dini1}
\int_{x-1}^{x+1}\frac{\vert f(x)-f(t)\vert dt}{\vert x-t\vert}<+\infty,
\end{align}
then for all non-tangential limits $\{u_n+iv_n\}_n$ that converge to $x$, $\lim_{n\to \infty}H_{v_n}f(u_n)=\Hil f(x)$. For a proof of this fact see Proposition \ref{DiniConv} in the appendix.
\end{rem}

So far we have not used the fact that $f\in L^{\infty}(\R)$ and its consequences for its Hilbert transform $\Hil f$. However, the fact that $f\in L^{\infty}(\R)$ implies that $\mathcal{H}f\in BMO$, where $BMO$, denotes the class of functions of bounded mean oscillation. A function $f\in BMO$ if 
\begin{align}
\label{Bmo}
\sup_{\substack{x\in \R \\ h>0}}\frac{1}{2h}\int_{x-h}^{x+h}\vert f(t)-Mf(x,h)\vert dt<+\infty,
\end{align}
where 
\begin{align*}
Mf(x,h):=\frac{1}{2h}\int_{x-h}^{x+h} f(t)dt.
\end{align*}
The left hand side of (\ref{Bmo}) is the $BMO$ norm of $f$ and is denoted by $\Vert f \Vert_{BMO}$. In particular it follows that if $f\in BMO$, then $m_f\in BMO$, see Theorem 4.2 (b) in \cite{Sharp}. Moreover, functions of bounded mean oscillation are $L^p_{loc}(\R)$ for every $0<p<\infty$. Finally, functions $f\in BMO$ satisfies the John-Nirenberg inequality:
\begin{align}
\label{JohnNirenberg}
\vert \{t\in [x-h,x+h]:\vert f(t)-Mf(x,h)\vert>\alpha\}\vert\leq c_1\exp\bigg(-c_2\frac{\alpha}{\Vert f \Vert_{BMO}}\bigg)2h
\end{align}
for some positive constants $c_1,c_2$ independent of $x$. For more details see for example \cite{SteinW2} or \cite{Garnett}. \newline \newline
If $f\in \Hc$, then for every $x,y\in \overline{I}_k$, and every $k$, there exists a constant $C$, such that $\vert f(x)-f(y)\vert \leq C\vert x-y\vert^{\alpha}$. This implies, (see \cite{Musk}), that for every $x,y\in I_k$, there exists a constant $c$, that depends on $x$, such that $\vert \mathcal{H}f(x)-\mathcal{H}f(y)\vert \leq c\vert x-y\vert^{\alpha}$. Thus, in particular $\mathcal{H}f\in C(I_k)$, for every $k$. Note however that $\mathcal{H}f$ need not be continuous on the set $\bigcup_k\dv I_k$.
\newline \newline
For $f\in L^1(\R)$ and $y>0$ let 
\begin{align*}
M_Rf(x,y)&:=\frac{1}{y}\int_{x}^{x+y}f(t)dt\\
M_Lf(x,y)&:=\frac{1}{y}\int_{x-y}^{x}f(t)dt\\
\Delta Mf(x,y)&:=\frac{1}{y}\int_{x}^{x+y}f(t)dt-\frac{1}{y}\int_{x-y}^{x}f(t)dt\\
\Delta m_f(x)&:=\sup_{y>0}\vert \Delta Mf(x,y)\vert\\
m_f^{\delta}(x)&:=\sup_{0<y<\delta} \vert Mf(x,y)\vert
\end{align*}
for $x\in \R$ and $y\in\R^+$. It follows from the fact that $0\leq f(t)\leq 1$, that $0\leq M_Rf(x,y)\leq 1$, $0\leq M_Lf(x,y)\leq 1$ and  $-1\leq \Delta Mf(x,y)\leq 1$  for all $(x,y)\in \Hp$. In particular $\Delta m_f(x)$ is a maximal function for the cancellation of the right sided and left sided means, and $m_f^{\delta}(x)$ is a truncated maximal function. As will be shown in Lemma \ref{EstHilb1}, it is the size of $\Delta Mf(u_n,v_n)$ that controls the growth rate of the function $\pi \vert H_{v_n}f(u_n)\vert $ for non-tangential sequences $u_n+iv_n\in \Hp$ as $u_n+iv_n\rightarrow x\in \R$ as $n\rightarrow +\infty.$  In particular, we have the following important Lemma:
\begin{Lem}
\label{DeltaMean}
Assume that $x\in\Sn^{\circ}$ and that $f\in \rho_{c,1}^{\lambda}(\R)$. Then
\begin{align}
\label{M1}
\vert \Delta Mf(x,y)\vert <1,
\end{align}
and
\begin{align}
\label{M2}
\vert \Delta Mf(x,y)\vert\leq \frac{1}{y}
\end{align}
for all $y>0$.
\end{Lem}
\begin{proof}
Assume the contrary. Then there exists a $y^*>0$ such that $\vert \Delta Mf(x,y^*)\vert=1$. It is clear from the definition of $M_Rf$  and $M_Lf$ that either $M_Rf(x,y^*)=1$ and $M_Lf(x,y^*)=0$ or that $M_Lf(x,y^*)=1$ and $M_Rf(x,y^*)=0$. In the first case this implies that $f(t)=1$ for a.e. $t\in[x,x+y]$ and that $f(t)=0$ for a.e. $t\in[x-y,x]$, This implies that $(x,x+y)\subset \R\backslash \text{supp}(\lambda-\mu)$ and $(x-y,x)\subset \R\backslash \text{supp}(\mu)$. Thus,  $x\in R_2$. This however contradicts the assumption that $x\in \Sn^{\circ}$. The other case is analogous. To prove (\ref{M2}), we note that
\begin{align*}
\vert \Delta Mf(x,y)\vert\leq \frac{1}{y}\int_{x-y}^{x+y}f(t)dt\leq  \frac{1}{y}\int_{\R}f(t)dt=\frac{1}{y},
\end{align*}
since $f\geq 0$ and $f\in \rho_{c,1}^{\lambda}(\R)$. 
\end{proof}
In what follows it will be useful to define:
\begin{Def}
\label{MaxSup}
\begin{align}
f_R^+(x)&:=\limsup_{h\rightarrow 0^+}M_Rf(x,h)=\limsup_{h\rightarrow 0^+}\frac{1}{h}\int_{x}^{x+h}f(t)dt\\
f_R^-(x)&:=\liminf_{h\rightarrow 0^+}M_Rf(x,h)=\liminf_{h\rightarrow 0^+}\frac{1}{h}\int_{x}^{x+h}f(t)dt \\
f_L^+(x)&:=\limsup_{h\rightarrow 0^+}M_Lf(x,h)=\limsup_{h\rightarrow 0^+}\frac{1}{h}\int_{x-h}^{x}f(t)dt  \\
f_L^-(x)&:=\liminf_{h\rightarrow 0^+}M_Lf(x,h)=\liminf_{h\rightarrow 0^+}\frac{1}{h}\int_{x-h}^{x}f(t)dt.
\end{align}
\end{Def}
\begin{Lem}
\label{EstHilb1}
Fix $x\in \Sn$. Let
\begin{align*}
c(x):=\max\{\vert f_R^{+}(x)-f_L^{-}(x)\vert,\vert f_L^{+}(x)-f_R^{-}(x)\vert\}
\end{align*}
Then for every $\eps>0$ and every non-tangentially convergent sequence $\{u_n+iv_n\}_n$ to x, such that $\{u_n+iv_n\}_n\subset \G_k(x)$, there exists an $N>0$ and a constant $C=C(\eps,x,k)$, such that 
\begin{align}
\label{Hilb1}
v_ne^{\vert H_{v_n}f(u_n)\vert}\leq Cv_n^{1-c(x)-\eps}.
\end{align}
If in particular $x\in \mathscr{L}_f$, then
\begin{align}
\label{Hilb2}
v_ne^{\vert H_{v_n}f(u_n)\vert}\leq Cv_n^{1-\eps}.
\end{align}
Finally, we have the identity
\begin{align}
\label{identityHilb}
\pi H_vf(u)=\int_{0}^{+\infty}t\frac{d}{dt}\bigg(\frac{t}{t^2+v^2}\bigg)\Delta Mf(u,t)dt.
\end{align}
\end{Lem}
\begin{proof}
Assume that $\{u_n+iv_n\}_n$ is non-tangentially convergent to $x$. Then $\{u_n+iv_n\}_n\subset \Gamma_k(x)$ for some $k>0$. An integration by parts gives
\begin{align*}
\pi H_{v_n}f(u_n)&=\int_{\R}\frac{tf(u_n-t)dt}{t^2+v_n^2}=\int_{0}^{+\infty}\frac{t}{t^2+v_n^2}\big[f(u_n-t)-f(u_n+t)\big]dt\\
&=-\int_{0}^{+\infty}\frac{d}{dt}\bigg(\frac{t}{t^2+v_n^2}\bigg)\bigg[\int_{0}^{t}\big[f(u_n-s)-f(u_n+s)\big]ds\bigg]dt=\int_{0}^{+\infty}t\frac{d}{dt}\bigg(\frac{t}{t^2+v_n^2}\bigg)\Delta Mf(u_n,t)dt.
\end{align*}
Choose $d=\max\{1,k\}$. Write,
\begin{align}
\label{twotermsHilb}
\pi H_{v_n}f(u_n)&=\int_{0}^{dv_n}t\frac{d}{dt}\bigg(\frac{t}{t^2+v_n^2}\bigg)\Delta  Mf(u_n,t)dt+\int_{dv_n}^{+\infty}t\frac{d}{dt}\bigg(\frac{t}{t^2+v_n^2}\bigg)\Delta  Mf(u_n,t)dt\nonumber \\
&=I_1^{(n)}+I_2^{(n)}.
\end{align}
By Lemma \ref{DeltaMean}
\begin{align}
\label{firsttermHilb}
\vert I_1^{(n)}\vert&\leq \int_{0}^{dv_n}t\bigg\vert \frac{d}{dt}\bigg(\frac{t}{t^2+v_n^2}\bigg)\bigg\vert dt=\int_{0}^{dv_n}t\bigg\vert \frac{v_n^2-t^2}{(t^2+v_n^2)^2}\bigg\vert dt\\
&\leq \int_{0}^{dv_n} \frac{d}{v_n} dt=d^2.
\end{align}
Now consider $I_2^{(n)}$ so that $t\geq dv_n\geq kv_n>\vert u_n-x\vert$. Then,

\begin{align*}
&\Delta  Mf(u_n,t)=\frac{1}{t}\bigg[\int_{u_n}^{u_n+t}f(y)dy-\int_{u_n-t}^{u_n}f(y)dy\bigg]=\frac{1}{t}\bigg[\int_{x+(u_n-x)}^{x+(u_n-x)+t}f(y)dy-\int_{x+(u_n-x)-t}^{x+(u_n-x)}f(y)dy\bigg]\\
&=\frac{1}{t}\bigg[\int_{x}^{x+(u_n-x)+t}f(y)dy-\int_{x}^{x+(u_n-x)}f(y)dy-\int_{x+(u_n-x)-t}^{x}f(y)dy-\int_{x}^{x+(u_n-x)}f(y)dy\bigg]\\
&=\frac{1}{t}\bigg[\frac{(u_n-x)+t}{(u_n-x)+t}\int_{x}^{x+(u_n-x)+t}f(y)dy-2\frac{(u_n-x)}{(u_n-x)}\int_{x}^{x+(u_n-x)}f(y)dy-\frac{t-(u_n-x)}{t-(u_n-x)}\int_{x+(u_n-x)-t}^{x}f(y)dy\bigg]\\
&=\frac{(u_n-x)+t}{t}M_Rf(x,(u_n-x)+t)-2\frac{(u_n-x)}{t}M_Rf(x,(u_n-x))-\frac{t-(u_n-x)}{t}M_Lf(x,t-(u_n-x)).
\end{align*}

If $u_n-x<0$, then similarly,

\begin{align*}
\Delta  Mf(u_n,t)&=\frac{t-(x-u_n)}{t}M_Rf(x,t-(x-u_n))+2\frac{(x-u_n)}{t}M_Lf(x,x-u_n)-\frac{t+(x-u_n)}{t}M_Lf(x,t+(x-u_n)).
\end{align*}

Let $0<\eps<1$. By definition, there exists an $N=N(\eps)$ and an $h=h(\eps,x)<1$, such that 
\begin{align*}
f_R^-(x)-\eps\leq&M_Rf(x,(u_n-x)+t)\leq f_R^+(x)+\eps\\
f_R^-(x)-\eps\leq&M_Rf(x,(u_n-x))\leq f_R^+(x)+\eps\\
f_L^-(x)-\eps\leq&M_Lf(x,(u_n-x)-t)\leq f_L^+(x)+\eps
\end{align*}
whenever $n>N$ and $dv_n<t<h$ and $u_n-x\geq0$, and 
\begin{align*}
f_R^-(x)-\eps\leq&M_Rf(x,t-(x-u_n))\leq f_R^+(x)+\eps\\
f_L^-(x)-\eps\leq&M_Lf(x,(x-u_n))\leq f_L^+(x)+\eps\\
f_L^-(x)-\eps\leq&M_Lf(x,(x-u_n)+t)\leq f_L^+(x)+\eps
\end{align*}
whenever $n>N$ and $dv_n<t<h$ and $u_n-x<0$. Thus, when $u_n-x\geq 0$ and $n>N$ and $dv_n<t<h$ ,
\begin{align*}
\Delta Mf(u_n,t)&\leq \frac{(u_n-x)+t}{t}( f_R^+(x)+\eps)-2\frac{(u_n-x)}{t}(f_R^-(x)-\eps)-\frac{t-(u_n-x)}{t}(f_L^-(x)-\eps)\\
&\leq f_R^+(x)-f_L^-(x)+2\eps+\frac{(u_n-x)}{t}(f_R^+(x)+\eps-2(f_R^-(x)-\eps)+(f_L^-(x)-\eps))\\
&= f_R^+(x)-f_L^-(x)+2\eps+\frac{(u_n-x)}{t}(f_R^+(x)-2f_R^-(x)+f_L^-(x)+2\eps)\\
&\leq f_R^+(x)-f_L^-(x)+2\eps+\frac{6(u_n-x)}{t},
\end{align*}
and
\begin{align*}
\Delta Mf(u_n,t)&\geq \frac{(u_n-x)+t}{t}( f_R^-(x)-\eps)-2\frac{(u_n-x)}{t}(f_R^+(x)+\eps)-\frac{t-(u_n-x)}{t}(f_L^+(x)+\eps)\\
&\geq f_R^-(x)-f_L^+(x)-2\eps+\frac{(u_n-x)}{t}(f_R^-(x)-\eps-2(f_R^+(x)+\eps)+(f_L^+(x)+\eps))\\
&= f_R^-(x)-f_L^+(x)-2\eps+\frac{(u_n-x)}{t}(f_R^-(x)-2f_R^+(x)+f_L^+(x)-2\eps)\\
&\geq f_R^-(x)-f_L^+(x)-2\eps-\frac{6(u_n-x)}{t}.
\end{align*}
 If instead $u_n-x<0$, then whenever $n>N$ and $dv_n<t<h$ and $\eps$ sufficiently small, then similarly
\begin{align*}
\Delta Mf(u_n,t)&\leq f_R^+(x)-f_L^-(x)+2\eps+\frac{6(x-u_n)}{t}
\end{align*}
and
\begin{align*}
\Delta Mf(u_n,t)&\geq f_R^-(x)-f_L^+(x)-2\eps-\frac{6(x-u_n)}{t}
\end{align*}
Hence, changing $\eps$ to $\eps/2$ in the calculation above we have for $dv_n\leq t <h(\eps,x)$ 
\begin{align}
\label{DeltaIneq}
\vert \Delta Mf(u_n,t)\vert \leq c(x)+\eps+\frac{6\vert u_n-x\vert}{t}.
\end{align}
Note that $\big\vert \frac{d}{dt}\big(\frac{t}{t^2+v_n^2}\big)\big\vert=-\frac{d}{dt}\big(\frac{t}{t^2+v_n^2}\big)$ when $t\geq v_n$. With $dv_n\leq t\leq h$, we can write
\begin{align*}
I_2^{(n)}=\bigg(\int_{dv_n}^{h}+\int_{h}^{1}+\int_{1}^{+\infty}\bigg)\bigg(-t\frac{d}{dt}\bigg(\frac{t}{t^2+v_n^2}\bigg)\Delta Mf(u_n,t)\bigg)dt.
\end{align*}
Use the estimate (\ref{DeltaIneq}) for $dv_n\leq t\leq h$, (\ref{M1}) for $h\leq t\leq 1$ and (\ref{M2}) for $t>1$. This gives
\begin{align*}
\vert I_2^{(n)}\vert &=\int_{dv_n}^{h}\bigg(-t\frac{d}{dt}\bigg(\frac{t}{t^2+v_n^2}\bigg)\bigg)\bigg(c(x)+\eps+\frac{6\vert u_n-x\vert}{t}\bigg)+\int_{h}^{1}\bigg(-t\frac{d}{dt}\bigg(\frac{t}{t^2+v_n^2}\bigg)\bigg)dt\\&+\int_{1}^{+\infty}\bigg(-\frac{d}{dt}\bigg(\frac{t}{t^2+v_n^2}\bigg)\bigg).
\end{align*}
We can now evaluate the integrals and use $\vert u_n-x\vert\leq kv_n$. Some straightforward estimates give
\begin{align*}
\vert I_2^{(n)}\vert \leq(c(x)+\eps)\log \sqrt{\frac{h^2+v_n^2}{(d^2+1)v_n^2}}+2+\frac{6k}{d}+\log\sqrt{\frac{R^2+v_n^2}{h^2+v_n^2}}.
\end{align*}
Together with (\ref{twotermsHilb}) and (\ref{firsttermHilb}) this proves (\ref{Hilb1}) together with an appropriate constant $C=C(\eps,x,h)$.
Finally, inequality (\ref{Hilb2}) follows from the fact that $c(x)=0$ whenever $x\in \mathscr{L}_f$.
\end{proof}
\begin{rem}
We note that inequality (\ref{Hilb1}) is trivial whenever $c(x)=1$ since $ve^{\pi \vert H_{v}f(u)\vert}\leq \tilde{c}$, for some positive constant $\tilde{c}$, whenever $f\in \rho_{c,1}^{\lambda}(\R)$.
\end{rem}
The following Lemma is a similar to Lemma \ref{EstHilb1}, but we only consider orthogonal limits. The estimate we derive will not depend on some $\eps>0$. This will be needed in the proof of Proposition \ref{GenericProp}.

\begin{Lem}
\label{EstHilb2}
We have the estimate
\begin{align}
\label{EstConv}
ve^{\pi \vert H_{v}f(x)\vert}\leq cv^{1-\Delta m_f(x)},
\end{align}
for $v>0$, where $c$ is a positive constant that does not depend on $x$.
\end{Lem}
\begin{proof}
Using (\ref{identityHilb}), we get the estimate 
\begin{align*}
\pi \vert H_{v}f(x)\vert &\leq \sup_{0\leq t\leq v}\vert \Delta Mf(x,t)\vert \int_{0}^{v}t\frac{d}{dt}\bigg(\frac{t}{t^2+v^2}\bigg)dt+\sup_{v\leq t\leq 1}\vert \Delta Mf(x,t)\vert \int_{v}^{1}-t\frac{d}{dt}\bigg(\frac{t}{t^2+v^2}\bigg)dt\\&+ \int_{1}^{\infty}-t\frac{d}{dt}\bigg(\frac{t}{t^2+v^2}\bigg)\vert \Delta Mf(x,t)\vert dt.
\end{align*}
Now using (\ref{M1}) in the first term, the definition of $\Delta m_f(x)$ in the second expression and (\ref{M2}) in the last expression to see that 
\begin{align*}
\pi \vert H_{v}f(x)\vert\leq C +\Delta m_f(x)\log(v^{-1}),
\end{align*}
where $C$ is a numerical constant.
\end{proof}

We will now consider the denominator $\sin(\pi P_vf(u))$ in (\ref{eqchin}) and (\ref{eqetan}). As will be shown in Lemma \ref{Sine}, the size of $\sin(\pi P_v f(u))$ can be estimated from below by $\frac{2d}{1+d^2}Mf(u,dv)$ for some arbitrary $d>0$, rather than the quantity $\Delta Mf(u,v)$ as in Lemma \ref{EstHilb1}.
\begin{Lem}
\label{Poisson1}
For any $u\in\R$ and $v>0$,
\begin{align*}
0<P_{v} f(u)<1.
\end{align*}
\end{Lem}
\begin{proof}
Clearly, $P_v f(u)>0$ since $\Vert f\Vert_1=1$ and $f(t)\geq 0$. Similarly, using that $f$ has compact support so that $\text{supp}(f)\subset [-R,R]$ for $R>0$ sufficiently large, we get
\begin{align*}
P_{v} f(u)&=\frac{1}{\pi}\int_{\R}\frac{vf(t)dt}{(u-t)^2+v^2}\leq \frac{1}{\pi}\int_{-R}^R\frac{vdt}{(u-t)^2+v^2}<1
\end{align*}
since $f(t)\leq 1$.
\end{proof}

\begin{Lem}
\label{Sine}
For any fixed $d>0$,
\begin{align*}
\sin(\pi P_{v}f(u))\geq  \frac{d}{1+d^2}\min\{ Mf(u,dv), M(1-f)(u,dv)\}.
\end{align*}
\end{Lem}
\begin{proof}
Using the inequality 
\begin{align*}
\sin t \geq \frac{\pi}{4}-\frac{1}{4}\vert 2t-\pi\vert
\end{align*}
valid for $t\in[0,\pi]$, we get using Lemma \ref{Poisson1},
\begin{align*}
\sin(\pi P_v f(u))\geq \frac{\pi}{4}-\frac{1}{4}\vert 2\pi P_v f(u)-\pi\vert.
\end{align*}
We now use the inequality 
\begin{align*}
\frac{v}{(u-t)^2+v^2}\geq \frac{1}{1+d^2}\frac{1}{v}
\end{align*}
valid for $t\in[u-dv,u+dv]$ and any fixed $d>0$, to get
\begin{align*}
\pi P_v f(u)\geq \int_{u-dv}^{u+dv}\frac{vf(t)dt}{(u-t)^2+v^2}\geq \frac{1}{1+d^2}\frac{1}{v}\int_{u-dv}^{u+dv}f(t)dt=\frac{2d}{1+d^2}Mf(u,dv),
\end{align*}
and similarly, 
\begin{align*}
\pi P_v f(u)=\pi-\pi P_v (1-f)(u)\leq \pi -\frac{2d}{1+d^2}M(1-f)(u,dv).
\end{align*}
Since
\begin{align*}
\vert 2\pi P_v f(u))-\pi\vert&= \left\{
\begin{array}{ll}   2\pi P_v f(u)-\pi& \text{if } P_v f(u)\geq \frac{1}{2}\\
\pi-2\pi P_v f(u) & \text{if } P_v f(u)<\frac{1}{2}\\
 \end{array} \right. ,
\end{align*}
\begin{align*}
\frac{\pi}{4}-\frac{1}{4}\vert 2\pi P_v f(u)-\pi\vert&= \left\{
\begin{array}{ll}   \frac{\pi}{2}-\frac{\pi}{2} P_v f(u)& \text{if } P_v f(u)\geq \frac{1}{2}\\
\frac{\pi}{2} P_v f(u)) & \text{if } P_v f(u)<\frac{1}{2}\\
 \end{array} \right. \\
&\geq \left\{
\begin{array}{ll}   \frac{d}{1+d^2}M(1-f)(u,dv)& \text{if } P_v f(u)\geq \frac{1}{2}\\
\frac{d}{1+d^2}Mf(u,dv) & \text{if } P_v f(u)<\frac{1}{2}\\
 \end{array} \right. \\
&\geq\frac{d}{1+d^2}\min\{Mf(u,dv),M(1-f)(u,dv)\}
\end{align*}
\end{proof}

\begin{Lem}
\label{PoissonEst}
Fix $x\in\Sn$. Let
\begin{align}
b(x):=\frac{1}{4}\min\{2-f_R^{+}(x)-f_L^{+}(x),f_R^-(x)+f_L^-(x)\}
\end{align}
Consider a non-tangentially convergent sequence such that $\{u_n+iv_n\}_n\subset \G_k(x)$ and fix $\eps>0$. Then, 
\begin{align*}
\sin(\pi P_{v_n} f(u_n))\geq \frac{2k}{1+4k^2}(b(x)-\eps)
\end{align*}
for $n$ sufficiently large.
\end{Lem}
\begin{proof}
By definition $\vert u_n-x\vert \leq kv_n$ for all $n$. Choose $d=2k$ in Lemma \ref{Sine}. Assume $u_n-x>0$. Then,
\begin{align*}
Mf(u_n,2kv_n)&=\frac{1}{4kv_n}\int_{u_n-2kv_n}^{u_n+2kv_n}f(t)dt=\frac{1}{4kv_n}\int_{x+(u_n-x)-2kv_n}^{x+(u_n-x)+2kv_n}f(t)dt\\&\geq \frac{1}{4kv_n}\int_{x}^{x+kv_n}f(t)dt+\frac{1}{4kv_n}\int_{x-kv_n}^{x}f(t)dt\\
&\geq \frac{1}{4}(M_Rf(x,kv_n)+M_Lf(x,kv_n))
\end{align*}
and the same estimate holds if $u_n-x<0$.
Thus,
\begin{align*}
&\min\{Mf(u_n,kv_n),M(1-f)(u_n,kv_n)\}\\&\geq \frac{1}{4}\min\{M_Rf(x,kv_n)+M_Lf(x,kv_n),M_R(1-f)(x,kv_n)+M_L(1-f)(x,kv_n)\}\\&
\geq b(x)-\eps
\end{align*}
whenever $n>N=N(\eps)$ say. Then Lemma \ref{Sine} implies that 
\begin{align*}
\sin(\pi P_{v_n} f(u_n))\geq \frac{2k}{1+4k^2}(b(x)-\eps)
\end{align*}
whenever $n>N$.
\end{proof}
We now give a version of Lemma \ref{PoissonEst} for orthogonal limits that will be need in Proposition \ref{GenericProp}.

\begin{Lem}
\label{PoissonEst1}
Fix $x\in\Sn$. Then, for any fixed $\delta>0$
\begin{align*}
\sin(\pi P_{v} f(x))\geq \frac{1}{2}\min\{1-m_f^\delta(x),1-m_{1-f}^\delta(x)\}
\end{align*}
for all $0<v<\delta$.
\end{Lem}

\begin{proof}
Since
\begin{align*}
\inf_{0<v<\delta}Mf(x,v)= 1-\sup_{0<v<\delta}M(1-f)(x,v)\geq 1-m_{1-f}^{\delta}(x),
\end{align*}
and
\begin{align*}
\inf_{0<v<\delta}1-Mf(x,v)= 1-\sup_{0<v<\delta}Mf(x,v)\geq 1-m_{f}^{\delta}(x),
\end{align*}
the result follows immediately from Lemma \ref{Sine} with $d=1$.
\end{proof}

\begin{Lem}
\label{Ineq}
Fix $x\in \Sn$. Then for every sequence $\{u_n+iv_n\}_n\subset \G_k(x)$ which converges non-tangentially to $x$, we have for every $\eps>0$ sufficiently small
\begin{align}
\label{NormIneq}
\vert (\chi_{\LL}(u_n,v_n)-u_n,\eta_{\LL}(u_n,v_n)-1)\vert \leq \frac{1+4k^2}{2k}\frac{\sqrt{20}Cv_n^{1-c(x)-\eps}}{b(x)-\eps},
\end{align}
where $C$ is the same constant as in Lemma \ref{EstHilb1}.
\end{Lem}
\begin{proof}
From (\ref{eqchin}) and (\ref{eqetan}) we see that
\begin{align*}
&\vert\sin[\pi P_{v_n} f(u_n)]\vert^2\vert (\chi_{\LL}(u_n,v_n)-u_n),\eta(u_n,v_n)_{\LL}-1)\vert^2 \\
&\leq (e^{-\pi H_{v_n}f(u_n)}-\cos(\pi P_{v_n} f(u))^2+(e^{\pi H_{v_n}f(u_n)}+e^{-\pi H_{v_n}f(u_n)}-2\cos(\pi P_{v_n} f(u_n))^2v_n^2\\
&\leq (e^{\pi \vert H_{v_n}f(u_n)\vert}+1)^2+(2e^{\pi \vert H_{v_n}f(u_n)\vert}+2)^2v_n^2\\
&\leq 5(1+3e^{2\pi \vert H_{v_n}f(u_n)\vert})v_n^2\\
&\leq 20C^2v_n^{2(1-c(x)-\eps)}
\end{align*}
by Lemma \ref{EstHilb1}, whenever $n$ is sufficiently large. Hence, by Lemma \ref{PoissonEst},
\begin{align*}
\vert (\chi_{\LL}(u_n,v_n)-u_n,\eta_{\LL}(u_n,v_n)-1)\vert \leq  \frac{1+4k^2}{2k}\frac{\sqrt{20}Cv_n^{1-c(x)-\eps}}{b(x)-\eps},
\end{align*}
if $\eps<b(x)$.
\end{proof}

We conclude this section with a similar estimate as in Lemma \ref{Ineq}, for orthogonal sequences, but where the constant is independent of $x$.
\begin{Lem}
\label{Ineq2}
For every $x\in \Sn$ and $\delta>0$, there exists a constant $c>0$ independent of $x$ and $\delta$, such that 
\begin{align}
\label{NormIneq2}
\vert (\chi_{\LL}(x,v)-x,\eta_{\LL}(x,v)-1)\vert \leq \frac{2\sqrt{20}cv^{1-\Delta m_f(x)}}{\min\{1-m_f^{\delta}(x),1-m_{1-f}^{\delta}(x)\}},
\end{align}
whenever $v<\delta$.
\end{Lem}
\begin{proof}
Combining Lemma \ref{EstHilb2} and Lemma \ref{PoissonEst}, a similar computation as in the proof of Lemma \ref{Ineq} gives (\ref{NormIneq2}).
\end{proof}

\section{Regular Points}
In this section we will discuss various sufficient criteria for a point to be regular and give examples of cases where these occur. As we shall see, it is natural to distinguish between regular Lebesgue points and regular non-Lebesgue points. This is due to the different behavior of $H_{v_n}f(u_n)$ and $P_{v_n}f(u_n)$ for non-tangential sequences $\{u_n+iv_n\}_n\in \Hp$ that converge to a regular point $x$, depending on whether $x$ is in the Lebesgue set or not.

\subsection{Regular Lebesgue Points}
We will first consider the cases when $x$ belongs to the Lebesgue set of $f$. The following proposition may be viewed as characterizing the typical case when a point $x$ is regular:
\begin{Prop}
\label{LPointReg}
Assume that $x\in \Sn\bigcap \mathscr{L}_f$ and that $0<f(x)<1$. Then $x$ is regular.
\end{Prop}
\begin{proof}
Since $x$ belongs to the Lebesgue set of $f$ we have by Lemma \ref{EstHilb1}, that $v_ne^{\pi \vert H_{v_n}f(u_n)\vert}\rightarrow 0$ for a non-tangential sequence $u_n+iv_n\in \Hp$ such such that $\lim_{n\rightarrow+\infty}u_n+iv_n=x$. Moreover, (see page 11)
\begin{align*}
\lim_{n\rightarrow+\infty}P_{v_n} f(u_n)=f(x)
\end{align*}
also holds for every such sequence. Hence
\begin{align*}
\lim_{n\rightarrow\infty}\frac{v_ne^{-\pi H_{v_n}f(u_n) } -v_n\cos(\pi P_{v_n} f(u_n))}{ \sin(\pi P_{v_n} f(u_n))}=0
\end{align*}
and
\begin{align*}
\lim_{n\rightarrow\infty}\frac{v_n e^{\pi H_{v_n}f(u_n)}+v_n e^{-\pi H_{v_n}f(u_n)}-2v_n\cos(\pi P_{v_n}f(u_n))}{\sin(\pi P_{v_n} f(u_n))}=0
\end{align*}
hold, which implies the claim by (\ref{eqchin}) and (\ref{eqetan}).
\end{proof}

We now consider the remaining cases where $f(x)=0$ or $f(x)=1$.
\begin{Prop}
\label{RegPoint1}
Assume that $x\in \Sn\bigcap \mathscr{L}_f$ is in the Lebesgue set of $f$ and that $f(x)=0$. Furthermore, assume that $\int_{\R}\frac{f(t)}{(x-t)^2}dt=\infty$ and that $H_vf(u)$ is non-tangentially bounded at $x$. Then $x$ is regular.
\end{Prop}
\begin{proof}
Since $x$ belongs to the Lebesgue set of $f$ and $f(x)=0$, we know that for every non-tangential sequence $\{u_n+iv_n\}_{n=1}^{+\infty}\subset \Hp$ we have $\lim_{n\rightarrow+\infty}P_{v_n}f(u_n)=0$. Moreover, $\sup_{(u,v)\in \Gamma_{\alpha}^1(x)}\vert H_vf(u)\vert=d_{\alpha}<+\infty$ by assumption. Consequently, for any non-tangential sequence $\{u_n+iv_n\}_n^{\infty}\subset \Gamma_{\alpha}^1(x)$, such that $\lim_{n\to \infty}u_n+iv_n=x$ we have
\begin{align*}
\lim_{n\rightarrow\infty}\frac{v_n\vert e^{\pi\vert H_{v_n}f(u_n) \vert}-\cos(\pi P_{v_n} f(u_n))\vert}{\vert \sin(\pi P_{v_n} f(u_n))\vert}\leq\lim_{n\rightarrow\infty}\frac{v_n\vert e^{\pi d_{\alpha}}+1\vert}{v_n\int_{\R}\frac{f(t)dt}{(u_n-t)^2+v_n^2}}=0
\end{align*}
since by Fatou's lemma $+\infty =\int_{\R}\frac{f(t)}{(x-t)^2}dt\leq\liminf_{n\rightarrow\infty}\int_{\R}\frac{f(t)dt}{(u_n-t)^2+v_n^2}$.
\end{proof}
\begin{ex}
Let $f(t)=\vert t\vert \chi_{[-1/2,1/2]}(t)+\frac{3}{4}\chi_{[1/2,3/2]}(t)$. Then $t=0$ satisfies the conditions of Proposition \ref{RegPoint1}. 
\end{ex}
\begin{Cor}
Assume that $x\in \Sn\bigcap \mathscr{L}_f$ is in the Lebesgue set of $f$ and that $f(x)=1$. Furthermore, assume that $\int_{\R}\frac{1-f(t)}{(x-t)^2}dt=\infty$ and that $H_vf(u)$ is non-tangentially bounded at $x$. Then $x$ is regular.
\end{Cor}
\begin{proof}
We see that
\begin{align*}
\sin(\pi P_v f(u))&=\sin(\pi-\pi P_v (1-f)(u))=\sin(\pi P_v (1-f)(u)).
\end{align*}
Now, repeating the same argument as in the Proposition \ref{RegPoint1} for $1-f$ gives the same result as before.
\end{proof}
\begin{ex}
Assume that for some $\delta>0$ $f\big\vert_{[-\delta,\delta]}(t)=t^2\chi_{[-\delta,0)}(t)+\frac{1}{\log t^{-1}}\chi_{[0,\delta]}(t)$. Then $f$ is continuous at $t=0$, in particular $t$ belongs to the Lebesgue set of $f$. However, since
\begin{align*}
\int_{0}^{\delta}\frac{\vert f(t)-f(-t)\vert dt}{t}=\int_{0}^{\delta}\frac{dt}{-t\log t}-\int_{0}^{\delta}tdt=+\infty,
\end{align*}
we have that $\vert \mathcal{H}(x)\vert =+\infty$. Therefore this example does not satisfy the assumptions of Proposition \ref{RegPoint1}. However, this example is covered by the following proposition:
\end{ex}
\begin{Prop}
\label{propRegPoint2}
Assume that $x\in  \Sn\bigcap \mathscr{L}_f$ is in the Lebesgue set of $f$ and that $f(x)=0$. Furthermore, assume that there exists constants $\tilde{c}$, $0<\beta<1$ and $\delta>0$ such that either 
\begin{align}
M_Rf(x,t)\geq \tilde{c}t^{\beta} 
\end{align}
or 
\begin{align}
M_Lf(x,t)\geq \tilde{c}t^\beta
\end{align}
holds whenever $t \leq \delta$. Then $x$ is regular.
\end{Prop}
\begin{proof}
Let $\{u_n+iv_n\}_n\subset \G_k(x)$ be non-tangentially convergent to $x$. Using Lemma \ref{Sine} with $d=2k$, a similar computation as in Lemma \ref{Ineq} gives
\begin{align*}
\vert (\chi_{\LL}(u_n,v_n)-u_n,\eta_{\LL}(u_n,v_n)-1)\vert &\leq  \frac{1+4k^2}{k}\frac{\sqrt{20}C(\eps)v_n^{1-\eps}}{M_Rf(x,kv_n)+M_Lf(x,kv_n)}\\
&\leq \frac{1+4k^2}{k}\frac{\sqrt{20}C(\eps)v_n^{1-\eps}}{\tilde{c}k^{\beta}v_n^{\beta}}=\frac{1+4k^2}{k}\frac{\sqrt{20}C(\eps)v_n^{1-\eps-\beta}}{\tilde{c}k^{\beta}}
\end{align*}
for every $\eps>0$ sufficiently small. Choosing $\eps$ so small that $1-\eps-\beta>0$ shows that $\vert (\chi_{\LL}(u_n,v_n)-u_n,\eta_{\LL}(u_n,v_n)-1)\vert\to 0$ as $n\to \infty$.
\end{proof}
\begin{rem}
We expect that the assumption that $\int_{\R}\frac{f(t)dt}{(x-t)^2}=+\infty$ is  sufficient for the point $x$ to be regular, irrespective of whether $\vert H_vf(u)\vert$ is non-tangentially bounded at $x$ or not. Moreover we note that if $f\in \Hc$, then the set of points such that $f(x)=0$, $x$ is a Lebesgue point and $\vert H_vf(u)\vert$ is non-tangentially unbounded, is empty. This is due to the fact that if $f$ is locally H\"older continuous, the so is $\mathcal{H}f$, which implies that $\vert H_vf(u)\vert$ is non-tangentially bounded.
\end{rem}
With this remark in mind we are lead to make the following conjecture:
\begin{Con}
Assume that $x\in \Sn$ is in the Lebesgue set of $f$ and that $f(x)=0$. Furthermore, assume that  $\int_{\R}\frac{f(t)dt}{(x-t)^2}=+\infty$. Then $x$ is regular.
\end{Con}
\begin{rem}
The situation in Proposition 4 is slightly special in the sense that the point $x$ is regular due to the fact that $\mathcal{H}f(x)=0$. If this is not the case then the point is singular by Proposition \ref{Sing3}.
\end{rem}

\begin{Prop}
\label{propRegPoint3}
Assume that $x\in  \Sn\bigcap \mathscr{L}_f$. Furthermore, assume that $\mathcal{H}f(x)=0$ and that $\int_{\R}\frac{f(t)dt}{(x-t)^2}<+\infty$, or that $\int_{\R}\frac{1-f(t)dt}{(x-t)^2}<+\infty$. Then $x$ is regular. In particular this holds when $f$ is symmetric around $x$, that is if $f(x-t)=f(x+t)$ for all $t$.
\end{Prop}
\begin{proof}
We consider the case when $\int_{\R}\frac{f(t)dt}{(x-t)^2}<+\infty$. Note that this implies that $f(x)=0$ which in turn implies that $\int_{\R}\frac{f(t)dt}{\vert x-t \vert }=\int_{\R}\frac{\vert f(t)-f(x)\vert dt}{\vert x-t \vert }<+\infty$. However, by Remark \ref{R1} and our assumption that $\mathcal{H}f(x)=0$, this implies that all non-tangential limits of $H_{v_n}f(u_n)$ converge to 0 at $x$. Thus, for every non-tangential sequence $\{u_n+iv_n\}_{n=1}^{+\infty}\subset \Hp$ we have 
\begin{align*}
\lim_{n\rightarrow\infty}\vert e^{\pi H_{v_n}f(u_n) }-\cos(\pi P_{v_n} f(u_n))\vert=\lim_{n\rightarrow\infty}\vert e^{-\pi H_{v_n}f(u_n)}-\cos(\pi P_{v_n} f(u_n))\vert=0.
\end{align*}
Since 
\begin{align*}
\lim_{n\rightarrow\infty}\frac{v_n}{\sin(\pi P_{v_n}\ast f(u_n))}=\frac{1}{\int_{\R}\frac{f(t)}{(x-t)^2}dt}<+\infty
\end{align*}
the result follows.
\end{proof}
\begin{ex}
Let $f(t)=t^2\chi_{[-1,1]}(t)+\chi_{[-1-1/6,-1]}(t)+\chi_{[1,1+1/6]}(t)$. Then $x=0$ satisfies the assumptions of Proposition \ref{propRegPoint3}.
\end{ex}

\subsection{Regular Non-Lebesgue Points}
In this section we will provide sufficient conditions that can be used to show that certain non-Lebesgue points are regular. 
\begin{Prop}
\label{NLpointReg}
Fix $x\in \Sn$. Furthermore, assume that 
\begin{align*}
c(x)=\min\{\vert f_R^+(x)-f_L^-(x)\vert,\vert f_L^+(x)-f_R^-(x)\vert \}<1
\end{align*}
and that 
\begin{align*}
b(x)=\min\{ 2-f_R^+(x)-f_L^+(x),f_R^-(x)+f_L^-(x) \}>0.
\end{align*}
Then $x$ is regular.
\end{Prop}
\begin{proof}
Let $\{u_n+iv_n\}_n\subset \G_k(x)$ be non-tangentially convergent to $x$. Choose $\eps>0$ so that $1-c(x)-\eps>0$ and $b(x)-\eps>0$. Then Lemma \ref{Ineq} gives
\begin{align*}
\vert (\chi_{\LL}(u_n,v_n)-u_n,\eta_{\LL}(u_n,v_n)-1)\vert \leq  \frac{1+4k^2}{2k}\frac{\sqrt{20}C(\eps)v_n^{1-c(x)-\eps}}{b(x)-\eps}\to 0
\end{align*}
as $n\to \infty$.
\end{proof}
We now give two examples where one can use Proposition \ref{NLpointReg} to conclude that a point $(x=0)$, is regular. The difference between the examples is that we will be able to use Proposition \ref{GenericProp} to also conclude that in the first example $x=0$ is also a generic point. In the second example however, out analysis in the present paper is insufficient to determine whether $x=0$ is a generic point or not.
\begin{ex}
\label{NRegEx0}
Let $f(t)=\frac{1}{2}\vert \sin(t^{-1})\vert^{1/2}\chi_{[-a,a]}(t)$, where $a$ is chosen so that $\Vert f\Vert_1=1$. Since $0\leq f(t)\leq 1/2$ for all $t\in [-a,a]$, it follows that $c(0)<1$ and that $f_R^+(0)=f_L^+(0)<1$. By symmetry of $f$, $M_Rf(0,h)=M_Lf(0,h)$. Take $h>0$ sufficiently small, and choose an integer $n$, such that $\pi (n-1)<1/h<\pi n$. Then using that $\sin t\geq \frac{2}{\pi}t$ when $0\leq t\leq \pi/2$ 
\begin{align*}
M_Rf(0,h)&=\frac{1}{2h}\int_0^h\vert \sin(t^{-1})\vert^{1/2}dt=\frac{1}{2h}\int_{h^{-1}}^{\infty}\vert \sin(x)\vert^{1/2}\frac{dx}{x^2}\\
&>\frac{\pi n}{2}\sum_{k=n}^{\infty}\frac{1}{\pi^2(k+1)^2}\int_{\pi k}^{\pi (k+1)}\vert \sin(x)\vert^{1/2}dx>\frac{1}{2\pi}\sum_{k=n}^{\infty}\frac{n}{(k+1)^2}\int_{0}^{\pi/2}\sqrt{2t}dt\\
&>\frac{\sqrt{\pi}}{6}n\int_{n+1}^{\infty}\frac{dx}{x^2}=\frac{\sqrt{\pi}}{6}>0.
\end{align*}
Hence, $f_R^-(0)=f_L^-(0)>0$, which implies that $b(x)>0$. Thus, $x$ is regular.
\end{ex}
\begin{ex}
\label{NRegEx2}
Let $I_n=(2^{-(n+1)},2^{-n}]$ and let
\begin{align*}
f(t)=\sum_{k=1}^{+\infty}\bigg(1-\frac{1}{2k}\bigg)(\chi_{I_{2k}}(t)+\chi_{I_{2k}}(-t))+\sum_{k=1}^{+\infty}\frac{1}{2k+1}(\chi_{I_{2k+1}}(t)+\chi_{I_{2k+1}}(-t))+\chi_{(1/2,a]}(t),
\end{align*}
where $a$ is chosen so that $\int_{\R}f(t)dt=1$. One can show that
\begin{align*}
\limsup_{h\rightarrow 0^+}M_Rf(0,h)&=\limsup_{h\rightarrow 0^+}M_Lf(0,h)\leq \frac{3}{4}<1.
\end{align*}
Similarly, one can also show that
\begin{align*}
\liminf_{h\rightarrow 0^+}M_Rf(0,h)&=\liminf_{h\rightarrow 0^+}M_Lf(0,h)
\geq  \frac{1}{4}>0.
\end{align*}
Similarly, one can show that $\lim_{h\rightarrow 0^+}M_f(0,h)$ does not exist. Hence, $x=0$ does not belong to the Lebesgue set of $f$. However we see that $f$ satisfies the conditions of Proposition \ref{NLpointReg}. This implies that $0$ is a regular point. 
\end{ex}
In this proposition we consider a particular case of when $f$ jumps from 0 to 1 or 1 to 0 in mean.
\begin{Prop}
\label{NReg1}
Assume that $x\in \Sn^{\circ}$ and that for some $\delta>0$ we have either
\begin{align*}
f(t)&=\chi_{[x,x+\delta]}(t)+\varphi(t) \quad (i)\quad \text{or}\\
f(t)&=\chi_{[x-\delta,x]}(t)+\varphi(t) \quad (ii).
\end{align*}
Furthermore, assume that $x$ belongs to the Lebesgue set of $\varphi$, $\varphi(x)=0$ and that $\vert \mathcal{H}\varphi(x)\vert=+\infty$. Then $x$ is a regular point.
\end{Prop}

\begin{proof}
We prove the proposition assuming $(i)$ holds. The case when $(ii)$ holds follows from case $(i)$. Fix an arbitrary $k>0$ and let $\{u_n+iv_n\}_n\in \Gamma_k(x)$ be an arbitrary non-tangential sequence such that $\lim_{n\to+\infty}u_n+iv_n=x$. A computation gives
\begin{align*}
\pi P_{v_n} f(u_n)&=\int_{x}^{x+\delta}\frac{v_ndt}{(u_n-t)^2+v_n^2}+\int_{\R}\frac{v_n\varphi(t)dt}{(u_n-t)^2+v_n^2}\\
&=\arctan\frac{u_n-x}{v_n}-\arctan\frac{u_n-x-\delta}{v_n}+\pi P_{v_n} \varphi(u_n).
\end{align*}
Hence,
\begin{align*}
\frac{\pi}{2}-\arctan{k}+o(1)\leq \pi P_{v_n} f(u_n)
&\leq \frac{\pi}{2}+\arctan{k}+o(1)
\end{align*}
since $\vert u_n-x\vert < kv_n$ and $\lim_{n\to \infty}\pi P_{v_n} \varphi(u_n)=0$ since $x$ is a Lebesgue point of $\varphi$. This implies that $\vert \sin(\pi P_{v_n}f(u_n))\vert \geq c_1^{-1}>0$, for some constant $c_1$. We now consider $\pi H_{v_n}\varphi(u_n)$. A simple modification of the proof of Lemma \ref{EstHilb1} shows that for every $\eps\in(0,1)$ there exists a constant $C=C(\eps,k,x)$ such that $\vert \pi H_{v_n}\varphi(u_n)\vert \leq \eps \log(v_n^{-1})+C$. 

We now show that in fact $\pi \mathcal{H}\varphi(x)=+\infty$. By assumption $\varphi(t)\leq 0$ for all $t\in [x,x+\delta]$ and $\varphi(t)\geq 0$ otherwise. Hence,
\begin{align*}
\lim_{v\to 0^+}\int_{x-\delta}^{x+\delta}\frac{(x-t)\varphi(t)dt}{(x-t)^2+v^2}=\lim_{v\to 0^+}\int_{x-\delta}^{x+\delta}\frac{\vert (x-t)\varphi(t)\vert dt}{(x-t)^2+v^2}=+\infty
\end{align*}
by assumption and Lemma 1.2 in chapter VI in \cite{Stein}. 

We now show that in fact we must have $\lim_{n\to\infty}\pi H_{v_n}\varphi(u_n)=+\infty$. Since $x$ is in the Lebesgue set of $\varphi$, we have
\begin{align*}
\int_{\R}\frac{(u_n-t)\varphi(t)dt}{(u_n-t)^2+v_n^2}=o(1)+\int_{\R\backslash [x-2\vert u_n-x\vert,x+2\vert u_n-x\vert]}\frac{(u_n-t)\varphi(t)dt}{(u_n-t)^2+v_n^2}
\end{align*}
since
\begin{align*}
\bigg\vert \int_{x-2\vert u_n-x\vert }^{x+2\vert u_n-x\vert}\frac{(u_n-t)\varphi(t)dt}{(u_n-t)^2+v_n^2}\bigg\vert \leq \frac{k}{v_n}\int_{x-2kv_n }^{x+2kv_n}\vert \varphi(t) \vert dt=o(1).
\end{align*}
There exists a constant $c>0$, such that 
\begin{align*}
\frac{\vert u_n-t\vert }{(u_n-t)^2+v_n^2}\geq \frac{c}{\vert x-t\vert }
\end{align*}
for all $t\in \R\backslash [x-2\vert u_n-x\vert,x+2\vert u_n-x\vert]$. Consequently,
\begin{align*}
\pi H_{v_n}\varphi(u_n)\geq o(1)+c\int_{\R\backslash [x-2\vert u_n-x\vert,x+2\vert u_n-x\vert]}\frac{\varphi(t)dt}{x-t}.
\end{align*}
Since $\lim_{\eps \to 0^+}\int_{\vert x-t\vert>\eps}\frac{\varphi(t)dt}{x-t}=\mathcal{H}\varphi(x)=+\infty$, the result follows.

A computation gives
\begin{align*}
\pi H_{v_n}f(u_n)&=\int_{x}^{x+\delta}\frac{(u_n-t)dt}{(u_n-t)^2+v_n^2}+\int_{\R}\frac{(u_n-t)\varphi(t)dt}{(u_n-t)^2+v_n^2}\\
&=\log\sqrt{(u_n-x)^2+v_n^2}-\log\sqrt{(x+\delta-u_n)^2+v_n^2}+\pi H_{v_n}\varphi(u_n)\\
&=\log\sqrt{k_n^2v_n^2+v_n^2}-\log\sqrt{(x+\delta-u_n)^2+v_n^2}+\pi H_{v_n}\varphi(u_n)\\
&=\log(v_n)+\log\sqrt{k_n^2+1}-\log\sqrt{(x+\delta-u_n)^2+v_n^2}+\pi H_{v_n}\varphi(u_n)\\
&=\log(v_n)+\pi H_{v_n}\varphi(u_n)+O(1),
\end{align*}
where $k_n=(u_n-x)/v_n$ and $\vert k_n\vert \leq k$ for all $n$. Hence,
\begin{align*}
&\lim_{n\rightarrow \infty}\frac{v_n\Big\{e^{-\pi H_{v_n}f(u_n)}-\cos(\pi P_{v_n} f(u_n))\Big\}}{\sin(\pi P_{v_n} f(u_n))}\\&=\lim_{n\rightarrow \infty}c_1e^{-\pi H_{v_n}\varphi(u_n)+O(1)}=0.
\end{align*}
and
\begin{align*}
&\lim_{n\rightarrow \infty}\frac{v_n\Big\{e^{\pi H_{v_n}f(u_n)}+e^{-\pi H_{v_n}f(u_n)}-2\cos(\pi P_{v_n} f(u_n))\Big\}}{\sin(\pi P_{v_n} f(u_n))}
\\&=\lim_{n\rightarrow \infty}c_1v_n^2e^{\pi H_{v_n}\varphi(u_n)+O(1)}+c_1e^{-\pi H_{v_n}\varphi(u_n)+O(1)}\\
&=\lim_{n\rightarrow \infty}c_1v_n^2e^{\eps \log v_n^{-1}+C+O(1)}+c_1e^{-\pi H_{v_n}\varphi(u_n)+O(1)}=0.
\end{align*}
This concludes the proof.
\end{proof}
\begin{ex}
Assume that for some $\delta>0$ $f\big\vert_{[-\delta,\delta]}(t)=(1-t^2)\chi_{[-\delta,0)}(t)+\frac{1}{\log t^{-1}}\chi_{[0,\delta]}(t)$. Then $t=0$ satisfies the assumptions of Proposition \ref{NReg1}.
\end{ex}
\begin{rem}
We note that if $f\in \Hc$ then the set of points which satisfies the assumptions of Proposition \ref{NReg1} is empty.
\end{rem}

\section{Generic Points}

\subsection{Generic Points Are Dense}

\begin{Prop}
\label{Topology1}
Define the set $G\subset \Sn^{\circ}=(\text{supp}(\mu)\cap \text{supp}(\lambda-\mu))^{\circ}$ according to
\begin{align}
\label{GenericSet}
G:= \Sn^{\circ}\bigcap \mathscr{L}_f\bigcap \mathscr{L}_{m_{\Hil f}}\bigcap \{t\in\R: 0<f(t)<1\}
\end{align}
Then every $x\in G$ is regular, and $G$ is dense in $\Sn^{\circ}$. Furthermore, for every interval $I \subset \Sn^{\circ}$, $\vert G\bigcap I\vert=\lambda(G\bigcap I)>0$. If in addition the inequality $0<f(t)<1$ holds almost everywhere in $\Sn^{\circ}$, then almost every $x\in \Sn^{\circ}$ belongs to $G$.
\end{Prop}
\begin{proof}
By Proposition \ref{LPointReg} every $x\in G$ is regular, and since $G$ is the finite intersection of measurable sets, $G$ is measurable. Since $f,m_{\Hil f}\in L^1_{loc}(\R)$ it follows that almost every $x\in \Sn^{\circ}$ belongs to $\mathscr{L}_f\bigcap \mathscr{L}_{m_{\Hil f}}$. Let $X=\Sn^{\circ}\bigcap \{t\in\R: 0<f(t)<1\}$.  By Hypothesis \ref{hypConv2}, $\lambda(G\bigcap I)>0$ for every interval $I\subset \Sn^{\circ}$, thus in particular $I\bigcap G\neq \varnothing$. This proves that $G$ is dense in $\Sn^{\circ}$. Finally, if the inequality $0<f(t)<1$ holds almost everywhere in $\Sn^{\circ}$, then $\lambda(X\bigcap\mathscr{L}_f\bigcap \mathscr{L}_{m_{\Hil f}}\bigcap \Sn^{\circ })=\lambda(\Sn^{\circ})$.
\end{proof}
\begin{rem}
We believe that Hypothesis \ref{hypConv2} is not necessary. That is, we believe that if Hypothesis \ref{hypConv2} is not true, then Proposition \ref{Topology1} remains true if the set $G$ is changed to
\begin{align*}
G= \Sn^{\circ}\bigcap \mathscr{L}_f\bigcap \mathscr{L}_{m_{\Hil f}}\bigcap\bigg(\Sn^{\circ}\backslash\bigg(\bigg\{x: \int_{\R}\frac{f(t)dt}{(x-t)^2}<+\infty\bigg\}\bigcup\bigg\{x: \int_{\R}\frac{1-f(t)dt}{(x-t)^2}<+\infty\bigg\}\bigg)\bigg).
\end{align*}
This change of typical set however, would require a substantial change of Lemma \ref{Generic1}. 
\end{rem}

We now give a lemma which will be useful for proving that $\Dl(x)=\{(x,1)\}$ for a typical set in $\Sn^{\circ}$.
\begin{Lem}
\label{Generic1}
Assume that $x\in G$. Then there exists sequences $\{r_n\}_n\subset G$ and $\{l_n\}_n\subset G$, such that $x<r_{n+1}<r_n$ and $l_{n+1}>l_n>x$ for all $n$ and $\lim_{n\to \infty}r_n=\lim_{n\to \infty}l_n=x$. Moreover
\begin{align}
\max\{\sup_{n}m_{\Hil f}(r_n),\sup_{n}m_{\Hil f}(l_n)\}<+\infty
\end{align}
and
\begin{align}
\min\{\inf_{n}f(r_n),\inf_{n}(1-f(r_n)),\inf_{n}f(l_n),\inf_{n}(1-f(l_n))\}>0.
\end{align}
\end{Lem}
\begin{proof}
Since $0<f(x)<1$ we can take $\eps>0$ so that $0<f(x)-\eps<f(x)+\eps<1$, and since $x\in\mathscr{L}_{f}\bigcap\mathscr{L}_{m_{\mathcal{H}f}}$, there exists an $\delta=\delta(\eps,x)$ such that 
\begin{align}
\label{MeasureIneq1}
&\frac{1}{h}\int_{x}^{x+h}\vert f(t)-f(x)\vert dt <\frac{\eps}{2}\\
\label{MeasureIneq2}
&\frac{1}{h}\int_{x}^{x+h}\vert m_{\Hil f}(t)-m_{\Hil f}(x)\vert dt<\frac{\eps}{2}
\end{align}
whenever $h<\delta$. Now, assume that 
\begin{align*}
\liminf_{h\to 0^+}\frac{\vert\{t\in[x,x+h]:\vert f(t)-f(x)\vert<\eps\} \vert}{h}=0.
\end{align*}
Then there exists a sequence $\{h_k\}_k$ such that $\lim_{k\to\infty}h_k=0$ and $\frac{\vert\{t\in[x,x+h_k]:\vert f(t)-f(x)\vert<\eps\} \vert}{h_k}<\eps/2$. Consequently, $\vert \{t\in[x,x+h_k]:\vert f(t)-f(x)\vert\geq\eps\}\vert \geq h_k(1-\eps/2)$, and so
\begin{align*}
\frac{1}{h_k}\int_{x}^{x+h_k}\vert f(t)-f(x)\vert dt \geq \eps(1-\eps/2)>\frac{\eps}{2}
\end{align*}
for all $k$ since $\eps<1/2$. However, this contradicts (\ref{MeasureIneq1}). Therefore, let 
\begin{align*}
\inf_{0<h<\delta}h^{-1}\vert\{t\in[x,x+h]:\vert f(t)-f(x)\vert<\eps\} \vert=d>0.
\end{align*}
Recall the John-Nirenberg inequality (\ref{JohnNirenberg}) and choose an $N$ so large that
\begin{align*}
c_1\exp\bigg(-c_2\frac{N}{\Vert m_{\Hil f}\Vert_{BMO}}\bigg)<\frac{d}{4}.
\end{align*}
Then 
\begin{align*}
\frac{\vert \{t\in [x-h,x+h]:\vert m_{\Hil f}(t)-m_{\Hil f}(x)\vert>N+\eps\}\vert}{2h}\leq \frac{d}{4}
\end{align*}
so that
\begin{align*}
\vert \{t\in [x,x+h]:\vert m_{\Hil f}(t)-m_{\Hil f}(x)\vert\leq N-\eps\}\vert\geq h-\frac{dh}{2},
\end{align*}
where we have used that if $\vert m_{\Hil f}(t)-Mm_{\Hil f}(x,h)\vert>N$ then $ \vert m_{\Hil f}(t)-m_{\Hil f}(x)\vert>N+\vert m_{\Hil f}(x)-Mm_{\Hil f}(x,h)\vert>N+\eps$ by (\ref{MeasureIneq2}).
Therefore, by Lemma \ref{lemIncExcMeasure}
\begin{align*}
\vert \{t\in [x,x+h]:\vert m_{\Hil f}(t)-m_{\Hil f}(x)\vert\leq N-\eps\}\bigcap \{t\in[x,x+h]:\vert f(t)-f(x)\vert<\eps\}\vert \geq \frac{dh}{2}
\end{align*}
for every $0<h<\delta$. Fix an $0<h_0<\delta $ and choose an $r_0$ in $ \{t\in (x,x+h_0):\vert m_{\Hil f}(t)-m_{\Hil f}(x)\vert\leq N\}\bigcap \vert\{t\in(x,x+h_0):\vert f(t)-f(x)\vert<\eps\}\vert $. Now take $h_1<\min\{r_0,h_0/2\}$, and choose $r_1$ in $ \{t\in (x,x+h_1):\vert m_{\Hil f}(t)-m_{\Hil f}(x)\vert\leq N\}\bigcap \{t\in(x,x+h_1):\vert f(t)-f(x)\vert<\eps\} $. Iteration of this process gives a sequence $\{r_n\}_n$ with the desired properties since $N$ is fixed, $f(r_n)>f(x)-\eps$ and $1-f(r_n)>1-(f(x)+\eps)>0$. A similar argument as above also yields the sequence $\{l_n\}_n$.
\end{proof}

We now want to consider the question of whether we can determine $\dv \mathcal{L}(x)$ whenever $x\in \Sreg$. Recall Lemma 2.5 in \cite{Duse14a}, where we showed that 
if $x\in\Sn^{\circ}$ and there exists a neighborhood $N_x$ of $x$ such that 
\begin{align}
\label{Cond1}
\sup_{t\in N_x}\{f(t),1-f(t)\}<1,
\end{align}
then $\dv \mathcal{L}(x)=\{(x,1)\}$. Note that if the density $f\in \Hc$, then every point $x\in\Sn^{\circ}$ for which $0<f(x)<1$ satisfies this condition. If condition (\ref{Cond1}) is not satisfied for $x\in\Sreg^{\circ}$, then it is considerably harder to prove that $\dv \mathcal{L}(x)=\{(x,1)\}$. The reason is the following:  Even though for some point $x$, one knows that for every point $x'$ in a neighborhood $N_x$ of $x$ one has that $\lim_{n\to +\infty}(\chi_\LL(w_n),\eta_\LL(w_n))=(x',1)$ whenever $\{w_n\}_n$ is a non-tangential sequence such that $\lim_{n\to +\infty}w_n=x'$, this does not necessarily imply that $\dv \mathcal{L}(x)=\{(x,1)\}$. The difficulty comes from the fact that tangential limits also have to be considered. The next example illustrates the difficulty.
\begin{ex}
Let 
\begin{align*}
\varphi(x,y)=e^{16/x^8}\exp\bigg\{-\frac{1}{(y-x^2)(y-2x^2)}\bigg\}\chi_{x^2<y<2x^2}(x,y).
\end{align*}
Then $\varphi\in C^{\infty}(\Hp)$, and $\text{supp}(\varphi)=\{(x,y)\in \overline{\Hp}: x^2\leq y\leq 2x^2\}$. Moreover, for every non-tangential limit $\{w_n\}_n\in \Hp$, such that $\lim_{n\to\infty}w_n=x$, we have $\lim_{n\to \infty}\varphi(w_n)=0$. However, $\lim_{x\to 0^+}\varphi(x,3x^2/2)=1$. Hence, $\varphi\notin C(\overline{\Hp})$.
\end{ex}
The argument that is missing in order to conclude that for a regular point $x\in \Sreg$, one has $\dv \mathcal{L}(x)=\{(x,1)\}$, is that if one for example specialize to orthogonal limits, then one need that $\lim_{v\to 0^+}(\chi_\LL(y,v),\eta_\LL(y,v))=(y,1)$ uniformly, for every $y$ in a compact neighborhood of $x$. We now prove that this is sufficient.
\begin{Lem}
\label{lemUniformconv}
Assume that $x\in\Sn^{\circ}=(\text{supp}(\mu)\cap \text{supp}(\lambda-\mu))^{\circ}$ is regular. Furthermore, assume that there exists sequences of regular points $\{r_n\}_n$ and $\{l_m\}_m$ such that $r_n>x$ and $r_m<x$ for all $n,m$, and such that $\lim_{n\to \infty}r_n=\lim_{m\to \infty}l_m=x$, and such that $\lim_{v\to 0^+}(\chi(r_n,v),\eta(r_n,v))=(r_n,1)$ and $\lim_{v\to 0^+}(\chi_\LL(l_m,v),\eta_\LL(l_m,v))=(l_m,1)$ uniformly for all $n,m$. Then $x$ is generic. More exactly, for any sequence $u_n+iv_n\in \Hp$ such that $\lim_{n\rightarrow +\infty}u_n+iv_n=x$,
\begin{align*}
\lim_{n\rightarrow +\infty}(\chi_\LL(u_n,v_n),\eta_\LL(u_n,v_n))&=(x,1).
\end{align*}
\end{Lem}

\begin{proof}
By possibly passing to a subsequence, we may assume that the sequence $\{w_l\}_l$ is tangential and that $u_l>x$ for all $l$. According to the assumptions, there exists a sequence of regular points $\{r_n\}$ such that $r_n>x$ and $\lim_{n\to \infty}r_n=x$, and $\lim_{v\to 0^+}(\chi_\LL(r_n,v),\eta_\LL(r_n,v))=(r_n,1)$ uniformly for all $n$. In particular, for every $\eps>0$ sufficiently small, there exists a $\delta=\delta(\eps)$, such that \newline$\vert (\chi_\LL(r_n,v)-r_n,\eta_\LL(r_n,v)-1)\vert<\eps$ whenever $v<\delta$ for all $n$. Choose $k<\eps/(r_1-x)$ and consider the non-tangential line $\{t+ikt: t\in(0,+\infty]\}$. Since $x$ is regular, it follows that $\lim_{t\to 0^+}(\chi_\LL(x+t,kt),\eta_\LL(x+t,kt))=(x,1)$. Consider the sequence of open sets $X_n^{(k)}$ defined according to
\begin{align*}
X_n^{(k)}:=\{(u,v)\in \Hp: x<u< r_n, 0<v<k(u-x)\}.
\end{align*}
Then $w_l\in X_n^{(k)}$ whenever $l>L$, for some $L=L(n)$. Moreover, $(\chi_\LL(w_l),\eta_\LL(w_l))\in \overline{W_{\mathcal{L}}^{-1}(X_n^{(k)})}$ since the map $W_{\mathcal{L}}$ is a homeomorphism. Then,
\begin{align*}
\vert (\chi_\LL(w_l)-x,\eta_\LL(w_l)-1)\vert \leq d((x,1),\overline{W_{\mathcal{L}}^{-1}(X_n^{(k)})})=d((x,1),\dv{W_{\mathcal{L}}^{-1}(X_n^{(k)})})
\end{align*}
whenever $l>L(n)$. Note that $d$ denotes the Hausdorff distance between sets, that is, if $X,Y\subset \R^2$, then 
\begin{align*}
d(X,Y)=\max\{\sup_{x\in X}\inf_{y\in Y}\vert x-y\vert,\sup_{y\in Y}\inf_{x\in X}\vert x-y\vert\}.
\end{align*}
Let $T_n$ be the closed region, whose boundary $\dv T_n$, can be decomposed into three components according to
\begin{align*}
\dv T_n^1&=\{(t,1):x\leq t\leq r_n\}\\
\dv T_n^2&=\{(\chi_\LL(x+t,kt),\eta_\LL(x+t,kt)):t\in(0,r_n-x)\}\\
\dv T_n^3&=\{(\chi_\LL(r_n,t),\eta_\LL(r_n,t)):t\in(0,k(r_n-x)]\}
\end{align*}
See figure \ref{figSequence} and \ref{figImageSeq}. Since all $\{r_n\}_n$ and $x$ are regular points and $W_{\mathcal{L}}^{-1}$ is a homeomorphism it follows  that $\dv T_n^2\cup\dv T_n^3\subset  \dv W_{\mathcal{L}}^{-1}$. Again, since $W_{\mathcal{L}}^{-1}$ is a homeomorphism and $T_n$ is a closed set it follows that $\overline{W_{\mathcal{L}}^{-1}(X_n^{(k)})}\subset T_n$. First note that by assumption on the sequence $\{r_n\}_n$, there exists an $N_1=N_1(\eps)$ such that $r_n-x<\eps$, whenever $n>N_1$. Hence,
\begin{align*}
d((x,1),\dv T_n^1)<\eps
\end{align*}
whenever $n>N_1$. By the assumption that $x$ was regular, there exists an $N_2=N_2(\eps)$, such that
\begin{align*}
d((x,1),\dv T_n^2)<\eps
\end{align*}
whenever $n>N_2$. Finally, as discussed above by the assumption of uniform convergence,
\begin{align*}
d((x,1),\dv T_n^3)<\eps
\end{align*}
for all $n$. Take $N=\max\{N_1,N_2\}$, and take $n>N$ and $l>L(N)$, then 
\begin{align*}
\vert (\chi_\LL(w_l)-x,\eta_\LL(w_l)-1)\vert \leq d((x,1),\dv T_n)< \eps
\end{align*}
Since $\eps>0$ was arbitrary, this implies that $\lim_{l\to \infty}\vert (\chi_\LL(w_l)-x,\eta_\LL(w_l)-1)\vert=0$, and the proof is complete.

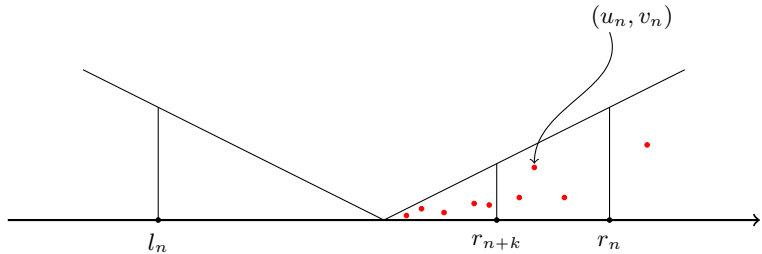
\begin{figure}[H]
\centering
\begin{tikzpicture}
\draw[thick,->] (-5,0) -- (5,0);
\draw (0,0) -- (-4,2);
\draw (0,0) -- (4,2);
\draw (-3,1.5)--(-3,0); 
\draw (3,1.5)--(3,0);
\filldraw[black] (3,0) circle (0.03 cm); 
\filldraw[black] (-3,0) circle (0.03 cm); 
\draw (3,-0.3) node {\small$r_n$};
\draw (-3,-0.3) node {\small$l_n$};
\filldraw[black] (1.5,0) circle (0.03 cm); 
\draw (1.5,-0.3) node {\small$r_{n+k}$};
\draw (1.5,0.75) --(1.5,0);
\filldraw[red] (3.5,1) circle (0.03 cm); 
\filldraw[red] (2.4,0.3) circle (0.03 cm); 
\filldraw[red] (2,0.7) circle (0.03 cm);
\filldraw[red] (1.8,0.3) circle (0.03 cm);
\filldraw[red] (1.4,0.2) circle (0.03 cm);
\filldraw[red] (1.2,0.22) circle (0.03 cm);
\filldraw[red] (0.8,0.1) circle (0.03 cm);
\filldraw[red] (0.5,0.15) circle (0.03 cm);
\filldraw[red] (0.3,0.06) circle (0.03 cm);
\draw[<-] (2,0.75) to [out=90,in=290] (3,2.5);
\draw (3.3,2.7) node {\small$(u_n,v_n)$};
\end{tikzpicture}

\caption{\label{figSequence}The region $X_n^{(k)}$ is depicted above. The red dots represents the positions of the sequence $\{u_n+iv_n\}_{n=1}^{+\infty}$}

\end{figure}

\begin{figure}[H]
\centering
\begin{tikzpicture}
\draw[thick,->] (-5,0) -- (5,0);
\draw (0,0.3) node {\small$(x,1)$};
\filldraw[black] (3,0) circle (0.03 cm); 
\filldraw[black] (-3,0) circle (0.03 cm); 
\draw (-3.15,-1.5) -- (0,0);
\draw (-3.6,-1.5) node {\small$(\chi,\eta)$};
\draw (3,0.3) node {\small$(r_n,1)$};
\draw (1.5,0.3) node {\small$(r_{n+k},1)$};
\draw (-3,0.3) node {\small$(l_n,1)$};
\draw (-3,0) to [out=270,in=45] (-4,-3);
\draw (3,0) to [out=270,in=120] (2.7,-3);
\draw (-4,-3) to [out=35,in=250]  (0,0);
\draw (2.7,-3) to [out=165,in=290] (0,0);
\draw (1.5,0) to [out=270,in=120](1.7,-2.55);
\filldraw[red] (3,-2) circle (0.03 cm); 
\filldraw[red] (2.7,-0.4) circle (0.03 cm); 
\filldraw[red] (2.1,-2.1) circle (0.03 cm);
\filldraw[red] (1.8,-0.3) circle (0.03 cm);
\filldraw[red] (1.3,-0.5) circle (0.03 cm);
\filldraw[red] (1,-0.6) circle (0.03 cm);
\filldraw[red] (0.5,-0.2) circle (0.03 cm);
\filldraw[red] (0.4,-0.3) circle (0.03 cm);
\filldraw[red] (0.2,-0.2) circle (0.03 cm);
\draw[<-] (-1.5,-0.75) to [out=300,in=45] (-0.5,-2);
\draw (-0.5,-2.3) node {\small$\sqrt{(x-\chi)^2+(1-\eta)^2}$};
\end{tikzpicture}

\caption{\label{figImageSeq}Depiction of the set $T_n$. The dots represent the images of the tangential sequence under the homeomorphism $W_{\LL}^{-1}$.} 
\end{figure}
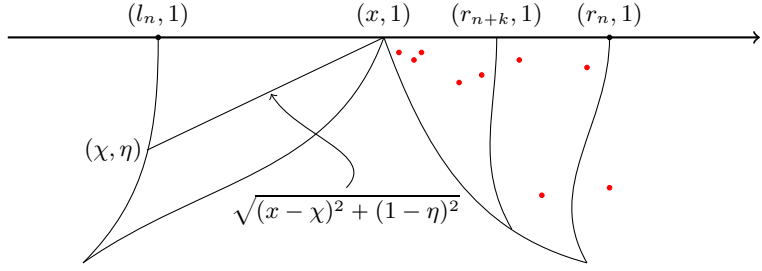
\end{proof}
We now show that all points in the set $G$ defined in Proposition \ref{Topology1} are generic.
\begin{Thm}
\label{GenericThm}
Assume that $x\in G$. Then $\Dl(x)=\{(x,1)\}$.
\end{Thm}
\begin{proof}
Let $x\in G$. Take sequences $\{r_n\}_n\subset G$ and $\{l_n\}_n\subset G$ as in Lemma \ref{Generic1} and $N,\eps>0$ such that 
\begin{align}
\label{GenericH}
\sup_{n,k}\{\sup_{v>0}\vert H_vf(r_n)\vert,\sup_{v>0}\vert H_vf(l_k)\vert\}<N
\end{align}
and
\begin{align}
\label{Genericf}
\min\{\inf_{n}\min\{f(r_n),1-f(r_n)\},\inf_{k}\min\{f(l_k),1-f(l_k)\}\}>\eps.
\end{align}
Assume that $u\in X:=\bigcup_{n}\{r_n\}\bigcup \{x\}\bigcup_{n}\{l_n\}$. Then an integration by parts gives
\begin{align*}
v^{-1}\pi P_v f(u)&=2\int_{-\infty}^{+\infty}\frac{(t-u)}{((t-u)^2+v^2))^2}\int_{u}^{t}f(t')dt'dt\\
&\geq \int_{u}^{+\infty}\frac{(t-u)^2}{((t-u)^2+v^2))^2}(M_Rf(u,t-u)\chi_{t>u})dt.
\end{align*}
Since $u\in X\subset G$, and thus in particular $x\in \mathscr{L}_f\bigcap \{t\in \Sn: 0<f(t)<1\}$, there exists a $\delta=\delta(u)$ such that $\min\{M_Rf(u,t-u),M_R(1-f)(u,t-u)\}\geq \eps/2$ whenever $\vert t-u\vert<\delta$. This implies that 
\begin{align}
\label{GenericP}
\min\{v^{-1}P_v f(u),v^{-1}P_v (1-f)(u)\}\geq \frac{\eps}{2\pi}\int_{u}^{\delta}\frac{(t-u)^2}{((t-u)^2+v^2))^2}dt\to +\infty
\end{align}
as $v\to 0^+$.
We get from (\ref{eqchin}) and (\ref{eqetan})
\begin{align*}
&\vert\sin[\pi P_v f(u)])\vert^2\vert (\chi_\LL(u,v)-u,\eta_\LL(u,v)-1)\vert^2 \\
&=v^2\big(e^{-H_vf(u)}-\cos(\pi P_v f(u)\big)^2+v^2\big(e^{H_vf(u)}+e^{-H_vf(u)}-2\cos(\pi P_v f(u)\big)^2\\
&\leq v^2(1+e^{\vert H_vf(u)\vert})^2+v^2(2e^{\vert H_vf(u)\vert}+2)^2\leq 20v^2e^{2N}v^2,
\end{align*}
by (\ref{GenericH}). Since $\sin (\pi P_v f(u))  \geq \frac{\pi}{2}\min\{P_vf(u),P_v(1-f)(u)\}$, 
\begin{align}
\label{GenericMonotone}
\vert (\chi_\LL(u,v)-u,\eta_\LL(u,v)-1)\vert &\leq \frac{\sqrt{20}ve^{N}}{v\min\{\pi P_v f(u),\pi P_v (1-f)(u)\}}\nonumber\\&=\frac{\sqrt{20}e^{N}}{\min\{v^{-1}\pi P_v f(u),v^{-1}\pi P_v (1-f)(u)\}}:=g_v(u).
\end{align}

We now show that $g_v(u)$ is an increasing function in $v$ for each $u$. We must show that 
\begin{align*}
\min\{v'^{-1}\pi P_{v'}f(u),{v'}^{-1}\pi P_{v'} (1-f)(u)\}<\min\{v^{-1}\pi P_v f(u),v^{-1}\pi P_v (1-f)(u)\}
\end{align*}
is a decreasing function in $v$. Since both $v^{-1}\pi P_{v}f(u)$ and $v^{-1}\pi P_v (1-f)(u)$ are decreasing functions in $v$, Lemma \ref{MinMon} shows that $\min\{v^{-1}\pi P_v f(u),v^{-1}\pi P_v (1-f)(u)\}$ is decreasing. Thus $g_v(u)$ is an increasing function in $v$ for all $x\in X$. 

Note that by (\ref{GenericMonotone}), $g_v(u)\to 0$ as $v\to 0^+$ for all $u\in X$. We now show that $g_v(u)$ is continuous function for all fixed $v$. It is sufficient to show that for some sequence $\{r_{m_n}\}_n$ such that $\lim_{n\to \infty}r_{m_n}=x$ we have $\lim_{n\to \infty} g_v(r_{m_n})=g_v(x)$. However, this follows immediately from the fact that $P_v(u)$ is a continuous function on $\Hp$. Since $X$ is compact in the subspace topology from $\R$ and $0$ is a continuous function on $X$, it follows by Dini's theorem that $g_v(u)\to 0$ as $v\to 0^+$ uniformly on $X$. The estimate (\ref{GenericMonotone}) then shows that $(\chi_\LL(u,v),\eta_\LL(u,v))\to (u,1)$ uniformly on $X$. Hence, by Lemma \ref{lemUniformconv}, it follows that $\dv\LL(x)=\{(x,1)\}$. 
\end{proof}

\subsection{Sufficient Conditions for Points to be Generic}
We now give a proposition, which provides a sufficient condition for when the assumptions of Lemma \ref{lemUniformconv} are satisfied. 
\begin{Prop}
\label{GenericProp}
Assume that $x\in\Sn^{\circ}$, and that the density $f$ at $x$ satisfies the assumptions of Proposition \ref{NLpointReg}, so that $x$ is regular. Furthermore, assume that there exists a $\delta>0$, and sequences of regular points $\{r_n\}_n$ and $\{l_n\}$, such that $r_n>x$ and $l_n<x$ for all $n$, and such that $\lim_{n\to\infty}r_n=\lim_{n\to\infty}l_n=x$. Furthermore, assume that 
\begin{align}
\label{GenericCond1}
\max\{\sup_{n}\Delta m_f(r_n),\sup_{n}\Delta m_f(l_n)\}=m_1<1
\end{align}
and for some $\delta>0$
\begin{align}
\label{GenericCond2}
\min\{\sup_n \max\{1-m_f^{\delta}(r_n),1-m_{1-f}^{\delta}(r_n)\},\sup_n \max\{1-m_f^{\delta}(l_n),1-m_{1-f}^{\delta}(l_n)\}\}=m_2>0.
\end{align}
Then $\lim_{n\to \infty} (\chi_\LL(r_n,v),\eta_\LL(r_n,v))=(x,1)$ and $\lim_{n\to \infty} (\chi_\LL(l_n,v),\eta_\LL(l_n,v))=(x,1)$ uniformly, and by Lemma \ref{lemUniformconv} it follows that $\dv \mathcal{L}(x)=\{(x,1)\}$.
\end{Prop}
\begin{proof}
We may assume that $v<\delta$. By (\ref{GenericCond1}) and (\ref{GenericCond2}) and Lemma \ref{Ineq2}, we have
\begin{align*}
\max\{\sup_n\vert (\chi_{\LL}(r_n,v)-r_n,\eta_{\LL}(r_n,v)-1)\vert,\sup_n\vert (\chi_{\LL}(l_n,v)-r_n,\eta_{\LL}(l_n,v)-1)\vert\}\leq \frac{2\sqrt{20}cv^{1-m_1}}{m_2}
\end{align*}
This completes the proof.
\end{proof}
\begin{ex} 
Consider the density in Example \ref{NRegEx0}. We had  $f(t)=\frac{1}{2}\vert \sin t^{-1}\vert^{1/2}\chi_{[-a,a]}(t)$, where $a$ was chosen so that $\int_{\R}f(t)dt=1$. In Example \ref{NRegEx0}, we showed that $f$ satisfied the assumptions of Proposition \ref{NLpointReg} at $x=0$, which implied that $0$ was a regular point. We now show that $0$ is also a generic point.
Since $0\leq f(t)\leq 1/2$ for all $t\in[-a,a]$, it follows that $\Delta m_f(x)\leq 1/2$ and $Mf(x,h)\leq 1/2$ for all $x$ and $h>0$. Choose $r_n=(\pi/2+\pi n)^{-1}$ and $l_n=(\pi/2-\pi n)^{-1}$ for $n>0$. Fix any $h>0$ and let $m$ be the integer such that $\pi/2+\pi(m-1)<(r_n+h)^{-1}<\pi/2+\pi m$. In particular we note that $m\leq n$. We now estimate $Mf(r_n,h) (=M(l_n,h))$ from below. If $m<n$, then
\begin{align*}
Mf(r_n,h)&=\frac{1}{4h}\int_{r_n-h}^{r_n+h}\vert \sin t^{-1}\vert^{1/2}dt>\frac{1}{4h}\int_{r_n}^{r_n+h}\vert \sin t^{-1}\vert^{1/2}dt=\frac{1}{4h}\int_{(r_n+h)^{-1}}^{r_n^{-1}}\vert \sin s\vert^{1/2}\frac{ds}{s^2}\\
&>\frac{1}{4(r_m-r_n)}\int_{r_m^{-1}}^{r_n^{-1}}\vert \sin s\vert^{1/2}\frac{ds}{s^2}\geq \frac{1}{4(r_m-r_n)}\sum_{k=m}^{n-1}\frac{1}{r_{k+1}^{-2}}\int_{r_k^{-1}}^{r_{k+1}^{-1}}\vert \sin s\vert^{1/2}ds\\
&\geq \frac{1}{4(r_m-r_n)}\sum_{k=m}^{n-1}\frac{1}{r_{k+1}^{-2}}\int_{0}^{\pi/2}\sqrt{\pi^{-1}s}ds=\frac{1}{8(r_m-r_n)}\sum_{k=m}^{n-1}\frac{1}{(\pi/2+\pi (k+1))^2}\\
&\geq \frac{1}{8}\frac{\pi mn}{(n-m)}\int_{m}^{n}\frac{dx}{\pi^{2}(x+3)^2}\geq \frac{1}{4\pi}\frac{ mn}{(n-m)}\frac{(n-m)}{(n+3)(m+3)}=\frac{1}{4\pi}\frac{nm}{(n+3)(m+3)}\\
&\geq \frac{1}{128\pi}.
\end{align*}
If on the other hand $m=n$, then if $\frac{\pi/2}{\pi n(\pi n+\pi/2)}<h\leq \frac{\pi}{\pi^2n^2+\pi^2/4}$ then
\begin{align*}
Mf(r_n,h)&>\frac{1}{4h}\int_{(r_n+h)^{-1}}^{r_n^{-1}}\vert \sin t\vert^{1/2}dt>\frac{\pi^2n^2+\pi^2/4}{4\pi}\int_{\pi n}^{\pi/2+\pi n}\frac{\vert t-\pi n \vert^{1/2}}{t^2}dt\\& >\frac{\pi^2n^2+\pi^2/4}{4\pi}\frac{1}{(\pi n+\pi/2)^2}\int_{\pi n}^{\pi/2+\pi n}\pi^{-1/2}\vert t-\pi n \vert^{1/2}dt>\bigg(\frac{1}{4\pi}+o(1)\bigg)\frac{1}{2}\\\
&=\frac{1}{8\pi}+o(1)>0
\end{align*}
and if  then $0<h<\frac{\pi/2}{\pi n(\pi n+\pi/2)}$, then
\begin{align*}
Mf(r_n,h)&>\frac{1}{4h}\int_{(r_n+h)^{-1}}^{r_n^{-1}}\vert \sin t\vert^{1/2}dt>\frac{\pi^{-1/2}}{4h}\int_{(r_n+h)^{-1}}^{r_n^{-1}}\frac{\vert t-\pi n \vert^{1/2}}{t^2}dt\\& >\frac{\pi^{-1/2}}{4h}\bigg[\frac{\arctan{\sqrt{1-t/(\pi n)}}}{\sqrt{\pi n}}-\frac{\sqrt{t-\pi n}}{t}\bigg]_{(r_n+h)^{-1}}^{r_n^{-1}}\\&=\frac{\pi^{-1/2}}{4h}\bigg\{(r_n+h)\sqrt{(r_n+h)^{-1}-\pi n)}-r_n\sqrt{r_n^{-1}-\pi n}\bigg\}+o(1)>\frac{\pi^{-1/2}}{4h}\frac{h\sqrt{\pi}}{\sqrt{2}}+o(1)=\frac{1}{4\sqrt{2}}+o(1)>0
\end{align*}
Hence, $f$ satisfies the assumptions of Proposition \ref{GenericProp} at $x=0$, which implies that $0$ is a generic point.
\end{ex}
\begin{ex}
Consider the density in Example \ref{NRegEx2}, where 
\begin{align*}
f(t)=\sum_{k=1}^{+\infty}\bigg(1-\frac{1}{2k}\bigg)(\chi_{I_{2k}}(t)+\chi_{I_{2k}}(-t))+\sum_{k=1}^{+\infty}\frac{1}{2k+1}(\chi_{I_{2k+1}}(t)+\chi_{I_{2k+1}}(-t))+\chi_{(1/2,a]}(t),
\end{align*}
and where $I_n=(2^{-(n+1)},2^{-n}]$, and where $a$ is chosen so that $\int_{\R}f(t)dt=1$. If for any sequence $\{r_n\}_n$ we have that $r_n\in (2^{-(n+1)},2^{-n})$, then
for any $h>0$ sufficiently small we have
\begin{align*}
Mf(r_{2n},h)=1-\frac{1}{2n}
\end{align*}
and 
\begin{align*}
Mf(r_{2n+1},h)=\frac{1}{2n+1}.
\end{align*}
This implies that $\inf_n{Mf(r_{n},h),M(1-f)(r_{n},h)}=0$. Therefore condition (\ref{GenericCond2}) is violated. If we on the other hand choose $r_n=2^{-n}$, then $h>0$ for sufficiently small 
\begin{align*}
\vert \Delta Mf(r_{n},h)\vert \geq 1-\frac{1}{n}.
\end{align*}
Hence, $\sup_n{\vert \Delta Mf(r_{n},h)\vert}=1$ which violates condition (\ref{GenericCond1}). Therefore the assumptions of Proposition \ref{GenericProp} are never satisfied and we are unable to conclude that $\Dl(0)=\{(0,1)\}$ by the methods developed in this paper. This is due to the fact that the estimate of Lemma \ref{EstHilb2} is too rough. 
\end{ex}

\begin{Prop}
\label{HoldeReg}
Assume that $x\in\mathscr{L}_f\cap \Sn^{\circ}$, $f(x)=0$, $\int_{\R}(x-t)^{-2}f(t)dt=+\infty$ and that $H_v(u)$ is tangentially bounded at $x$. If in addition we assume that $x\in\mathscr{L}_{\Hil f}$, then $\dv \mathcal{L}(x)=\{(x,1)\}$. In particular, this holds when $f\in \Hc$.
\end{Prop}
\begin{proof}
By Proposition \ref{RegPoint1}, $x$ is regular. Using that the set $G$ is dense in $\Sn^{\circ}$, an analogous argument as in Lemma \ref{Generic1} shows that there exists sequences $\{r_n\}_n\subset G$ and $\{l_n\}_n\subset G$, satisfying the same properties as in Lemma \ref{Generic1}, except that $\inf_{n,k}\{f(r_n),f(l_k)\}=0$, (since $f(x)=0$ by assumption). Similarly to the proof of Theorem \ref{GenericThm} one can show that $\lim_{v\to 0^+}\min\{v\pi P_v\ast f(u),v\pi P_v\ast (1-f)(u)\}=+\infty$ for all $u\in X$, where $X$ is as in the proof of Theorem \ref{GenericThm}. Therefore, a completely analogous proof gives that $\Dl(x)=\{(x,1)\}$.
\end{proof}

\begin{Thm}
\label{ExtendedMap}
Assume that for all $x\in\Sn$, $\dv\mathcal{L}(x)=\{(x,1)\}$. Then the homeomorphism $W_{\LL}:\LL\to \Hp $ extends to a homeomorphism $\overline{W}_{\LL}:\overline{\LL}\to \overline{\Hp} $
\end{Thm}
\begin{proof}
Recall Theorem 2.3 in \cite{Duse14a}, which states that the map $W_{\EE}^{-1}: R\to \EE$ is a bijective real analytic parametrization of $\EE\subset \dv \LL$ such that $W_{\EE}^{-1}(R)=\EE$ and $W_{\EE}^{-1}(t)=(\chi_{\EE}(t),\eta_{\EE}(t))$. Moreover, for every $\{w_n\}_n\subset \Hp$ such that $w_n\to t$ as $n\to \infty$, $\lim_{n\to \infty}(\chi_{\LL}(w_n),\eta_{\LL}(w_n))=(\chi_{\EE}(t),\eta_{\EE}(t))$. Recall that $R$ is an open set. Let $x\in \dv R$. Fix $\eps>0$ and let $\{w_n\}_n$ be any sequence such that $u_n=\text{Re}[w_n]\in R$ for all $n$. By choosing $\text{Im}[w_n]=v_n$ sufficiently small, we can assume that $\vert (\chi_{\LL}(w_n)-\chi_{\EE}(u_n),\eta_{\LL}(w_n)-\eta_{\EE}(u_n))\vert <\eps$ for all $n>N$ say. Since $\eps>0$ was arbitrary, the assumption that $\dv\mathcal{L}(x)=\{(x,1)\}$ implies that $\lim_{\substack{t\to x\\ t\in R}}(\chi_{\EE}(t),\eta_{\EE}(t))=(x,1)$. It follows that $\dv \LL=\EE\bigcup (\Sn\times \{1\})$. Moreover, the map $W_{\dv \LL}^{-1}:\R\to \dv \LL$ defined through
\begin{align*}
W_{\dv \LL}^{-1}(t)=\left\{
 \begin{array}{ll}
   W_{\EE}^{-1}(t)& \text{if  $t\in R$} \\
   (t,1) & \text{if $t\in \Sn$} 
 \end{array} \right.
\end{align*}
is continuous and injective. Hence $\dv \LL$ is a simple curve. This implies that the homeomorphism $W_{\LL}:\LL\to \Hp $ extends to a homeomorphism $\overline{W}_{\LL}:\overline{\LL}\to \overline{\Hp} $.
\end{proof}

\subsection{Generic Points of $\Siso$ and The Edge $\mathcal{E}$}
In this section we will consider the limits $\lim_{t\rightarrow x}\chi_{\mathcal{E}}(t)$ and $\lim_{t\rightarrow x}\eta_{\mathcal{E}}(t)$ when $x\in \Siso$ and $x$ is a regular point. In particular we want to show that in many cases $\lim_{t\rightarrow x}(\chi_{\mathcal{E}}(t),\eta_{\mathcal{E}}(t))=(x,1)$ and that $\dv\mathcal{L}(x)=\{(x,1)\}$. Recall that the parametrization (\ref{eqchiEEetaEE}) of the edge $\mathcal{E}$ is given by
\begin{align*}
\chi_{\mathcal{E}}(t)&=t+\frac{1-e^{-C(t)}}{C'(t)}\\
\eta_{\mathcal{E}}(t)&=1+\frac{e^{C(t)}+e^{-C(t)}-2}{C'(t)}
\end{align*}
for $t\in R_{\mu}$,  where
\begin{align*}
C(t)=\int_{\R}\frac{f(y)dy}{t-y}.
\end{align*}
\begin{Prop}
\label{Edge1}
Assume that $x\in \Siso$ and that $x$ belongs to the Lebesgue set of $f$. Furthermore assume that $ \mathcal{H}f(x)$ exists, $\vert \mathcal{H}f(x)\vert<+\infty$ and that $\int_{\R}\frac{f(y)dy}{(x-y)^2}=+\infty$ if $x\in \dv\mathcal{S}_R^0\cup \dv\mathcal{S}_L^0$ and $\int_{\R}\frac{(1-f(y))dy}{(x-y)^2}=+\infty$ if $x\in \dv\mathcal{S}_R^1\cup \dv\mathcal{S}_L^1$. Then
\begin{align*}
\lim_{\substack{t\rightarrow x\\ t\in R}}(\chi_{\mathcal{E}}(t),\eta_{\mathcal{E}}(t))&=(x,1)
\end{align*}
and
\begin{align*}
\dv \LL(x)=\{(x,1)\}.
\end{align*}
In particular, this holds if $f\in\Hc$.
\end{Prop}

\begin{proof}
We consider the case where $x\in\dv\mathcal{S}_R^0$, in which case $f(x)=0$. The other cases follow similarly. Since $x\in\dv S_r^0$, there exists a $\delta>0$, such that $f\equiv 0$ for all $x\in[x,x+\delta]$. By definition,
\begin{align*}
\mathcal{H}f(x)&=\lim_{\eps\to 0^+}\int_{-\infty}^{x-\eps}\frac{f(t)dt}{x-t}+\int_{x+\delta}^{+\infty}\frac{f(t)dt}{x-t}\\
&=\int_{-\infty}^{x}\frac{f(t)dt}{\vert x-t\vert }-\int_{x+\delta}^{+\infty}\frac{f(t)dt}{\vert x-t\vert}.
\end{align*}
Since $\int_{x+\delta}^{+\infty}\frac{f(t)dt}{\vert x-t\vert}<+\infty$ and $\vert \mathcal{H}f(x)\vert<+\infty$ it follows that $\int_{\R}\frac{f(t)dt}{\vert x-t\vert}<+\infty$. By Lebesgue's dominated convergence theorem it follows that $\lim_{t\to x^+}C(t)=\mathcal{H}f(x).$ 
By assumption, $\int_{\R}\frac{f(y)dy}{(x-y)^2}=+\infty$. Hence, by Fatou's lemma  
\begin{align*}
\lim_{t\rightarrow x^+}-C'(t)&=\lim_{t\rightarrow x^+}\int_{\R}\frac{f(y)dy}{(t-y)^2}=+\infty.
\end{align*}
Hence,
\begin{align*}
&\lim_{t\rightarrow x^+}\chi_{\mathcal{E}}(t)=\lim_{t\rightarrow x^+}t+\frac{1-e^{-C(t)}}{C'(t)}=x\\
&\lim_{t\rightarrow x^+}\eta_{\mathcal{E}}(t)=\lim_{t\rightarrow x}1+\frac{e^{C(t)}+e^{-C(t)}-2}{C'(t)}=1.
\end{align*}
Now, assume that $\{w_n\}_n\in\Hp$ is a sequence such that $\lim_{n\to\infty}w_n=x$, and such that $\text{Re}[w_n]=u_n\geq x$. Since $x\in \dv\mathcal{S}_R^0$ by assumption, there exists a $\delta>0$, such that $f(t)=0$ for $t\in(x,x+2\delta)$ and consequently, since $\vert \mathcal{H}f(x)\vert<+\infty$, we have $\int_{\R}\frac{f(t)dt}{\vert x-t\vert}<+\infty$. We may assume that $u_n\leq x+\delta$ for all $n$. Using that
\begin{align*}
\frac{\vert u_n-t\vert}{(u_n-t)^2+v_n^2}\leq \frac{1}{\vert x-t\vert}
\end{align*}
for $t\in(-\infty,x]$ which implies that
\begin{align*}
\int_{-\infty}^{x}\frac{\vert u_n-t\vert f(t)dt}{(u_n-t)^2+v_n^2}\leq \int_{-\infty}^{x}\frac{f(t)dt}{\vert x-t\vert},
\end{align*}
and the fact that
\begin{align*}
\sup_{n>0}\bigg\{\int_{x+2\delta}^{\infty}\frac{\vert u_n-t\vert f(t)dt}{(u_n-t)^2+v_n^2}\bigg\}=C<+\infty,
\end{align*}
we conclude that $\pi \vert H_{v_n}f(u_n)\vert\leq  \pi \vert \mathcal{H}f(x)\vert+C$. A similar estimate as in Lemma \ref{Ineq2} gives
\begin{align*}
\vert (\chi_\LL(w_n)-u_n,\eta_\LL(w_n)-1)\vert\leq \frac{2\sqrt{5}(1+e^{\pi\vert \mathcal{H}f(x)\vert+C})}{\int_{\R}\frac{f(t)dt}{(u_n-t)^2+v_n^2}}
\end{align*}
Since 
\begin{align*}
\int_{\R}\frac{f(t)dt}{(u_n-t)^2+v_n^2}\geq \min\bigg\{\int_{\R}\frac{f(t)dt}{2(u_n-t)^2},\int_{\R}\frac{f(t)dt}{2v_n^2}\bigg\},
\end{align*}
and
\begin{align*}
\lim_{n\to \infty}\min\bigg\{\int_{\R}\frac{f(t)dt}{2(u_n-t)^2},\int_{\R}\frac{f(t)dt}{2v_n^2}\bigg\}=+\infty,
\end{align*}
$\lim_{n\to\infty}(\chi_\LL(w_n),\eta_\LL(w_n))=(x,1)$. Now assume that $\{w_n\}_n\subset\Hp$ such that $\lim_{n\to\infty}w_n=x$ and $\text{Re}[w_n]\leq x$. By Lemma \ref{DiniConv}, $H_vf(u)$ is non-tangentially bounded at $x$. An identical argument as in Proposition \ref{HoldeReg} shows that $\lim_{n\to\infty}(\chi(w_n),\eta(w_n))=(x,1)$. Hence $\dv\mathcal{L}(x)=\{(x,1)\}$. 
\end{proof}
\begin{Prop}
\label{Edge2}
Assume that $x\in \Siso$ and that $x\notin \mathscr{L}_f$. Furthermore assume that: 
\begin{itemize}
\item
If $x\in\dv \mathcal{S}_R^0$ then $0<f_L^-(x)\leq f_L^+(x)<1$.
\item
If $x\in\dv\mathcal{S}_L^0$ then $0<f_R^-(x)\leq f_R^+(x)<1$.
\item
If $x\in\dv\mathcal{S}_R^1$ then $0<f_L^-(x)\leq f_L^+(x)<1$.
\item
If $x\in\dv\mathcal{S}_L^1$ then $0<f_R^-(x)\leq f_R^+(x)<1$.
\end{itemize}
Then, 
\begin{align*}
\lim_{\substack{t\rightarrow x\\ t\in R}}(\chi_{\mathcal{E}}(t),\eta_{\mathcal{E}}(t))=(x,1).
\end{align*}
Moreover, if $x\in \dv\mathcal{S}_R^0\cup\dv\mathcal{S}_R^1$, then for every sequence $\{w_n\}_n\in\Hp$, such that  $\lim_{n\to\infty}w_n=x$ and $\text{Re}[w_n]\geq x$
\begin{align*}
\lim_{n\to \infty}(\chi_\LL(w_n),\eta_\LL(w_n))=(x,1).
\end{align*}
If on the other-hand $x\in \dv\mathcal{S}_L^0\cup\dv\mathcal{S}_L^1$, then for every sequence $\{w_n\}_n\in\Hp$, such that  $\lim_{n\to\infty}w_n=x$ and $\text{Re}[w_n]\leq x$
\begin{align*}
\lim_{n\to \infty}(\chi_\LL(w_n),\eta_\LL(w_n))=(x,1).
\end{align*}
Finally, if we assume that $f$ satisfies the right-sided or left-sided versions of the conditions of Proposition \ref{GenericProp}, then $\dv \LL(x)=\{(x,1)\}$.
In particular this holds when $f\in \Hc$.
\end{Prop}
\begin{proof}
Consider the case where $x\in\dv\mathcal{S}_R^0$. The other cases follow similarly. Then, by assumption there exists a $\delta>0$ such that $f(t)=0$ for $t\in(x,x+2\delta)$. Moreover, since $0<f_L^-(x)\leq f_L^+(x)<1$ by assumption, there exists an $\eps>0$ such that
\begin{align*}
M_Lf(x,h)<f_L^{+}(x)+\eps<1 
\end{align*}
and
\begin{align*}
M_Lf(x,h)>f_L^{-}(x)-\eps>0 
\end{align*}
for all $0<h<H$, for some $H>0$. Let $p^+=f_L^{+}(x)+\eps$ and $p^-=f_L^{-}(x)-\eps$. Then
\begin{align*}
\vert C(t)\vert &= \int_{-\infty}^{x}\frac{f(y)dy}{t-y}+\int_{x+2\delta}^{+\infty}\frac{f(y)dy}{y-t}\\
&\leq  \int_{-\infty}^{x}\frac{f(y)dy}{t-y}+c_1
\end{align*}
for all $t\in(x,x+\delta)$, where $c_1>0$ is a constant. An integration by parts gives
\begin{align*}
\int_{-\infty}^{x}\frac{f(y)dy}{t-y}&=\int_{-\infty}^{x}\frac{1}{(t-y)^2}\int_{y}^{x}f(s)dsdy\\
&=\int_{-\infty}^{x}\frac{(x-y)}{(t-y)^2}M_Lf(x,x-y)dy\\
&\leq \int_{-\infty}^{x-H}\frac{(x-y)}{(x-y)(t-y)^2}dy+\int_{x-H}^{x}\frac{(x-y)}{(t-y)^2}dy\\
&\leq \int_{x-H}^{x}\frac{1}{(t-y)}dy+c_2\\
&\leq -p^+\log(t-x)+c_3,
\end{align*}
where $c_2$ and $c_3$ are positive constants independent of $t$, and where we have used the inequality $M_Lf(x,x-y)\leq \frac{1}{x-y}$ for all $y<x$ from Lemma \ref{DeltaMean}. Thus
\begin{align*}
\vert C(t)\vert \leq -p^+\log(t-x)+c_1+c_3.
\end{align*}
We now estimate $C'(t)$. An integration by parts yields 
\begin{align*}
\vert C'(t)\vert &\geq \int_{-\infty}^{x}\frac{f(y)dy}{(t-y)^2}
=\bigg[\frac{-1}{(t-y)^2}\int_{y}^{x}f(s)ds\bigg]_{-\infty}^{x}+\int_{-\infty}^{x}\frac{2}{(t-y)^3}\int_{y}^{x}f(s)dsdy\\
&=2\int_{-\infty}^{x}\frac{(x-y)}{(t-y)^3}M_Lf(x,x-y)dy\geq 2p^-\int_{x-H}^{x}\frac{(x-y)}{(t-y)^3}dy
=p^-\bigg[\frac{t+x-2y}{(t-y)^2}\bigg]_{x-H}^{x}\\
&\geq\frac{p^-}{t-x}-p^-\frac{t-x+2H}{(t-x+H)^2}.
\end{align*}
Note that $\displaystyle\frac{p^-}{t-x}-p^-\frac{t-x+2H}{(t-x+H)^2}>0$ whenever $\vert t-x\vert$ sufficiently small. This implies that
\begin{align*}
\frac{1+e^{\vert C(t)\vert}}{\vert C'(t)\vert}\leq \frac{1+(t-x)^{-p^+}e^{C'}}{\displaystyle\frac{p^-}{t-x}-p^-\frac{t-x+2H}{(t-x+H)^2}}=\frac{1}{p^-+O((t-x))}((t-x)+(t-x)^{1-p^+}e^{C'})\to 0
\end{align*}
as $t\rightarrow x^+$. This proves that $\lim_{t\to x^+}(\chi_{\mathcal{E}}(t),\eta_{\mathcal{E}}(t))=(x,1)$. Similarly, one can show that
\begin{align*}
\vert (\chi_{\LL}(t,v)-t,\eta_\LL(t,v)-1)\vert \leq \frac{d}{p^-+O((t-x))}((t-x)+(t-x)^{1-p^+}e^{C'})
\end{align*}
for some constant $d$ independent of $t$. This shows that
\begin{align*}
\lim_{n\to \infty}(\chi_\LL(w_n),\eta_\LL(w_n))=(x,1)
\end{align*}
for every sequence $\{w_n\}_n\in\Hp$, such that  $\lim_{n\to\infty}w_n=x$ and $\text{Re}[w_n]\geq x$.
The rest of the proof is analogous to the proof of Proposition \ref{GenericProp}.
\end{proof}
We conclude this section with the following conjecture:
\begin{Con}
Assume that $x\in \Siso$ and that $x$ is a regular point. Then $(x,1)\in\overline{\mathcal{E}}$ and $\dv\mathcal{L}(x)=\{(x,1)\}$. 
\end{Con}

\section{Singular Points}
\subsection{Sufficient Conditions for the Existence of Singular Points of the Non-Trivial Support}
For $x\in \mathcal{S}_{nt}(\mu)$, we recall that $x$ is a singular point of the support of $\mu$ if there exists a non-tangential sequence $\{w_n\}_{n=1}^{+\infty}\subset \Hp$ such that $\lim_{n\rightarrow +\infty}w_n=x$ but $\lim_{n\rightarrow \infty}(\chi_\LL(u_n,v_n),\eta_\LL(u_n,v_n))\neq(x,1)$. Consequently this means that $\dv \LL(x)$ contains more points than $(x,1)$.
For the types of singular points considered in this article, the limit $\lim_{n\rightarrow +\infty}(\chi_\LL(w_n),\eta_\LL(w_n))$ will exist for every sequence $\{w_n\}_n$ that converges non-tangentially to $x$ and be independent of the non-tangential sequence chosen. Let the limit be $(\chi_\G(x),\eta_\G(x))$. 
\begin{Lem}
\label{TangTop}
For every $x\in \Sn^{\circ}$ there exists a sequence $\{w_n\}_{n=1}^{\infty}$ such that
\begin{align*}
\lim_{n\rightarrow \infty}(\chi_\LL(u_n,v_n),\eta_\LL(u_n,v_n))=(x,1).
\end{align*}
In particular this implies that $\{(x,1)\}\subset \dv\LL(x)$ for all $x\in \Sn^{\circ}=(\text{supp}(\mu)\cap \text{supp}(\lambda-\mu))^{\circ}$.
\end{Lem}
\begin{proof}
According to Proposition \ref{Topology1} the set of regular points is dense in $x\in \Sn^{\circ}$. Therefore, for every $x$ there exists a sequence $\{u_n\}_{n=1}^{\infty}$ such that $\lim_{n\rightarrow \infty}u_n=x$ and such that $u_n$ is regular for every $n$. Hence, for every $\eps>0$ sufficiently small and each $n$ there exists a $v_n$ such that
\begin{align*}
\vert (\chi_\LL(u_n,v_n),\eta_\LL(u_n,v_n))-(u_n,1)\vert<\frac{\eps}{n}
\end{align*}
This implies that
\begin{align*}
\lim_{n\rightarrow \infty}(\chi_\LL(u_n,v_n),\eta_\LL(u_n,v_n))=(x,1)
\end{align*}
which concludes the proof.
\end{proof}
From this and the fact that $\LL$ is simply connected and $W_\LL$ is a homeomorphism, $(\chi_\G(x),\eta_\G(x))$ has to be connected to the point $(x,1)$. But since for all non-tangentially convergent sequences $\{w_n\}_n$ to $x$, $\lim_{n\rightarrow +\infty}(\chi_\LL(w_n),\eta_\LL(w_n))=(\chi_\G(x),\eta_\G(x))$, we will have to consider tangentially convergent sequences to $x$ in order to determine the whole of $\dv \LL(x)$. We will not attempt to determine $\dv \LL(x)$ in full generality, nor will we attempt a complete classification of all singular points, but contend ourselves with some more restrictive assumptions on the density $f$.

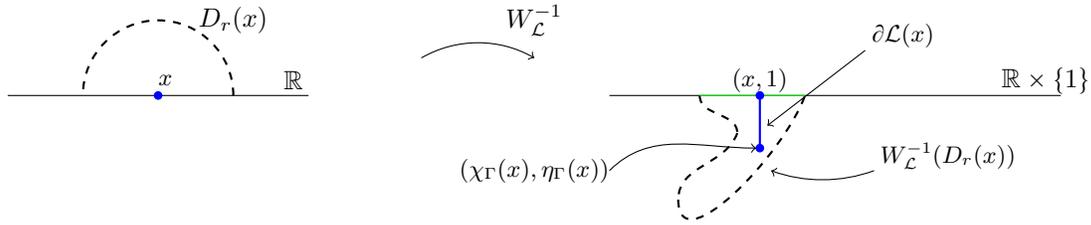
\begin{figure}[H]
\centering
\begin{tikzpicture}

\draw (-10,0) -- (-6,0); 
\draw[dashed,thick] (-7,0) arc (0:180:10 mm);
\draw (-7,1) node {$D_r(x)$};
\draw (-6.2,0.2) node {$\mathbb{R}$};
\filldraw[blue] (-8,0) circle (0.5 mm);
\draw (-7.9,0.2) node {\small$x$};
\draw (-2,0) --(4,0);
\draw (3.8,0.2) node {$\mathbb{R}\times\{1\}$};
\draw[<-] (-3,0.5) arc (60:120:15 mm);
\draw (-3,1) node {$W^{-1}_{\mathcal{L}}$};
\draw[blue,thick] (0,0) --(0,-0.7);
\filldraw[blue] (0,-0.7) circle (0.5mm);
\draw[dashed,thick] (-0.8,0) to[out=270,in=90] (-0.3,-0.5);
\draw[dashed,thick] (-0.3,-0.5) to[in=90,out=250] (-0.7,-1,1);
\draw[dashed,thick] (-0.7,-1,1) to[in=250,out=270]  (0.6,0);
\draw (0,0.2) node {\small$(x,1)$};
\draw[<-] (0.15,-1) arc (250:290:20 mm);
\draw (2.5,-0.8) node {\small$W^{-1}_{\mathcal{L}}(D_r(x))$};
\draw[<-] (0.1,-0.4) -- (1.4,0.6);
\draw (1.9,0.8) node {\small$\partial \mathcal{L}(x)$};
\draw[green] (-0.8,0) -- (0.6,0);
\filldraw[blue] (0,0) circle (0.5 mm);
\draw[->] (-2,-1) to [out=45,in=180] (-0.05,-0.7);
\draw (-3,-1) node {\small$(\chi_{\G}(x),\eta_{\G}(x)$)};
\end{tikzpicture}
\caption{\label{figcontours1}Depiction of a singular point.}
\end{figure}

Recall from the introduction that $\mathcal{S}_{nt}^{sing,I}(\mu)$ are the set of points $x\in \Sn$ such that for some $\delta>0$, $f(t)=\chi_{[x-\delta,x]}(t)+\varphi(t)$,
and such that $x$ belongs to the Lebesgue set of $\varphi$ and $\varphi(x)=0$ and $\vert \mathcal{H}\varphi(x)\vert<+\infty$. 
\begin{Prop}
\label{Sing1}
Assume that  $x\in\mathcal{S}_{nt}^{sing,I}(\mu)$. Then $x$ is a singular point, and
\begin{align}
\label{eqSing1NTangLim}
\lim_{n\rightarrow +\infty}(\chi_\LL(u_n,v_n),\eta_\LL(u_n,v_n))=(x,1-\delta e^{\pi\mathcal{H}\varphi(x)}):=(\chi_\G^{I}(x),\eta_\G^{I}(x))
\end{align}
for every non-tangential sequence $\{u_n+iv_n\}_{n=1}^{\infty}\in \Hp$ such that $\lim_{n\rightarrow +\infty}u_n+iv_n=x$.
\end{Prop}
\begin{proof}
We first note that the fact that $\varphi(t)\leq 0$ for $t\in(x-\delta,x)$ and $\varphi(t)\geq 0$ otherwise, together with the assumptions that $\varphi(x)=0$ and $\vert \mathcal{H}\varphi(x)\vert<+\infty$, imply that 
\begin{align}
\label{DiniCond}
\text{p.v.}\int_{x-\delta}^{x+\delta}\frac{ \varphi(t)dt}{t-x}=\int_{x-\delta}^{x+\delta}\frac{\vert \varphi(t)-\varphi(x)\vert dt}{\vert t-x\vert}<+\infty.
\end{align}
Now assume that $\{u_n+iv_n\}_n$ converges non-tangentially to $x$. Then,
\begin{align*}
\pi H_{v_n}f(u_n)&=\int_{x-\delta}^{x}\frac{(u_n-t)dt}{(u_n-t)^2+v_n^2}+\int_{\R}\frac{(u_n-t)\varphi(t)dt}{(u_n-t)^2+v_n^2}\\
&:=I_n+\pi H_{v_n}\varphi(u_n)
\end{align*}
By (\ref{DiniCond}) and Proposition \ref{DiniConv}, it follows that
\begin{align*}
\lim_{n\to\infty}\pi H_{v_n}\varphi(u_n)=\pi\mathcal{H}\varphi(x).
\end{align*}
Computing $\log v_n+I_n$ yields,
\begin{align*}
\log v_n+I_n&=\log v_n-\log\sqrt{(u_n-x)^2+v_n^2}+\log\sqrt{(u_n-x+\delta)^2+v_n^2}\\&=-\log\sqrt{1+(u_n-x)^2/v_n^2}+\log\sqrt{(u_n-x+\delta)^2+v_n^2}
\end{align*}
We now consider $\pi P_{v_n} f (u_n)$. A computation gives
\begin{align*}
\pi P_{v_n} f (u_n)&=\int_{x-\delta}^{x}\frac{v_ndt}{(u_n-t)^2+s^2}+\int_{\R}\frac{v_n\varphi(t)dt}{(u_n-t)^2+s^2}\\
&=-\arctan\frac{u_n-x}{v_n}+\arctan\frac{u_n-x+\delta}{v_n}+\pi P_{v_n} \varphi(u_n)
\end{align*}
Furthermore, since $x$ is in the Lebesgue set of $\varphi$, (see page 11), 
\begin{align*}
\lim_{n\to\infty}\pi P_{v_n} \varphi (u_n)&=\pi\varphi(x)=0.
\end{align*}
Thus,
\begin{align*}
\pi P_{v_n} f (u_n)&=\frac{\pi}{2}-\arctan\frac{u_n-x}{v_n}+o(1).
\end{align*}
Since $\sin\big(\frac{\pi}{2}-\arctan\alpha\big)=\frac{1}{\sqrt{1+\alpha^2}}$,
we get
\begin{align*}
&\lim_{n\to \infty}\frac{v_n(e^{\pi H_{v_n}f(u_n)}+e^{-\pi H_{v_n}f(u_n)}-2\cos(\pi P_{v_n} f(u_n))}{\sin(\pi P_{v_n} f(u_n))}
\\&=\lim_{n\to \infty}\frac{\sqrt{(u_n-x+\delta)^2+v_n^2}}{\sqrt{1+(u_n-x)^2/v_n^2}}\frac{e^{\pi\mathcal{H}\varphi(x)}}{\frac{1}{\sqrt{1+(u_n-x)^2/v_n^2}}}=\delta e^{\pi\mathcal{H}\varphi(x)}\neq 0.
\end{align*}
In addition,
\begin{align*}
\lim_{n\to \infty}\frac{v_n(e^{-\pi H_{v_n}f(u_n)}-\cos(\pi P_{v_n} f(u_n))}{\sin(\pi P_{v_n} f(u_n))}=0.
\end{align*}
Hence $x$ is a singular point and we have proved (\ref{eqSing1NTangLim}). 
\end{proof}
\begin{rem}
It may appear as if the limit $\lim_{n\rightarrow +\infty}\eta(u_n,v_n)=1-\delta e^{\pi\mathcal{H}\varphi(x)}$ depends on the arbitrary choice of $\delta>0$. However this is not so. Let $\delta'>0$ be another choice and we may assume that $\delta'>\delta$. We then get two functions  $\varphi_{\delta}(x)=f(t)-\chi_{[x-\delta,x]}(t)$ and $\varphi_{\delta'}(x)=f(t)-\chi_{[x-\delta',x]}(t)$.
Computing the difference of their Hilbert transforms at $x$ gives
\begin{align*}
\pi\mathcal{H}\varphi_{\delta}(x)-\pi\mathcal{H}\varphi_{\delta'}(x)=\int_{x-\delta'}^{x-\delta}\frac{dt}{x-t}=\log\frac{\delta'}{\delta}.
\end{align*}
Hence, $\delta e^{\pi\mathcal{H}\varphi_{\delta}(x)}=\delta'e^{\pi\mathcal{H}\varphi_{\delta'}(x)}$, and we see that the value of $\delta$ is unimportant in (\ref{eqSing1NTangLim}).
\end{rem}
\begin{ex}
Let $f(t)=t\chi_{[0,1]}(t)+(t+1)\chi_{[-1,0]}(t)$. Then $t=0$ satisfies the assumptions of the Proposition \ref{Sing1}. Hence $t=0$ is a singular point.
\end{ex}
Recall from the introduction that $\mathcal{S}_{nt}^{sing,II}(\mu)$ are the set of points $x\in \Sn$ such that for some $\delta>0$, $f(t)= \chi_{[x,x+\delta]}(t)+\varphi(t)$,
and that $x$ belongs to the Lebesgue set of $\varphi$ and that $\varphi(x)=0$ and $\vert \mathcal{H}\varphi(x)\vert<+\infty$

\begin{Prop}
\label{Sing2}
Assume that  $x\in\mathcal{S}_{nt}^{sing,II}(\mu)$. Then $x$ is a singular point, and
\begin{align}
\label{eqSing2nt}
\lim_{n\rightarrow +\infty}(\chi_\LL(u_n,v_n),\eta_\LL(u_n,v_n))=(x+\delta e^{-\pi\mathcal{H}\varphi(x)},1-\delta e^{-\pi\mathcal{H}\varphi(x)}):=(\chi_\G^{II}(x),\eta_\G^{II}(x))
\end{align}
for every non-tangential sequence $\{u_n+iv_n\}_{n=1}^{\infty}\in \Hp$ such that $\lim_{n\rightarrow +\infty}u_n+iv_n=x$.
\end{Prop}
\begin{proof}
Argument is analogus to the one given in the proof of the previous proposition. 
\end{proof}

\begin{rem}
\label{EdgeSing1}
Let $x\in\mathcal{S}_{nt}(\mu)^{\circ}$ and assume that the density $f$ of $\mu$ satisfies either the assumptions of Proposition \ref{Sing1} or Proposition \ref{Sing2} at $x$. Then $x$ does not belong to the Lebesgue set of $f$. If the assumptions of Proposition \ref{Sing1} hold we have $\lim_{h\rightarrow 0^+}M_Lf(x,h)=1$ and $\lim_{h\rightarrow 0^+}M_Rf(x,h)=0$, and if the assumptions of Proposition \ref{Sing2} hold then $\lim_{h\rightarrow 0^+}M_Lf(x,h)=0$ and $\lim_{h\rightarrow 0^+}M_Rf(x,h)=1$.
Furthermore, the sets $\mathcal{S}_{nt}^{sing,I}(\mu)$ and $\mathcal{S}_{nt}^{sing,II}(\mu)$ can be viewed as analogous to the sets $R_1$ and $R_2$ respectively, in the parametrization of the \emph{edge} $\mathcal{E}$. In particular we recall from Lemma 2.3 in \cite{Duse14a} that
\begin{align*}
(\chi_{\mathcal{E}}(t),\eta_{\mathcal{E}}(t))=\left\{
\begin{array}{ll} (t,1-\delta e^{C_{\delta}(t)}) &t\in R_1 \\
(t+\delta e^{-C_{\delta}(t)},1-\delta e^{-C_{\delta}(t)}) &t\in R_2,\\
\end{array} \right.
\end{align*}
where $\delta>0$ is sufficiently small and $C_{\delta}(t)=\int_{\R\backslash [x-\delta,x+\delta]}\frac{d\mu(x)}{t-x}$.
\end{rem}
\begin{ex}
Let $f(t)=(1-t)\chi_{[0,1]}(t)-t\chi_{[-1,0]}(t)$. Then $x=0$ satisfies the assumptions of Proposition \ref{Sing2}. Hence $x=0$ is a singular point.
\end{ex}
We now give a Proposition that shows that there is no gap between where the \emph{edge} $\mathcal{E}$ ends and the limit point of the non-tangential limits.

\begin{Prop}
Assume $x\in\Siso\cap (\mathcal{S}_{nt}^{sing,I}(\mu)\cup\mathcal{S}_{nt}^{sing,II})(\mu)$. Then:
\begin{itemize}
\item If $x\in\dv\mathcal{S}_R^0$ then $\lim_{\substack{t\to x\\ t\in R_{\mu}}}(\chi_{\mathcal{E}}(t),\eta_{\mathcal{E}}(t))=(x,1-\delta e^{\pi\mathcal{H}\varphi(x)})$
\item If $x\in\dv\mathcal{S}_R^1$ then $\lim_{\substack{t\to x\\ t\in R_{\lambda-\mu}}}(\chi_{\mathcal{E}}(t),\eta_{\mathcal{E}}(t))=(x+\delta e^{-\pi\mathcal{H}\varphi(x)},1-\delta e^{-\pi\mathcal{H}\varphi(x)})$
\item If $x\in\dv\mathcal{S}_L^0$ then $\lim_{\substack{t\to x\\ t\in R_{\mu}}}(\chi_{\mathcal{E}}(t),\eta_{\mathcal{E}}(t))=(x,1-\delta e^{\pi\mathcal{H}\varphi(x)})$
\item If $x\in\dv\mathcal{S}_L^1$ then $\lim_{\substack{t\to x\\ t\in R_{\lambda-\mu}}}(\chi_{\mathcal{E}}(t),\eta_{\mathcal{E}}(t))=(x+\delta e^{-\pi\mathcal{H}\varphi(x)},1-\delta e^{-\pi\mathcal{H}\varphi(x)})$
\end{itemize}
Furthermore, for every sequence $\{w_n\}_n\subset \Hp$ such that $\text{Re}[w_n]\geq x$,
\begin{align*}
\lim_{n\to \infty}(\chi_{\LL}(w_n),\eta_{\LL}(w_n))&=(\chi_\G^{I}(x),\eta_\G^{I}(x)) \quad \text{if $x\in \mathcal{S}_{nt}^{sing,I}(\mu)\cap \dv\mathcal{S}_R^0$ and}\\
\lim_{n\to \infty}(\chi_{\LL}(w_n),\eta_{\LL}(w_n))&=(\chi_\G^{II}(x),\eta_\G^{II}(x)) \quad \text{if $x\in \mathcal{S}_{nt}^{sing,II}(\mu)\cap \dv\mathcal{S}_R^1$},
\end{align*}
and for every sequence $\{w_n\}_n\subset \Hp$ such that $\text{Re}[w_n]\leq x$,
\begin{align*}
\lim_{n\to \infty}(\chi_{\LL}(w_n),\eta_{\LL}(w_n))&=(\chi_\G^{II}(x),\eta_\G^{II}(x)) \quad \text{if $x\in \mathcal{S}_{nt}^{sing,II}(\mu)\cap \dv\mathcal{S}_L^0$ and}\\
\lim_{n\to \infty}(\chi_{\LL}(w_n),\eta_{\LL}(w_n))&=(\chi_\G^{I}(x),\eta_\G^{I}(x)) \quad \text{if $x\in \mathcal{S}_{nt}^{sing,I}(\mu)\cap \dv\mathcal{S}_L^1$}.
\end{align*}

\end{Prop}

\begin{proof}
We prove the result in the case $x\in\dv\mathcal{S}_R^0\cap\mathcal{S}_{nt}^{sing,I}(\mu)$ and note that the other cases follow similarly. By definition of $\in \dv\mathcal{S}_R^0$,  there exists an interval $I=(x,x+\eps)$ for some $\eps>0$ such that $I\subset \supp (\mu)^c$. By (\ref{eqchiEEetaEE}) we then have a parametrization of a piece of the edge $\mathcal{E}$, which we recall is given by
\begin{align*}
\chi_{\mathcal{E}}(t)&=t+\frac{1-e^{-C(t)}}{C'(t)}\\
\eta_{\mathcal{E}}(t)&=1+\frac{e^{C(t)}+e^{-C(t)}-2}{C'(t)}
\end{align*}
for $t\in I$. We note that since $f(t)=\chi_{[x-\delta,x]}(t)+\varphi(t)$
\begin{align*}
C(t)&=\int_{\R}\frac{f(y)dy}{t-y}=\log\bigg\vert\frac{t-x+\delta}{t-x}\bigg\vert+\int_{\R}\frac{\varphi(y)dy}{t-y}
\end{align*}
and
\begin{align*}
C'(t)&=-\frac{1}{t-x}+\frac{1}{t-x+\delta}-\int_{\R}\frac{\varphi(y)dy}{(t-y)^2}.
\end{align*}
In particular we note that
\begin{align*}
(t-x)\int_{x-\delta}^{x}\frac{\varphi(y)dy}{(t-y)^2}&=\int_{x-\delta}^{x}\frac{(t-y)\varphi(y)dy}{(t-y)^2}-\int_{x-\delta}^{x}\frac{(x-y)\varphi(y)dy}{(t-y)^2}\\
&=\int_{x-\delta}^{x}\frac{\varphi(y)dy}{t-y}-\int_{x-\delta}^{x}\frac{(x-y)\varphi(y)dy}{(t-y)^2}.
\end{align*}
Since $y\leq x<t$ we have
\begin{align*}
0<\frac{(x-y)}{(t-y)^2}\leq \frac{1}{x-y}
\end{align*}
for all $t\in I$. By the assumption that $x\in \mathcal{S}_{nt}^{sing,I}(\mu)$, we have as in the proof of Proposition \ref{Sing1}
\begin{align*}
\int_{\R}\frac{\vert \varphi(y)\vert dy}{\vert x-y\vert}<+\infty.
\end{align*}
The Lebesgue dominated convergence theorem implies that
\begin{align*}
\lim_{t\rightarrow x^+}\int_{x-\delta}^{x}\frac{\varphi(y)dy}{t-y}-\int_{x-\delta}^{x}\frac{(x-y)\varphi(y)dy}{(t-y)^2}=\int_{x-\delta}^{x}\frac{\varphi(y)dy}{x-y}-\int_{x-\delta}^{x}\frac{\varphi(y)dy}{x-y}=0.
\end{align*}
Therefore,
\begin{align}
\label{RightLimit1}
\lim_{t\rightarrow x^+}(t-x)\int_{\R}\frac{\varphi(y)dy}{(t-y)^2}=0.
\end{align}
Similarly,
\begin{align}
\label{RightLimit2}
\lim_{t\rightarrow x^+}\int_{\R}\frac{\varphi(y)dy}{t-y)}=\int_{\R}\frac{\varphi(y)dy}{x-y)}=\mathcal{H}\varphi(x)
\end{align}
Computing the limits $\lim_{t\rightarrow x^+}\chi_{\mathcal{E}}(t)$ and $\lim_{t\rightarrow x^+}\eta_{\mathcal{E}}(t)$ using (\ref{RightLimit1}) and (\ref{RightLimit2}) gives
\begin{align*}
\lim_{t\rightarrow x^+}\chi_{\mathcal{E}}(t)&=\lim_{t\rightarrow x^+}t+\frac{1-\big\vert\frac{t-x}{t-x+\delta}\big\vert \exp{(-\int_{\R}\frac{\varphi(y)dy}{t-y}})}{-\frac{1}{t-x}+\frac{1}{t-x+\delta}-\int_{\R}\frac{\varphi(y)dy}{(t-y)^2}}\\
&=\lim_{t\rightarrow x^+}x+\frac{-(t-x)+\vert t-x+\delta\vert^{-1}(t-x)^2\exp{(-\int_{\R}\frac{\varphi(y)dy}{t-y}})}{1-\frac{t-x}{t-x+\delta}+(t-x)\int_{\R}\frac{\varphi(y)dy}{(t-y)^2}}=x
\end{align*}
and
\begin{align*}
\lim_{t\rightarrow x^+}\eta_{\mathcal{E}}(t)&=\lim_{t\rightarrow x^+}1+\frac{\big\vert\frac{t-x+\delta}{t-x}\big\vert \exp{(\int_{\R}\frac{\varphi(y)dy}{t-y}})+\big\vert\frac{t-x}{t-x+\delta}\big\vert \exp{(-\int_{\R}\frac{\varphi(y)dy}{t-y}})-2}{-\frac{1}{t-x}+\frac{1}{t-x+\delta}-\int_{\R}\frac{\varphi(y)dy}{(t-y)^2}}\\
&=\lim_{t\rightarrow x^+}1-\frac{(t-x+\delta)\exp{(\int_{\R}\frac{\varphi(y)dy}{t-y}})+\frac{(t-x)^2}{t-x+\delta}\exp{(-\int_{\R}\frac{\varphi(y)dy}{t-y}})}{1-\frac{t-x}{t-x+\delta}+(t-x)\int_{\R}\frac{\varphi(y)dy}{(t-y)^2}}\\
&=1-\delta e^{\pi\mathcal{H}\varphi(x)}.
\end{align*}
Now let $\{w_n\}_n\subset \Hp$ be a sequence such that $\text{Re}[w_n]\geq x$. Then a similar argument as above shows that $\lim_{n\to \infty}(\chi_{\LL}(w_n),\eta_{\LL}(w_n))=(\chi_\G^{I}(x),\eta_\G^{I}(x))$. The other cases can be treated analogously. 
\end{proof}

Proposition \ref{Sing1} and \ref{Sing2} deal with singular points where the density jumps in mean from 1 to 0 or from 0 to 1, see Remark \ref{EdgeSing1}. In particular these singular points do not belong to the Lebesgue set of $f$. In Proposition \ref{Sing3} below we instead consider singular points which belong to the Lebesgue set, and for which the density is either 0 or 1. Recall from the definition that $x\in \mathcal{S}_{nt}^{sing,III}(\mu)$ if
\begin{align*}
\int_{\R}\frac{f(t)dt}{(x-t)^2}<+\infty
\end{align*}
and $\mathcal{H}f(x)\neq 0$, and that $x\in \mathcal{S}_{nt}^{sing,IV}(\mu)$ if
\begin{align*}
\int_{\R}\frac{(1-f(t))dt}{(x-t)^2}<+\infty
\end{align*}
and $\mathcal{H}(1-f)(x)\neq 0$.

\begin{Prop}
\label{Sing3}
Assume that  $x\in \mathcal{S}_{nt}^{sing,III}(\mu)\bigcup \mathcal{S}_{nt}^{sing,IV}(\mu)$.  Then $x\in \mathscr{L}_f$ and $\vert \mathcal{H}f(x)\vert<+\infty$. Furthermore, $x$ is a singular point.  Finally, for every non-tangential sequence $\{u_n+iv_n\}_{n=1}^{\infty}\in \Hp$ that converges to $x$, if $x\in\mathcal{S}_{nt}^{sing,III}(\mu)$, then
\begin{align}
\label{eqSing3nt}
&\lim_{n\rightarrow +\infty}(\chi_\LL(u_n,v_n),\eta_\LL(u_n,v_n))\nonumber\\&=\big(x- (\pi(\mathcal{H}f)'(x))^{-1}(e^{-\pi\mathcal{H}f(x)}-1),1+(\pi(\mathcal{H}f)'(x))^{-1}(e^{\pi\mathcal{H}f(x)}+e^{-\pi\mathcal{H}f(x)}-2)\big)\nonumber \\
&:=(\chi_\G^{III}(x),\eta_\G^{III}(x)),
\end{align}
and if $x\in\mathcal{S}_{nt}^{sing,IV}(\mu)$, then
\begin{align}
\label{eqSing4nt}
&\lim_{n\rightarrow +\infty}(\chi_\LL(u_n,v_n),\eta_\LL(u_n,v_n))\nonumber \\&=\big(x- (\pi(\mathcal{H}(1-f))'(x))^{-1}(e^{\pi\mathcal{H}(1-f)(x)}+1),1+(\pi(\mathcal{H}(1-f))'(x))^{-1}(e^{\pi\mathcal{H}(1-f)(x)}+e^{-\pi\mathcal{H}(1-f)(x)}+2)\big)\nonumber\\
&:=(\chi_\G^{IV}(x),\eta_\G^{IV}(x))
\end{align}
where
\begin{align*}
(\mathcal{H}f)'(x):=-\frac{1}{\pi}\int_{\R}\frac{f(t)dt}{(x-t)^2}.
\end{align*}
\end{Prop}

\begin{proof}
Consider the case when $x\in \mathcal{S}_{nt}^{sing,III}(\mu)$. We first show that $x$ belongs to the Lebesgue set of $f$. Recall that
\begin{align*}
\int_{\R}\frac{f(t)dt}{(x-t)^2}<+\infty.
\end{align*}
Hence,
\begin{align*}
\lim_{h\rightarrow 0^+}\int_{x-h}^{x+h}\frac{f(t)dt}{(x-t)^2}=0.
\end{align*}
But for small $h$,
\begin{align*}
\int_{x-h}^{x+h}\frac{f(t)dt}{(x-t)^2}>h^{-2}\int_{x-h}^{x+h}f(t)dt>h^{-1}\int_{x-h}^{x+h}f(t)dt.
\end{align*}
This implies that
\begin{align*}
\lim_{h\rightarrow 0^+}\frac{1}{2h}\int_{x-h}^{x+h}f(t)dt=0,
\end{align*}
so $f(x)=0$ and $x\in\mathscr{L}_f$ by the discussion in the beginning of section 2.1. Take $\delta>0$ small, then
\begin{align*}
\int_{x-\delta}^{x+\delta}\frac{f(t)dt}{\vert x-t\vert}<\int_{x-\delta}^{x+\delta}\frac{f(t)dt}{(x-t)^2}<+\infty,
\end{align*}
by assumption, so $\vert \mathcal{H}f(x)\vert<+\infty$. Moreover, by Proposition \ref{DiniConv} this also implies that for every $k>0$
\begin{align*}
\lim_{\substack{(u,v)\to (x,0)\\ (u,v)\in \Gamma_{k}(x)}}H_vf(u)=\mathcal{H}f(x).
\end{align*}
By definition
\begin{align*}
\lim_{\substack{(u,v)\to (x,0)\\ (u,v)\in \Gamma_{k}(x)}}\frac{\sin(P_v f(u))}{v}=\lim_{\substack{(u,v)\to (x,0)\\ (u,v)\in \Gamma_{k}(x)}}\int_{\R}\frac{f(t)dt}{(u-t)^2+v^2}.
\end{align*}
Note that $(x-t)^2= 2((x-u)+(u-t))^2\leq (2k^2+2)(v^2+(u-t)^2)$ for all $(u,v)\in \G_k(x)$, $t\in \R$.
Hence, Lebesgue's dominated convergence theorem implies that
\begin{align*}
\lim_{\substack{(u,v)\to (x,0)\\ (u,v)\in \Gamma_{k}(x)}}\int_{\R}\frac{f(t)dt}{(u-t)^2+v^2}=\int_{\R}\frac{f(t)dt}{(t-x)^2}.
\end{align*}
In addition, since $x\in \mathscr{L}_f$, it follows that
\begin{align*}
\lim_{\substack{(u,v)\to (x,0)\\ (u,v)\in \Gamma_{k}(x)}}P_v f(u)=0.
\end{align*}

Consequently,
\begin{align*}
\lim_{\substack{(u,v)\to (x,0)\\ (u,v)\in \Gamma_{k}(x)}}\frac{v(e^-{\pi H_vf(u)}-\cos(\pi P_v f(u))}{\sin(\pi P_v f(u))}=\frac{e^{-\pi\mathcal{H}f(x)}-1}{\int_{\R}\frac{f(t)dt}{(x-t)^2}}\quad( \text{$\neq 0$ since $\mathcal{H}f(x)\neq 0$})
\end{align*}
and
\begin{align*}
\lim_{\substack{(u,v)\to (x,0)\\ (u,v)\in \Gamma_{k}(x)}}\frac{v(e^{\pi H_vf(u)}+e^{-\pi H_vf(u)}-2\cos(\pi P_v f(u))}{\sin(\pi P_v f(u))}=\frac{e^{\pi\mathcal{H}f(x)}+e^{-\pi\mathcal{H}f(x)}-2}{\int_{\R}\frac{f(t)dt}{(x-t)^2}}.
\end{align*}
Hence $\lim_{n\to \infty}(\chi_{\mathcal{L}}(w_n),\eta_{\mathcal{L}}(w_n))\neq (x,1)$, and so $x$ is singular. Again we notice that the limit is independent of $k$. In the case $x\in \mathcal{S}_{nt}^{sing,IV}(\mu)$ a completely analogous argument gives
\begin{align*}
\lim_{\substack{(u,v)\to (x,0)\\ (u,v)\in \Gamma_{k}(x)}}\frac{v(e^{-\pi H_vf(u)}-\cos(\pi P_v\ast f(u))}{\sin(\pi P_v\ast f(u))}=\frac{e^{\pi\mathcal{H}(1-f)(x)}+1}{\int_{\R}\frac{(1-f(t))dt}{(x-t)^2}},
\end{align*}
and
\begin{align*}
\lim_{\substack{(u,v)\to (x,0)\\ (u,v)\in \Gamma_{k}(x)}}\frac{v(e^{\pi H_vf(u)}+e^{-\pi H_vf(u)}-2\cos(\pi P_v\ast f(u))}{\sin(\pi P_v\ast f(u))}=\frac{e^{\pi\mathcal{H}(1-f)(x)}+e^{-\pi\mathcal{H}(1-f)(x)}+2}{\int_{\R}\frac{(1-f(t))dt}{(x-t)^2}},
\end{align*}
since $\mathcal{H}(1)=0$.
\end{proof}
\begin{rem}
It is interesting to compare the the results of Proposition \ref{Sing3} with the parametrization of the \emph{edge} $\mathcal{E}$. One sees that the sets $\mathcal{S}_{nt}^{sing,III}(\mu)$ and $\mathcal{S}_{nt}^{sing,III}(\mu)$ are analogues to the parametrization sets $R_{\mu}$ and $R_{\lambda-\mu}$, i.e., $(\chi_\EE(x),\eta_\EE(x))=(\chi_\G^{III}(x),\eta_\G^{III}(x))$ whenever $x\in \mathcal{S}_{nt}^{sing,III}(\mu)$, and $(\chi_\EE(x),\eta_\EE(x))=(\chi_\G^{IV}(x),\eta_\G^{IV}(x))$ whenever $x\in \mathcal{S}_{nt}^{sing,IV}(\mu)$.
\end{rem}

\begin{Prop}
\label{EdgeSing2}
Assume $x\in\Siso\cap (\mathcal{S}_{nt}^{sing,III}(\mu)\cup\mathcal{S}_{nt}^{sing,IV}(\mu))$. Then
\begin{align*}
\lim_{\substack{t\rightarrow x\\ t\in R}}(\chi_{\mathcal{E}}(t),\eta_{\mathcal{E}}(t))&=(\chi_\G^{III}(x),\eta_\G^{III}(x))\quad \text{if $x\in \mathcal{S}_{nt}^{sing,III}(\mu)$, and}\\
\lim_{\substack{t\rightarrow x\\ t\in R}}(\chi_{\mathcal{E}}(t),\eta_{\mathcal{E}}(t))&=(\chi_\G^{IV}(x),\eta_\G^{IV}(x))\quad \text{if $x\in \mathcal{S}_{nt}^{sing,IV}(\mu)$}.
\end{align*}
Furthermore, for every sequence $\{w_n\}_n\subset \Hp$ such that $\text{Re}[w_n]\geq x$,
\begin{align*}
\lim_{n\to \infty}(\chi_{\LL}(w_n),\eta_{\LL}(w_n))&=(\chi_\G^{III}(x),\eta_\G^{III}(x)) \quad \text{if $x\in \mathcal{S}_{nt}^{sing,III}(\mu)\cap \dv\mathcal{S}_R^0$ and}\\
\lim_{n\to \infty}(\chi_{\LL}(w_n),\eta_{\LL}(w_n))&=(\chi_\G^{IV}(x),\eta_\G^{IV}(x)) \quad \text{if $x\in \mathcal{S}_{nt}^{sing,IV}(\mu)\cap \dv\mathcal{S}_R^1$},
\end{align*}
and for every sequence $\{w_n\}_n\subset \Hp$ such that $\text{Re}[w_n]\leq x$,
\begin{align*}
\lim_{n\to \infty}(\chi_{\LL}(w_n),\eta_{\LL}(w_n))&=(\chi_\G^{III}(x),\eta_\G^{III}(x)) \quad \text{if $x\in \mathcal{S}_{nt}^{sing,III}(\mu)\cap \dv\mathcal{S}_L^0$ and}\\
\lim_{n\to \infty}(\chi_{\LL}(w_n),\eta_{\LL}(w_n))&=(\chi_\G^{IV}(x),\eta_\G^{IV}(x)) \quad \text{if $x\in \mathcal{S}_{nt}^{sing,IV}(\mu)\cap \dv\mathcal{S}_L^1$}.
\end{align*}
\end{Prop}
\begin{proof}
We consider the case $x\in \mathcal{S}_{nt}^{sing,III}(\mu)\cap \dv\mathcal{S}_R^0$. The other cases follow similarly. Since $x\leq u_n$ by assumption, $\frac{\vert u_n-y\vert}{(u_n-y)^2+v_n^2}\leq \frac{1}{x-y}$ and $\frac{1}{(u_n-y)^2+v_n^2}\leq \frac{1}{(x-y)^2}$ for $y\in (-\infty, x]$. Hence, by Lebesgue's dominated convergence theorem $\lim_{n\to \infty}H_{v_n}f(u_n)=\mathcal{H}f(x)$ and $\int_{\R}\lim_{n\to \infty}\frac{f(y)dy}{(u_n-y)^2+v_n^2}=\int_{\R}\frac{f(y)dy}{(x-y)^2}$. Therefore, a direct computation gives $\lim_{n\to \infty}(\chi_{\LL}(w_n),\eta_{\LL}(w_n))=(\chi_\G^{III}(x),\eta_\G^{III}(x))$. Similarly, $\lim_{\substack{t\rightarrow x\\ t\in R}}(\chi_{\mathcal{E}}(t),\eta_{\mathcal{E}}(t))=(\chi_\G^{III}(x),\eta_\G^{III}(x))$.  
\end{proof}

We now conclude this section by two lemmas that show that the set $\Ssing$ can have $\lambda(\Ssing)>0$ and that it may be dense in $\Sn^{\circ}$. In particular, this shows that the set $\Ssing$ can be at least a meagre set. Moreover, it also shows that the boundary $\dv\LL$ can be very complicated.

\begin{Lem}
\label{ExSingset1}
There exists a function $f\in \rho_{1,c}^{\lambda}(\R)$ such that $\lambda(S_{nt}^{sing,III}(\mu))>0$.
\end{Lem}
\begin{proof}
Let $I=(a,b)$ be an open interval. Associate to it the function
\begin{align}
\phi_I(t)=\exp\bigg(-\frac{1}{(t-a)^2(t-b)^2}\bigg)\chi_{(a,b)}(t).
\end{align}
Then $\phi\in C^{\infty}(\R)$ and $\text{supp}(\phi)=[a,b]$. Consider the interval $[0,1]$. Remove from it the middle $1/4$ interval $(3/8,5/8)$. Let $I_1^1=(3/8,5/8)$ and $\varphi_1(t)=\phi_{I_1^1}(t)$. Now remove the middle $1/16$ interval from the middle of the remaining intervals. That is, let $I_2^1=(5/32,7/32)$ and $I_2^2=(25/32,27/32)$, and let $\varphi_2(t)=\phi_{I_1^1}(t)+\phi_{I_2^1}(t)+\phi_{I_2^2}(t)$. Then $[0,1]-(I_1^1\cup I_2^1\cup I_2^2)=[0,5/32]\cup [7/32,3/8]\cup[5/8,25/32]\cup [27/32,1]$. Continue this process, by at step $n$, remove the middle $1/2^{2n}$:th interval from the remaining $2^{n-1}$ intervals. Let $I_n^k$ be the $k$:th of the $2^{n-1}$ open intervals that are deleted at each step, and let $\varphi_n(t)=\sum_{m=1}^n\sum_{k=1}^{2^{m-1}}\phi_{I_m^k}(t)$. Let $\mathfrak{C}=[0,1]\backslash (\bigcup_{n=1}^{\infty}\bigcup_{k=1}^{2^{n-1}}I_n^k)$. The set $\mathfrak{C}$ is called the Smith-Volterra-Cantor set. It can be shown that $\mathfrak{C}$ is a compact nowhere dense set such that $\lambda(\mathfrak{C})=1/2$. In particular $\mathfrak{C}^{\circ}=\varnothing$. Let $\varphi(t)=\lim_{n\to \infty}\varphi_n(t)$. Since all the intervals $I_n^k$ are disjoint and $\varphi\in C^{\infty}(\overline{I_n^k})$ for all $n$ and $k$ it follows that $\varphi\in C^{\infty}(\R)$ and $\text{supp}(\varphi)=[0,1]\backslash (\mathfrak{C}\backslash \dv\mathfrak{C})$. Moreover, $0\leq\varphi(t)<1$ for all $t$. However, if we choose a measure $\mu$ such that $\mu\big\vert_{[0,1]}=\varphi(t)dt$, then the measure-theoretic support of $\mu$ restricted to $[0,1]$ equals $[0,1]$. This is because, for every $x\in[0,1]$, we have for every neighbourhood $N_x$ of $x$ that $\lambda(N_x\cap \text{supp}(\varphi))>0$. We now note that $0\leq\varphi_n(t)\leq \varphi_{n+1}(t)\leq \varphi(t)$ for all $n$.
Hence, by Lebesgue's monotone convergence theorem we have for every $x\in\mathfrak{C}$
\begin{align*}
\lim_{n\to\infty}\int_{\R}\frac{\varphi_n(t)dt}{(x-t)^2}=\int_{\R}\frac{\varphi(t)dt}{(x-t)^2}.
\end{align*}
We now observe that for all $0\leq a<b\leq 1$ and $t\in \R$ and $x\in \R\backslash (a,b)$ we have
\begin{align*}
\phi_I(t)\leq e^{-1/(t-x)^2}.
\end{align*}
Therefore,
\begin{align*}
\int_{\R}\frac{\varphi_n(t)dt}{(x-t)^2}&=\sum_{m=1}^n\sum_{k=1}^{2^{m-1}}\int_{I_m^k}\frac{\phi_{I_m^k}(t)dt}{(x-t)^2}\leq  \sum_{m=1}^n\sum_{k=1}^{2^{m-1}}\int_{I_m^k}\frac{e^{-1/(x-t)^2}dt}{(x-t)^2}\\&\leq  \int_{\R}\frac{e^{-1/(x-t)^2}dt}{(x-t)^2}<+\infty,
\end{align*}
for all $n$. Hence $\int_{\R}\frac{\varphi(t)dt}{(x-t)^2}<+\infty$. Now assume that $\lambda( \{x\in \mathfrak{C}: \mathcal{H}\varphi(x)\neq0\})=0$. Take the density $f$ to be
\begin{align*}
f(t)=\varphi(t)+\chi_{[-1,-\Vert \varphi\Vert_1]}(t)
\end{align*}
Then $\mathcal{H}f=\mathcal{H}\varphi+\mathcal{H}\chi_{[-1,-\Vert \varphi\Vert_1]}$. Then for every $x\in [0,1]$, $\mathcal{H}\chi_{[-1,-\Vert \varphi\Vert_1]}(x)>0$. Thus $\lambda( \{x\in \mathfrak{C}: \mathcal{H}f(x)\neq0\})=\lambda(\mathfrak{C})$. This implies that $\lambda(S_{nt}^{sing,III}(\mu))=\lambda(\mathfrak{C})>0$. If on the other hand $\lambda( \{x\in \mathfrak{C}: \mathcal{H}\varphi(x)\neq0\})>0$, then let $f(t)=\varphi(\delta t)$, where $\delta=\Vert \varphi\Vert_1$. Then $\int_{\R}f(t)dt=\int_{\R}\varphi(t \delta)dt=\delta^{-1}\int_{\R}\varphi(x)dx=1$ and $0\leq f(t)<1$. Moreover, since $\mathcal{H}f(x)=\mathcal{H}\varphi(\delta t)(x)=\delta \mathcal{H}\varphi(\delta x)$, it follows that $\lambda( \{x\in \mathfrak{C}: \mathcal{H}f(x)\neq0\})>0$. Hence $\lambda(\{x\in \delta^{-1} \mathfrak{C}: \mathcal{H}f(x)\neq0\})>0$, and consequently $\lambda(S_{nt}^{sing,III}(\mu))>0$.
\end{proof}

\begin{rem}
Note that the set $S_{nt}^{sing,III}(\mu)$ in Lemma \ref{ExSingset1} is uncountable. Moreover, by Proposition \ref{GeoSing2} and Remark \ref{remSing} it follows that $\mathcal{H}^1(\dv \LL(x))>1-\eta_{\G}^{III}(x)>0$ for every $x\in S_{nt}^{sing,III}(\mu)$. Therefore, $\mathcal{H}^1(\dv \LL)>\sum_{x\in S_{nt}^{sing,III}(\mu)}(1-\eta_{\G}^{III}(x))$. However, every uncountable sum of positive real numbers is always infinite. Therefore $\mathcal{H}^1(\dv \LL(x))=+\infty$ for the measure in Lemma \ref{ExSingset1}. In Proposition \ref{propHausdorff} we give a simpler an more explicit example of a measure for which $\mathcal{H}^1(\dv \LL)=+\infty$.
\end{rem}

\begin{Lem}
\label{ExSingset2}
There exists a function $f\in \rho_{1,c}^{\lambda}(\R)$ and an interval $I\subset \text{supp}(f)$, such that $S_{nt}^{sing,III}(\mu)$ is dense in $I$.
\end{Lem}

\begin{proof}
Consider the dyadic set $S=\{x_k^r=k2^{-r}:r\geq 1, 1\leq k<2^r, \text{$k$ odd} \}$. The set $S$ is dense in the interval $[0,1]$. Let $I_k^r=(x_k^r-2^{-2r-1},x_k^r+2^{-2r-1})$. Then 
$U=\bigcup_{r,k}I_k^r$ is an open cover of the dense set $S$ such that
\begin{align*}
\vert U \vert &\leq \sum_{r=1}^{\infty} \sum _{\substack{k=1\\ \text{$k$ odd}}}^{2^r}\vert I_k^r\vert\leq\sum_{r=1}^{\infty} \sum _{k=1}^{2^r}\frac{1}{2^{2r}}=\sum_{r=1}^{\infty} \frac{1}{2^{r}}=\frac{1}{2}.
\end{align*}
Define the function $\varphi$ according to
\begin{align*}
\varphi(t):=\frac{1}{2}\chi_{[0,1]\backslash U}(t)+\inf_{r,k}\{(x_k^r-t)^2\chi_{I_k^r}(t)\}.
\end{align*}
We will show that $\inf_{r,k}\{(x_k^r-t)^2\chi_{I_k^r}(t)\}$ is not identically 0. Consider the set $E=\{t\in U\backslash S:\inf_{r,k}\{(x_k^r-t)^2\chi_{I_k^r}(t)\}=0\}$. We now estimate the measure of $E$. It is clear that if $t\in U\backslash S$ is such that there exists an $r_0$, such that 
\begin{align*}
t\notin \bigcup_{r=r'}^{\infty}\bigcup_{\substack{k=1\\ \text{$k$ odd}}}^{2^{r}}I_k^r
\end{align*}
whenever $r'\geq r_0$, then $\inf_{r,k}\{(x_k^r-t)^2\chi_{I_k^r}(t)\}>0$. Thus $E\subset \bigcup_{r=r_0}^{\infty}\bigcup_{\substack{k=1\\ \text{$k$ odd}}}^{2^{r}}I_k^r$ for every $r_0\geq 1$. Hence for every $\eps>0$
\begin{align*}
\vert E\vert &\leq \sum_{r=r_0}^{\infty} \sum _{\substack{k=1\\ \text{$k$ odd}}}^{2^r}\vert I_k^r\vert\leq\sum_{r=r_0}^{\infty} \sum _{k=1}^{2^r}\frac{1}{2^{2r}}=\sum_{r=r_0}^{\infty} \frac{1}{2^{r}}<\eps,
\end{align*}
whenever $r_0$ is sufficiently large. Thus $\vert E\vert =0$. Moreover, we clearly have that for any $x_k^r\in S$
\begin{align*}
\int_{\R}\frac{\varphi(t)dt}{(x_k^r-t)^2}&\leq \frac{1}{2}\int_{[0,1]\backslash I_k^r}\frac{dt}{(x_k^r-t)^2}+\frac{1}{2}\int_{ I_k^r}\frac{\inf_{m,l}\{(x_l^m-t)^2\chi_{I_l^m}(t)\}dt}{(x_k^r-t)^2}\\
&\leq \frac{1}{2}\int_{[0,1]\backslash I_k^r}\frac{dt}{(x_k^r-t)^2}+\frac{1}{2}\int_{ I_k^r}\frac{(x_k^r-t)^2\chi_{I_k^r}(t)dt}{(x_k^r-t)^2}<+\infty.
\end{align*}
Now assume that the set $\{x\in S:\mathcal{H}\varphi(x)\neq0\}$ is not dense on any interval $J\subset [0,1]$. Then the set $\{x\in S:\mathcal{H}\varphi(x)= 0\}$ is dense in $[0,1]$. Let 
\begin{align*}
f(t)=\varphi(t)+\chi_{[-1,-\Vert \varphi\Vert_1]}(t). 
\end{align*}
Then $\mathcal{H}\chi_{[-1,-\Vert \varphi\Vert_1]}(x)>0$ for all $x\in[0,1]$. Hence the set $\{x\in S:\mathcal{H}f(x)\neq 0\}$ is dense in $[0,1].$ Thus $S_{nt}^{sing,III}(\mu)$ is dense in [0,1]. If on the other hand the set $\{x\in S:\mathcal{H}\varphi(x)\neq0\}$ is dense in some interval $J\subset [0,1]$, let $f(t)=\varphi(\delta t)$, where $\delta=\Vert \varphi\Vert_1$. Then $\{x\in \delta^{-1}S:\mathcal{H}f(x)\neq0\}$ is dense in $\delta^{-1}J$.  Thus $S_{nt}^{sing,III}(\mu)$ is dense in $\delta^{-1}J$.
\end{proof}

\subsection{Geometry of $\dv\mathcal{L}(x)$ when $x\in\Ssing$.}
We have seen that if $x$ is singular then $\dv \LL(x)\neq \{(x,1)\}$. In this section we will determine subsets of $\dv \LL(x)$, and under additional assumptions on the density $f$, we will be able to determine the entire set $\dv\LL(x)$.
\begin{Prop}
\label{GeoSing1}
Let $x\in \mathcal{S}_{nt}^{sing,I}(\mu)\cap(\Sn^{\circ}\cup\Siso)$ and write for some $\delta>0$ the density as
\begin{align*}
f(t)= \chi_{[x-\delta,x]}(t)+\varphi(t).
\end{align*}
\begin{itemize}
\item
If $x\in \mathcal{S}_{nt}^{sing,I}(\mu)\cap\Sn^{\circ}$ assume that there exists an $\eps>0$ such that for all $y\in(x-\eps,x)\cup(x,x+\eps)$ 
\begin{align}
\label{Infty1}
\int_{x-\delta}^{x+\delta}\frac{\vert \varphi(t)\vert dt}{(y-t)^2}=+\infty.
\end{align} 
\item
If $x\in \mathcal{S}_{nt}^{sing,I}(\mu)\cap\Siso\cap\dv S_R^0$ assume that there exists an $\eps>0$ such that for all $y\in(x,x+\eps)$
\begin{align}
\label{Infty1BoundaryR}
\int_{x}^{x+\delta}\frac{\vert \varphi(t)\vert dt}{(y-t)^2}=+\infty.
\end{align} 
\item
If $x\in \mathcal{S}_{nt}^{sing,I}(\mu)\cap\Siso\cap\dv S_L^1$ assume that there exists an $\eps>0$ such that for all $y\in(x-\eps,x)$
\begin{align}
\label{Infty1BoundaryL}
\int_{x-\delta}^{x}\frac{\vert \varphi(t)\vert dt}{(y-t)^2}=+\infty.
\end{align} 
\end{itemize}
Then the parametrized line segment
\begin{align}
\{(\chi_{I}(\xi),\eta_{I}(\xi)):\xi\in(0,+\infty)\}&:=\bigg\{\bigg(x,1-\frac{\delta e^{\pi \mathcal{H}\varphi(x)}}{1+\xi}\bigg):\xi\in(0,+\infty)\bigg\}\nonumber \\&\subset \{(\chi,\eta)\in \R^2:\chi=x\}
\end{align} 
is a subset of $\Dl(x)$. In particular if $x\in \Sn^{\circ}$, there exist two one-parameter families of tangential continuous curves $\{s+iv_{\xi}^{+}(s):s\in(0,+\infty),\xi\in(0,+\infty)\}$ and $\{s+iv_{\xi}^-(s):s\in(-\infty,0),\xi\in(0,+\infty)\}$ where $v_{\xi}^{\pm}(s)$ is a continuous function of $s$ for each $\xi\in(0,+\infty)$, such that 
\begin{align*}
\lim_{s\rightarrow 0^{\small\pm}}\big(\chi_\LL(x+s+iv_{\xi}^{\pm}(s)),\eta_\LL(x+s+iv_{\xi}^{\pm}(s))\big)=(\chi_I(\xi),\eta_I(\xi)),
\end{align*}
and if $x\in \dv S_R^0$ there exists a one-parameter family of tangential continuous curves $\{s+iv_{\xi}^{+}(s):s\in(0,+\infty),\xi\in(0,+\infty)\}$ such that
\begin{align*}
\lim_{s\rightarrow 0^{\small+}}\big(\chi_\LL(x+s+iv_{\xi}^{\pm}(s)),\eta_\LL(x+s+iv_{\xi}^{\pm}(s))\big)=(\chi_I(\xi),\eta_I(\xi)),
\end{align*}
and if $x\in \dv S_L^1$ there exists a one-parameter family of tangential continuous curves $\{s+iv_{\xi}^{-}(s):s\in(-\infty,0),\xi\in(0,+\infty)\}$ such that
\begin{align*}
\lim_{s\rightarrow 0^{\small+}}\big(\chi_\LL(x+s+iv_{\xi}^{\pm}(s)),\eta_\LL(x+s+iv_{\xi}^{\pm}(s))\big)=(\chi_I(\xi),\eta_I(\xi)).
\end{align*}
The curves $\{s+iv_{\xi}^{+}(s):s\in(0,+\infty),\xi\in(0,+\infty)\}$ and $\{s+iv_{\xi}^-(s):s\in(-\infty,0),\xi\in(0,+\infty)\}$ satisfy the equation
\begin{align*}
G_1(s,v_{\xi}^{\pm}(s);x)=\xi
\end{align*}
where 
\begin{align*}
G_1(s,v;x)=1_{s>0}\int_{x}^{x+2s}\frac{s\varphi(t)dt}{(x+s-t)^2+v^2}+1_{s<0}\int_{x+2s}^{x}\frac{s\varphi(t)dt}{(x+s-t)^2+v^2}.
\end{align*}
\end{Prop}

\begin{rem}
It is worth mentioning that the assumption that $x\in \mathcal{S}_{nt}^{sing,I}(\mu)\cap(\Sn^{\circ}\cup\Siso)$ is never strictly used. In particular, the assumption that there exists an $\eps>0$ such that for all $y\in(x-\eps,x)\cup(x,x+\eps)$ 
\begin{align*}
\int_{x-\delta}^{x+\delta}\frac{\vert \varphi(t)\vert dt}{(y-t)^2}=+\infty,
\end{align*} 
implies that $(x,x+\eps)\cap (\mathcal{S}_{nt}^{sing,III}\cup R_{\mu})=\varnothing$ and $(x-\eps,x)\cap (\mathcal{S}_{nt}^{sing,IV}\cup R_{\lambda-\mu})=\varnothing$.
\end{rem}

\begin{proof}
Let $x\in \mathcal{S}_{nt}^{sing,I}\cap\Sn^{\circ}$. We may without loss of generality assume $x=0$. The Poisson integral of $f$ can then be written as
\begin{align*}
\pi P_v f(u)&=\arctan\bigg(\frac{u+\delta}{v}\bigg)-\arctan\bigg(\frac{u}{v}\bigg)+\int_{\R}\frac{v\varphi(t)dt}{(u-t)^2+v^2}.
\end{align*} 
We now let $w(s)=s+iv(s)$ be any path in $\Hp$ such that $\lim_{s\rightarrow 0}w(s)=0$ and $\lim_{s\rightarrow 0}\vert \frac{v(s)}{s}\vert=0$, i.e., we assume that the path is tangential. Then, using the identities $\arctan\big(\frac{1}{x}\big)=\frac{\pi}{2}-\arctan x$ for $x>0$ and $\arctan\big(\frac{1}{x}\big)=-\frac{\pi}{2}-\arctan x$
for $x<0$ and $\arctan(x)=x+O(x^2)$
we find that
\begin{align*}
\arctan\bigg(\frac{s+\delta}{v(s)}\bigg)-\arctan\bigg(\frac{s}{v(s)}\bigg)=\left\{
\begin{array}{rl} \frac{v(s)}{s}+O(v(s)^2/s^2) & \text{if } s>0 \\
\pi+\frac{v(s)}{s}+O(v(s)^2/s^2)  & \text{if } s<0. 
 \end{array} \right.
\end{align*} 
Hence
\begin{align}
\label{P1}
\pi P_{v(s)} f(s)=\left\{
\begin{array}{rl} \displaystyle\frac{v(s)}{s}+\int_{\R}\frac{v(s)\varphi(t)dt}{(s-t)^2+v(s)^2}+O\bigg(\frac{v(s)^2}{s^2}\bigg) & \text{if } s>0 \\
\displaystyle\pi+\frac{v(s)}{s}+\int_{\R}\frac{v(s)\varphi(t)dt}{(s+t)^2+v(s)^2}+O\bigg(\frac{v(s)^2}{s^2}\bigg)  & \text{if } s<0. 
 \end{array} \right.
\end{align} 
We now compute $\pi H_{v(s)}f(s)$ getting
\begin{align}
\label{H1}
\pi H_{v(s)}f(s)&=\pi H_{v(s)}\varphi(s)-\log\sqrt{s^2+v(s)^2}+\log\sqrt{(s+\delta)^2+v(s)^2}
\end{align} 
Hence, (\ref{P1}) and (\ref{H1}) gives
\begin{align}
\label{eqSing1L}
\lim_{s\rightarrow 0}\frac{v(s)e^{\pi H_{v(s)}f(s)}}{\sin(\pi P_{v(s)} f(s))}&=\lim_{s\rightarrow 0}\frac{\text{sgn}(s)\sqrt{(s+\delta)^2+v(s)^2}}{\sqrt{s^2+v(s)^2}}\frac{v(s)e^{\pi H_{v(s)}\varphi(s)}}{\displaystyle\frac{v(s)}{s}+\int_{\R}\frac{v(s)\varphi(t)dt}{(s-t)^2+v(s)^2}+O\bigg(\frac{v(s)^2}{s^2}\bigg) }\nonumber \\
&=\lim_{s\rightarrow 0}\frac{\text{sgn}(s)\delta}{\sqrt{1+(v(s)/s)^2}}\frac{e^{\pi H_{v(s)}\varphi(s)}}{\displaystyle\text{sgn}(s)\bigg(1+s\int_{\R}\frac{\varphi(t)dt}{(s-t)^2+v(s)^2}+O\bigg(\frac{v(s)}{s}\bigg) \bigg)}\nonumber \\
&=\lim_{s\rightarrow 0}\frac{\delta e^{\pi H_{v(s)}\varphi(s)}}{\displaystyle1+s\int_{\R}\frac{\varphi(t)dt}{(s-t)^2+v(s)^2}},
\end{align} 
since $v(s)/s\to 0$ as $s\to 0^+$. This suggest that we find $(s,v(s))$ such that
\begin{align*}
s\int_{\R}\frac{\varphi(t)dt}{(s-t)^2+v(s)^2}=\xi+o(1)
\end{align*} 
as $s\to 0$ for some fixed $\xi\in(0,+\infty)$.
Restricting ourselves to the case when $s>0$, we split the integral into two terms according to
\begin{align*}
s \int_{\R}\frac{\varphi(t)dt}{(s-t)^2+v(s)^2}&=s \int_{\vert s-t\vert\geq s}\frac{\varphi(t)dt}{(s-t)^2+v(s)^2}+s \int_{0}^{2s}\frac{\varphi(t)dt}{(s-t)^2+v(s)^2}\\
&:=J_1(s)+J_2(s).
\end{align*} 
For any fixed $t\neq 0$
\begin{align*}
\lim_{s\rightarrow 0^+}\frac{s\chi_{\vert s-t\vert\geq s}(t)}{(s-t)^2+v(s)^2}=0
\end{align*} 
and we also have the estimate 
\begin{align*}
\frac{s\chi_{\vert s-t\vert\geq s}(t)}{(s-t)^2+v(s)^2}\leq\frac{\vert s-t\vert\chi_{\vert s-t\vert\geq s}(t)}{(s-t)^2+v(s)^2}\leq\frac{2}{\vert t\vert}
\end{align*} 
since $\vert t\vert\leq 2\vert s-t\vert$. The assumption that $\vert \mathcal{H}\varphi(0)\vert<+\infty$ together with the argument at the beginning of the proof of Proposition \ref{Sing1} shows that 
\begin{align}
\label{eqSing1Int}
\int_{\R}\bigg\vert\frac{\varphi(t)}{t}\bigg\vert dt<+\infty.
\end{align}
It follows from Lebesgue's dominated convergence theorem that 
\begin{align}
\label{eqSing1Lim}
\lim_{s\rightarrow 0^+}J_1(s)=0.
\end{align}
We therefore let
\begin{align}
G_1(s,v;0)&:= s \int_{0}^{2s}\frac{\varphi(t)dt}{(s-t)^2+v^2}.
\end{align} 
We note that since $\varphi(t)\geq 0$ for $t>0$, $G_1(s,v;0)$ is a positive monotonically decreasing function of $v$. By assumption (\ref{Infty1})
\begin{align*}
\int_{-\delta}^{\delta}\frac{\vert \varphi(t)\vert dt}{(s-t)^2}=+\infty.
\end{align*}
Since by assumption $\varphi(t)\geq 0$ for $t\in[0,\delta]$
\begin{align*}
\int_{-\delta}^{\delta}\frac{\vert \varphi(t)\vert dt}{(s-t)^2}=\int_{-\delta}^{0}\frac{\vert \varphi(t)\vert dt}{(s-t)^2}+\int_{0}^{2s}\frac{\varphi(t)dt}{(s-t)^2}+\int_{2s}^{\delta}\frac{\varphi(t)dt}{(s-t)^2},
\end{align*}
and clearly
\begin{align*}
\int_{-\delta}^{0}\frac{\vert \varphi(t)\vert dt}{(s-t)^2}<+\infty\quad, \int_{2s}^{\delta}\frac{\varphi(t)dt}{(s-t)^2}<+\infty,
\end{align*}
it follows that
\begin{align*}
\int_{0}^{2s}\frac{\varphi(t)dt}{(s-t)^2}=+\infty.
\end{align*}
for all $s\in(0,\eps)$ and some $0<\eps<\delta/2$, and it follows by Fatou's lemma that
\begin{align*}
\lim_{v\rightarrow 0^+}G_1(s,v;0)=+\infty.
\end{align*}
In addition, 
\begin{align*}
\lim_{v\rightarrow +\infty}G_1(s,v;0)=0.
\end{align*}
Therefore, the equation
\begin{align}
\label{Eq1}
G_1(s,v;0)=s\int_{0}^{2s}\frac{\varphi(t)dt}{(s-t)^2+v^2}=\xi
\end{align} 
has unique solution $v=v_{\xi}^+(s)$ for each fixed $\xi\in(0,+\infty)$ and all $s\in(0,\eps)$. Since 
\begin{align*}
\partial_{v}G_1(s,v;0)=-2sv\int_{\R}\frac{\varphi(t)dt}{((s-t)^2+v^2)^2}<0
\end{align*} 
for all $s,v\in\Hp$ the implicit function theorem implies that for each $\xi>0$, $v_{\xi}^+(s)$ is a continuous function of $s\in (0,\eps)$. 
We want to show that $v_\xi^+(s)/s\to 0$ as $s\to 0^+$, i.e., that the path is tangential. For simplicity of notation we will write just $v(s)$ for the solution of (\ref{Eq1}). Assume that there is a sequence $\{s_j\}_j$ such that $v(s_j)/s_j\to 2/k$ as $j\to \infty$. Then, for all sufficiently large $j$ we have $s_j\leq kv(s_j)$. However, since $0$ belongs to the Lebesgue set of $\varphi$ and $\varphi(0)=0,$ we have, by (\ref{Eq1}), that 
\begin{align*}
\xi&=\lim_{j\to \infty}s_j\int_{0}^{2s_j}\frac{\varphi(t)dt}{(s_j-t)^2+v(s_j)^2}\leq \lim_{j\to \infty}k\int_{0}^{2s_j}\frac{v(s_j)\varphi(t)dt}{(s_j-t)^2+v(s_j)^2}\\
&=\lim_{j\to \infty}k\pi P_{v(s_j)}\varphi(s_j)=0,
\end{align*} 
contradicting $\xi>0$. Here we also used (\ref{eqSing1Lim}) to extend the domain of integration to $\R$ in the next to last inequality. Next we want to prove that
\begin{align}
\label{eqSing1Hilb}
\lim_{s\to 0^+}H_{v(s)}\varphi(s)=\mathcal{H}\varphi(0),
\end{align} 
which is not immediate since the curve $(s,v(s))$ is tangential. Take $\eta\in (0,1)$ fixed but arbitrary and write 
\begin{align}
\label{eqSing1Lim2}
\pi H_{v(s)}\varphi(s)&=\int_{\vert s-t\vert\geq \eta s}\frac{(s-t)\varphi(t)dt}{(s-t)^2+v(s)^2}+\int_{\vert s-t\vert\leq \eta s}\frac{(s-t)\varphi(t)dt}{(s-t)^2+v(s)^2}\nonumber \\
&=I_1(s)+I_2(s). 
\end{align} 
If $\vert s-t\vert\geq \eta s$, then $\vert t\vert \leq \vert s-t \vert +s\leq (1+1/\eta)\vert s-t\vert$ and we see that 
\begin{align*}
\frac{\vert s-t\vert\chi_{\vert s-t\vert\geq \eta s}}{(s-t)^2+v(s)^2}\leq \frac{1+1/\eta}{\vert t\vert}.
\end{align*} 
Also, 
\begin{align*}
\lim_{s\to 0^+}\frac{(s-t)\varphi(t)}{(s-t)^2+v(s)^2}=-\frac{\varphi(t)}{t}
\end{align*} 
for all $t\neq 0$. By (\ref{eqSing1Int}) and the dominated convergence theorem we see that 
\begin{align}
\label{eqSing1Lim3}
\lim_{s\to 0^+}I_1(s)=\mathcal{H}\varphi(0).
\end{align} 
Now, 
\begin{align*}
\vert I_2(s)\vert \leq\int_{\vert s-t\vert\leq \eta s}\frac{\vert s-t\vert \varphi(t)dt}{(s-t)^2+v(s)^2}\leq \eta\int_0^{2s}\frac{s\varphi(t)}{(s-t)^2+v(s)^2}=\eta\xi
\end{align*} 
since $\eta\leq 1$. It follows from this estimate that, (\ref{eqSing1Lim2}) and (\ref{eqSing1Lim3}) that 
\begin{align*}
\limsup_{s\to 0^+}\vert \pi H_{v(s)}\varphi(s)-\mathcal{H}\varphi(0)\vert \leq \eta \xi. 
\end{align*} 
Since $\eta\in(0,1)$ was arbitrary we have proved (\ref{eqSing1Hilb}).
It follows from (\ref{H1}), (\ref{eqSing1L}) and (\ref{eqSing1Hilb}) that 
\begin{align*}
\lim_{s\rightarrow 0+}\frac{v(s)e^{\pi H_{v(s)}f(s)}}{\sin(\pi P_{v(s)} f(s))}&=\frac{\delta e^{\pi \mathcal{H}\varphi(0)}}{1+\xi}.
\end{align*} 
Looking back at (\ref{H1}) and (\ref{eqSing1L}) we also see that 
\begin{align*}
\lim_{s\rightarrow 0^+}\frac{v(s)e^{-\pi H_{v(s)}f(s)}}{\sin(\pi P_{v(s)} f(s))}&=0.
\end{align*} 
Finally,
\begin{align*}
\lim_{s\rightarrow 0^+}\frac{v(s)\cos(\pi P_{v(s)} f(s))}{\sin(\pi P_{v(s)} f(s))}=\lim_{s\rightarrow 0^+}\frac{\cos(\pi P_{v(s)} f(s))}{\frac{1}{s}+\frac{\xi}{s}+\frac{1}{s}O(v(s)/s)}=0.
\end{align*} 
To conclude this implies that
\begin{align*}
\lim_{s\rightarrow 0^+}\chi_\LL(x+s,v_{\xi }^{+}(s))&=x\\
\lim_{s\rightarrow 0^+}\eta_\LL(x+s,v_{\xi}^{+}(s))&=1-\frac{\delta e^{\pi \mathcal{H}\varphi(0)}}{1+\xi}
\end{align*} 
for any fixed $\xi>0$, and an identical argument shows that this holds also for $v_\xi(s)^-$. 
\end{proof}

\begin{Cor}
\label{corGeoSing1}
Let $x\in \mathcal{S}_{nt}^{sing,I}(\mu)\cap\Sn^{\circ}$ and write for some $\delta>0$ the density
\begin{align*}
f(t)=\chi_{[x-\delta,x]}(t)+\varphi(t).
\end{align*}
Assume that $x$ belongs to the Lebesgue set of $\varphi$ and that $\varphi(x)=0$ and $\vert \mathcal{H}\varphi(x)\vert<+\infty$. Then for each fixed $\xi\in(0,+\infty)$ there exists a sequence $\{u_n+iv_{\xi}(u_n)\}_{n=1}^{+\infty}\subset \Hp$ such that $\lim_{n\rightarrow+\infty}u_n+iv_{\xi}(u_n)=x$ and such that
\begin{align*}
&\lim_{n\rightarrow+\infty}\chi_\LL(u_n+iv_{\xi}(s_n))=\chi_{I}(\xi)\\
&\lim_{n\rightarrow+\infty}\eta_\LL(u_n+iv_{\xi}(s_n))=\eta_{I}(\xi).
\end{align*}
\end{Cor}

\begin{proof}
We may repeat the proof of Proposition \ref{GeoSing1} replacing a continuous path $(s,v(s))$ with a sequence $(s_n,v(s_n))$ every where. The major difference is that if we assume that condition (\ref{Infty1}) does not hold then in general we may not conclude that $\lim_{v\rightarrow 0^+}G_1(s,v;x)=+\infty$ for all $s\in(0,\eps)$. However, there exists a sequence
$\{s_j\}_{j=1}^{+\infty}\subset(0,\eps)$ such that $\lim_{j\rightarrow +\infty}s_j=0$ and
\begin{align*}
\lim_{v\rightarrow 0^+}G_1(s_j,v;x)=+\infty,
\end{align*}
since assume otherwise. Then $\int_{\R}\frac{\vert \varphi(t)\vert}{(y-t)^2}<+\infty$ for all $y\in N_x\backslash \{x\}$ in some neighborhood $N_x$ of $x$. However this implies that $\varphi(t)=0$ for every $t\in N_x$, contradicting the assumption that $x\in \Sn^{\circ}$. 
Thus, for each such $s_j$ and fixed $\xi$ we may therefore solve for $v=v_{\xi}(s_j)$ in equation (\ref{Eq1}). The rest of the proof remains the same.
\end{proof}

\begin{Prop}
\label{GeoSing4}
Let $x\in \mathcal{S}_{nt}^{sing,II}(\mu)\cap(\Sn^{\circ}\cup\Siso)$ and write for some $\delta>0$ the density as
\begin{align*}
f(t)=\chi_{[x,x+\delta]}(t)+\varphi(t).
\end{align*}
\begin{itemize}
\item
If $x\in \mathcal{S}_{nt}^{sing,I}(\mu)\cap\Sn^{\circ}$ assume that there exists an $\eps>0$ such that for all $y\in(x-\eps,x)\cup(x,x+\eps)$ 
\begin{align}
\label{Infty12}
\int_{x-\delta}^{x+\delta}\frac{\vert \varphi(t)\vert dt}{(y-t)^2}=+\infty.
\end{align} 
\item
If $x\in \mathcal{S}_{nt}^{sing,I}(\mu)\cap\Siso\cap \in \dv S_R^0$ assume that there exists an $\eps>0$ such that for all $y\in(x,x+\eps)$
\begin{align}
\label{Infty1BoundaryR2}
\int_{x}^{x+\delta}\frac{\vert \varphi(t)\vert dt}{(y-t)^2}=+\infty.
\end{align} 
\item
If $x\in  \mathcal{S}_{nt}^{sing,I}(\mu)\cap\Siso\cap \dv S_L^1$, assume that there exists an $\eps>0$ such that for all $y\in(x-\eps,x)$
\begin{align}
\label{Infty1BoundaryL2}
\int_{x-\delta}^{x}\frac{\vert \varphi(t)\vert dt}{(y-t)^2}=+\infty.
\end{align} 
\end{itemize}
Then the parametrized line segment
\begin{align}
\{(\chi_{II}(\xi),\eta_{II}(\xi)):\xi\in(0,+\infty)\}&:=\bigg\{\bigg(x+\frac{\delta e^{-\pi \mathcal{H}\varphi(x)}}{1+\xi},1-\frac{\delta e^{-\pi \mathcal{H}\varphi(x)}}{1+\xi}\bigg):\xi\in(0,+\infty)\bigg\}\nonumber\\ \subset \{(\chi,\eta)\in \R^2:\chi+\eta=x+1\}
\end{align} 
is a subset of $\Dl(x)$. In particular if $x\in \Sn^{\circ}$, there exist two one-parameter families of tangential continuous curves $\{s+iv_{\xi}^{+}(s):s\in(0,+\infty),\xi\in(0,+\infty)\}$ and $\{s+iv_{\xi}^-(s):s\in(-\infty,0),\xi\in(0,+\infty)\}$ where $v_{\xi}^{\pm}(s)$ is a continuous function of $s$ for each $\xi\in(0,+\infty)$, such that 
\begin{align*}
\lim_{s\rightarrow 0^{\small\pm}}\big(\chi_\LL(x+s+iv_{\xi}^{\pm}(s)),\eta_\LL(x+s+iv_{\xi}^{\pm}(s))\big)=(\chi_{II}(\xi),\eta_{II}(\xi)),
\end{align*}
and if $x\in \dv S_R^0$ there exists a one-parameter family of tangential continuous curves $\{s+iv_{\xi}^{+}(s):s\in(0,+\infty),\xi\in(0,+\infty)\}$ such that
\begin{align*}
\lim_{s\rightarrow 0^{\small+}}\big(\chi_\LL(x+s+iv_{\xi}^{\pm}(s)),\eta_\LL(x+s+iv_{\xi}^{\pm}(s))\big)=(\chi_{II}(\xi),\eta_{II}(\xi)),
\end{align*}
and if $x\in \dv S_L^1$ there exists a one-parameter family of tangential continuous curves $\{s+iv_{\xi}^{-}(s):s\in(-\infty,0),\xi\in(0,+\infty)\}$ such that
\begin{align*}
\lim_{s\rightarrow 0^{\small-}}\big(\chi_\LL(x+s+iv_{\xi}^{\pm}(s)),\eta_\LL(x+s+iv_{\xi}^{\pm}(s))\big)=(\chi_{II}(\xi),\eta_{II}(\xi)).
\end{align*}
The curves $\{s+iv_{\xi}^{+}(s):s\in(0,+\infty),\xi\in(0,+\infty)\}$ and $\{s+iv_{\xi}^-(s):s\in(-\infty,0),\xi\in(0,+\infty)\}$ satisfy the equation
\begin{align*}
G_1(s,v_{\xi}^{\pm}(s);x)=\xi
\end{align*}
where 
\begin{align*}
G_1(s,v;x)=1_{s>0}\int_{x}^{x+2s}\frac{s\varphi(t)dt}{(x+s-t)^2+v^2}+1_{s<0}\int_{x+2s}^{x}\frac{s\varphi(t)dt}{(x+s-t)^2+v^2}
\end{align*}
\end{Prop}

\begin{proof}
The proof is identical to the proof of Proposition \ref{GeoSing1}.
\end{proof}

\begin{Prop}
\label{GeoSing2}
Let $x\in  (\mathcal{S}_{nt}^{sing,III}(\mu)\cup\mathcal{S}_{nt}^{sing,IV}(\mu))\cap(\Sn^{\circ}\cup\Siso)$. Assume that there exists a $\delta>0$ such that:
\begin{itemize}
\item  If $x\in \mathcal{S}_{nt}^{sing,III}(\mu)\cap\Sn^{\circ}$
\begin{align}
\label{C1}
\int_{x-\delta}^{x+\delta}\frac{f(t)dt}{(y-t)^2}=+\infty,
\end{align}
for all $y\in(x-\delta,x)\cup(x,x+\delta)$.
\item  If $x\in \mathcal{S}_{nt}^{sing,IV}(\mu)\cap\Sn^{\circ}$
\begin{align}
\label{C2}
\int_{x-\delta}^{x+\delta}\frac{1-f(t)dt}{(y-t)^2}=+\infty,
\end{align}
for all $y\in(x-\delta,x)\cup(x,x+\delta)$.
\item If $x\in \mathcal{S}_{nt}^{sing,III}(\mu)\cap\dv S_R^0$
\begin{align}
\label{C3}
\int_{x-\delta}^{x}\frac{f(t)dt}{(y-t)^2}=+\infty,
\end{align}
for all $y\in(x-\delta,x)$. If $x\in \mathcal{S}_{nt}^{sing,IV}(\mu)\cap\dv S_R^1$
\begin{align}
\label{C4}
\int_{x-\delta}^{x}\frac{1-f(t)dt}{(y-t)^2}=+\infty,
\end{align}
for all $y\in(x-\delta,x)$.
\item If $x\in \mathcal{S}_{nt}^{sing,III}(\mu)\cap\dv S_L^0$
\begin{align}
\label{C5}
\int_{x}^{x+\delta}\frac{f(t)dt}{(y-t)^2}=+\infty,
\end{align}
for all $y\in(x-\delta,x)$.
\item If $x\in \mathcal{S}_{nt}^{sing,IV}(\mu)\cap\dv S_L^1$
\begin{align}
\label{C6}
\int_{x}^{x+\delta}\frac{1-f(t)dt}{(y-t)^2}=+\infty,
\end{align}
for all $y\in(x-\delta,x)$. 
\end{itemize}
Then the parametrized curve
\begin{align*}
\{(\chi_{III}(\xi),\eta_{III}(\xi)):\xi\in(0,+\infty)\}:=&\bigg\{\bigg(x+\frac{1-e^{-\pi\mathcal{H}f(x)}}{\xi-\pi(\mathcal{H}f)'(x)},1-\frac{e^{\pi\mathcal{H}f(x)}+e^{-\pi\mathcal{H}f(x)}-2}{\xi-\pi(\mathcal{H}f)'(x)}\bigg):\xi\in(0,+\infty)\bigg\}\\&\subset \{(\chi,\eta)\in \R^2:\eta-1=(1-e^{\pi\mathcal{H}f(x)})(\chi-x)\}
\end{align*}
is a subset of $\Dl(x)$ if $x\in\mathcal{S}_{nt}^{sing,III}(\mu)$, and the parametrized curve
\begin{align*}
\{(\chi_{IV}(\xi),\eta_{IV}(\xi)):\xi\in(0,+\infty)\}:=&\bigg\{\bigg(x+\frac{1+e^{\pi\mathcal{H}(1-f)(x)}}{\xi-\pi(\mathcal{H}(1-f)'(x)},1-\frac{e^{\pi\mathcal{H}(1-f)(x)}+e^{-\pi\mathcal{H}(1-f)(x)}+2}{\xi-\pi(\mathcal{H}(1-f)'(x)}\bigg):\xi\in(0,+\infty)\bigg\}\\&\subset \{(\chi,\eta)\in \R^2:\eta-1=(1+e^{-\pi\mathcal{H}(1-f)(x)})(\chi-x)\}
\end{align*}
is a subset of $\Dl(x)$ if $x\in\mathcal{S}_{nt}^{sing,IV}(\mu)$. In particular if $x\in \Sn^{\circ}$, there exist two one-parameter families of tangential continuous curves $\{s+iv_{\xi}^{+}(s):s\in(0,+\infty),\xi\in(0,+\infty)\}$ and $\{s+iv_{\xi}^-(s):s\in(-\infty,0),\xi\in(0,+\infty)\}$ where $v_{\xi}^{\pm}(s)$ is a continuous function of $s$ for each $\xi\in(0,+\infty)$, such that 
\begin{align*}
\lim_{s\rightarrow 0^{\small\pm}}\chi_\LL(x+s+iv_{\xi}^{\pm}(s))&:=\chi_{III/IV}(\xi)\\
\lim_{s\rightarrow 0^{\pm}}\eta_\LL(x+s+iv_{\xi}^{\pm}(s))&:=\eta_{III/IV}(\xi).
\end{align*}
and if $x\in \dv S_R^0\cup \dv S_R^1$ there exists a one-parameter family of tangential continuous curves $\{s+iv_{\xi}^{+}(s):s\in(0,+\infty),\xi\in(0,+\infty)\}$ such that
\begin{align*}
\lim_{s\rightarrow 0^{\small +}}(\chi_\LL(x+s+iv_{\xi}^{+}(s)),\eta_\LL(x+s+iv_{\xi}^{+}(s)))=(\chi_{III/IV}(\xi),\eta_{III/IV}(\xi)),
\end{align*}
and if $x\in \dv S_L^0\cup \dv S_L^1$ there exists a one-parameter family of tangential continuous curves $\{s+iv_{\xi}^{-}(s):s\in(-\infty,0),\xi\in(0,+\infty)\}$ such that
\begin{align*}
\lim_{s\rightarrow 0^{\small -}}(\chi_\LL(x+s+iv_{\xi}^{-}(s),\eta_\LL(x+s+iv_{\xi}^{-}(s))=(\chi_{III/IV}(\xi),\eta_{III/IV}(\xi)).
\end{align*}
The curves $\{s+iv_{\xi}^{+}(s):s\in(0,+\infty),\xi\in(0,+\infty)\}$ and $\{s+iv_{\xi}^-(s):s\in(-\infty,0),\xi\in(0,+\infty)\}$ satisfy the equation
\begin{align*}
G_2(s,v_{\xi}^{\pm}(s);x)=\xi
\end{align*}
where 
\begin{align*}
G_2(s,v;x)=1_{s>0}\int_{x}^{x+2s}\frac{f(t)dt}{(x+s-t)^2+v^2}+1_{s<0}\int_{x+2s}^{x}\frac{f(t)dt}{(x+s-t)^2+v^2}.
\end{align*}
\end{Prop}

\begin{rem}
\label{remSing}
Again, as in Proposition \ref{GeoSing1} and \ref{GeoSing4}, it is worth mentioning that the assumption that $x\in (\mathcal{S}_{nt}^{sing,III}(\mu)\cup  \mathcal{S}_{nt}^{sing,IV}(\mu))\cap(\Sn^{\circ}\cup\Siso)$ is never strictly used. In particular, the assumption that there exists an $\eps>0$ such that for all $y\in(x-\eps,x)\cup(x,x+\eps)$ 
\begin{align*}
\int_{-\delta}^{\delta}\frac{f(t)dt}{(y-t)^2}=+\infty.
\end{align*} 
implies that $((x-\eps,x)\cup(x,x+\eps))\cap (\mathcal{S}_{nt}^{sing,III}(\mu)\cup R_{\mu})=\varnothing$.
\end{rem}

\begin{proof}
We will first assume that $x\in \mathcal{S}_{nt}^{sing,III}(\mu)\cap\Sn^{\circ}$. Then $\int_{\R}\frac{f(t)dt}{(x-t)^2}<+\infty$ and $\mathcal{H}\varphi(x)\neq 0$  hold by the definition of $ \mathcal{S}_{nt}^{sing,III}$. We will begin by studying the limit of the integral
\begin{align}
\label{Eq5}
\int_{\R}\frac{f(t)dt}{(u(s)-t)^2+v_{\xi}(s)^2}
\end{align}
under a one parameter family of curves $w_{\xi}(s)=u(s)+iv_{\xi}(s)\in \Hp$ such that $\lim_{s\rightarrow 0^+}w_{\xi}(s)=x$. In particular we may take $u(s)=x+s$ and assume that $s>0$. The analysis of the case when $s<0$ is completely analogous to the case when $s>0$. The idea is to split (\ref{Eq5}) into two parts, one of which dominates the integrand, and apply Lebesgues dominated convergence theorem. Write
\begin{align}
\label{propGeoSing2Int}
\int_{\R}\frac{f(t)dt}{(x+s-t)^2+v^2}&=\int_{\vert x+s-t\vert\geq s}\frac{f(t)dt}{(x+s-t)^2+v^2}+\int_{x}^{x+2s}\frac{f(t)dt}{(x+s-t)^2+v^2}\nonumber\\
&=I(s,v;x)+G_2(s,v;x)
\end{align}
We now fix $s>0$. Then clearly, the function $G_2(s,v;x)$ is monotonically decreasing in $v$. By Fatou's lemma and equation (\ref{C1})
\begin{align*}
\liminf_{v\rightarrow 0^+}G_2(s,v;x)=\liminf_{v\rightarrow 0^+}\int_{x}^{x+2s}\frac{f(t)dt}{(x+s-t)^2+v^2}=+\infty
\end{align*}
for $s$ sufficiently small. By Lebesgue's dominated convergence theorem we also have
\begin{align*}
\lim_{v\rightarrow +\infty}G_2(s,v;x)=0.
\end{align*}
Therefore, since $G_2(s,v;x)$ is monotonically decreasing in $v$, the equation
\begin{align}
\label{G2Eq}
G_2(s,v;x)=\xi
\end{align}
has a unique solution $v_{\xi}(s)$ for all $\xi\in(0,+\infty)$. Differentiation under the integral sign gives
\begin{align*}
\frac{\partial G_2(s,v;x)}{\dv v}&=\int_{x}^{x+2s}\frac{\dv}{\dv v}\frac{1}{(x+s-t)^2+v^2}f(t)dt\\
&=-2\int_{x}^{x+2s}\frac{v}{((x+s-t)^2+v^2)^2}f(t)dt<0
\end{align*}
for all $v>0$. Hence the implicit function theorem implies that there exists a continuous path $(x+s,v_{\xi}(s))$ for $s\in(0,\delta)$ such that $G_2(s,v_{\xi}(s);x)=\xi$. Now, assume that $(x+s,v_{\xi}(s))$ contains a non-tangential subsequence $\{(x+s_j,v_{\xi}(s_j))\}_{j=0}^{+\infty}$, i.e., sequence such that there exists a $k>0$ such that $v_j=v_{\xi}(s_j)>ks_j$ for all $j$. This implies that
\begin{align*}
G_2(s_j,v_j;x)=\int_{x}^{x+2s_j}\frac{f(t)dt}{(x+s_j-t)^2+v_j^2}\leq \int_{x}^{x+2s_j}\frac{f(t)dt}{k^2s_j^2}=\frac{1}{k^2s_j^2}\int_{x}^{x+2s_j}f(t)dt.
\end{align*}
However, since $x\in \mathcal{S}_{nt}^{sing,III}(\mu)$ by assumption, for every $\eps>0$ there exists a $J$ such that whenever $j>J$ 
\begin{align*}
\int_{x}^{x+2s_j}\frac{f(t)}{(x-t)^2}dt<\eps.
\end{align*}
Since,
\begin{align*}
\int_{x}^{x+2s_j}\frac{f(t)}{(x-t)^2}dt\geq \frac{1}{4s_j^2}\int_{x}^{x+2s_j}f(t)dt
\end{align*}
we find that
\begin{align*}
G_2(s_j,v_j;x)\leq \frac{4\eps s_j^2}{k^2s_j^2}=\frac{4\eps}{k}.
\end{align*}
As $\eps$ was arbitrary this implies that $\lim_{j\rightarrow +\infty}G_2(s_j,v_j;x)=0$, a contradiction. Thus the path $(x+s,v_{\xi}(s))$ becomes tangential to the real axis, i.e. $\lim_{s\to 0^+}v_{\xi}(s)/s=0$. We now consider $I(s,v_{\xi}(s);x)$. Since $\vert x-t\vert<2\vert x+s-t\vert$ whenever $\vert x+s-t\vert\geq s$ we have
\begin{align*}
\frac{f(t)\chi_{\R\backslash[x,x+2s]}(t)}{(x+s-t)^2+v_{\xi}(s)^2}\leq \frac{4f(t)}{(x-t)^2},
\end{align*}
and since $\lim_{s\rightarrow 0^+}v_{\xi}(s)=0$, Lebesgue's dominated convergence theorem implies that
\begin{align*}
\lim_{s\rightarrow 0^+}\int_{\R}\frac{f(t)\chi_{\R\backslash[x,x+2s](t)}}{(x-s-t)^2+v_{\xi}(s)^2}=\int_{\R}\frac{f(t)dt}{(x-t)^2}.
\end{align*}
It now remains to study $\pi H_{v_{\xi}(s)}f(x+s)$ as $s\rightarrow 0^+$. We have
\begin{align*}
\pi H_{v_{\xi}(s)}f(x+s)&=\int_{\R}\frac{(x+s-t)f(t)dt}{(x+s-t)^2+v_{\xi}(s)^2}=\int_{\vert x+s-t\vert\geq s}\frac{(x+s-t)f(t)dt}{(x+s-t)^2+v_{\xi}(s)^2}\\&+\int_{x}^{x+2s}\frac{(x+s-t)f(t)dt}{(x+s-t)^2+v_{\xi}(s)^2}\\
&=J_1(s)+J_2(s).
\end{align*}
Note that for $t\in \R\backslash [x,x+2s]$ $\vert x-t\vert<2\vert x+s-t\vert$, which implies that
\begin{align*}
\frac{\vert x+s-t\vert f(t)}{(x+s-t)^2+v_{\xi}(s)^2}\leq \frac{2f(t)}{\vert x-t \vert}.
\end{align*}
Again by Lebesgue's dominated convergence theorem,
\begin{align*}
\lim_{s\rightarrow 0^+}J_1(s)=\lim_{s\rightarrow 0^+}\int_{\R}\frac{(x+s-t)f(t)\chi_{\R\backslash[x,x+2s](t)}}{(x-s-t)^2+v_{\xi}(s)}=\int_{\R}\frac{f(t)dt}{x-t}.
\end{align*}
Finally,
\begin{align*}
\vert J_2(s)\vert\leq  \int_{x}^{x+2s}\frac{\vert x+s-t\vert f(t)dt}{(x+s-t)^2+v_{\xi}(s)^2}\leq 2s\int_{x}^{x+2s}\frac{f(t)dt}{(x+s-t)^2+v_{\xi}(s)^2}=2s\xi,
\end{align*}
by (\ref{G2Eq}). Hence,
\begin{align*}
\lim_{s\rightarrow 0^+}J_2(s)=0.
\end{align*}
Altogether, this implies that
\begin{align*}
&\lim_{s\rightarrow 0^+}\frac{v_{\xi}(s)\big\{e^{\pi H_{v_{\xi}(s)}f(x+s)}-\cos(\pi P_{v_{\xi}(s)} f(x+s))\big\}}{\sin(\pi P_{v_{\xi}(s)} f(x+s))}=\lim_{s\rightarrow 0^+}\frac{e^{\pi H_{v_{\xi}(s)}f(x+s)}-\cos(\pi P_{v_{\xi}(s)} f(x+s))}{\displaystyle\int_{\R}\frac{f(t)dt}{(x+s-t)^2+v_{\xi}(s)^2}}\\
&=\frac{e^{\int_{\R}\frac{f(t)dt}{x-t}}-1}{\int_{\R}\frac{f(t)dt}{(x-t)^2}+\xi}
\end{align*}
and
\begin{align*}
&\lim_{s\rightarrow 0^+}\frac{v_{\xi}(s)\big\{1-e^{-\pi H_{v_{\xi}(s)}f(x+s)}\cos(\pi P_{v_{\xi}(s)} f(x+s))\big\}}{\sin(\pi P_{v_{\xi}(s)} f(x+s))}
=\frac{1-e^{-\int_{\R}\frac{f(t)dt}{x-t}}}{\int_{\R}\frac{f(t)dt}{(x-t)^2}+\xi}.
\end{align*}
Recall that the distributional derivative of the Cauchy principal value integral p.v.$\int_{\R}\frac{f(t)dt}{x-t}$ equals
\begin{align*}
\frac{d}{dx}\text{p.v.}\int_{\R}\frac{f(t)dt}{x-t}=-\text{f.p.}\int_{\R}\frac{f(t)dt}{(x-t)^2},
\end{align*}
where $\text{f.p.}\int_{\R}\frac{f(t)dt}{(x-t)^2}$ denotes Hadamard's finite part integral. However, as the integrals $\int_{\R}\frac{f(t)dt}{x-t}$ and $\int_{\R}\frac{f(t)dt}{(x-t)^2}$ exists in the ordinary sense we have that 
\begin{align*}
-\frac{d}{dx}\text{p.v.}\int_{\R}\frac{f(t)dt}{x-t}=-\pi (\mathcal{H}f)'(x)=\text{f.p.}\int_{\R}\frac{f(t)dt}{(x-t)^2}=\int_{\R}\frac{f(t)dt}{(x-t)^2}
\end{align*}
Using this we find that 
\begin{align*}
\lim_{s\rightarrow 0^+}\chi_\LL(w_{\xi}(s))=\chi_{III}(\xi)=x+\frac{1-e^{-\pi \mathcal{H}f(x)}}{\xi-\pi (\mathcal{H}f)'(x)}
\end{align*}
and
\begin{align*}
\lim_{s\rightarrow 0^+}\eta_\LL(w_{\xi}(s))=\eta_{III}(\xi)=1-\frac{e^{\pi \mathcal{H}f(x)}+e^{-\pi \mathcal{H}f(x)}-2}{\xi-\pi (\mathcal{H}f)'(x)}
\end{align*}
for each fixed $\xi\in(0,+\infty)$. In particular we note that this is a parametrization of a part of line given by the equation
\begin{align*}
\frac{\eta-1}{\chi-x}=-\frac{2-e^{\pi \mathcal{H}f(x)}-e^{-\pi \mathcal{H}f(x)}}{1-e^{-\pi \mathcal{H}f(x)}}=1-e^{\pi \mathcal{H}f(x)}.
\end{align*}
We now instead assume that $\int_{\R}\frac{1-f(t)dt}{(x-t)^2}<+\infty$ holds. Then
\begin{align*}
\sin(\pi P_{v} f(u))=\sin(\pi-\pi P_{v} (1-f)(u))=\sin(\pi P_{v}(1-f)(u))
\end{align*}
and
\begin{align*}
\cos(\pi P_{v} f(u))=\cos(\pi-\pi P_{v} (1-f)(u))=-\cos(\pi P_{v} (1-f)(u))
\end{align*}
In addition we can write
\begin{align*}
\pi H_{v}f(u)&=\int_{\R}\frac{(u-t)f(t)dt}{(u-t)^2+v^2}=\lim_{R\to \infty}\int_{-R}^{R}\frac{(u-t)dt}{(u-t)^2+v^2}-\int_{-R}^{R}\frac{(u-t)(1-f(t))dt}{(u-t)^2+v^2}\\
&=\lim_{R\to \infty}-\int_{-R}^{R}\frac{(u-t)(1-f(t))dt}{(u-t)^2+v^2}
\end{align*}
We can now apply the same analysis as before on the function $(1-f)$, giving an identical result.
\end{proof}

\begin{Cor}
Assume that  $x\in \Sn^{\circ}$ and that either
\begin{align*}
\int_{\R}\frac{f(t)dt}{(x-t)^2}<+\infty \quad (i)
\end{align*}
or
\begin{align*}
\int_{\R}\frac{(1-f(t))dt}{(x-t)^2}<+\infty \quad (ii)
\end{align*}
holds, and  that $\mathcal{H}f(x)\neq 0$, i.e., $x\in\mathcal{S}_{nt}^{sing,III}(\mu)\cup\in\mathcal{S}_{nt}^{sing,IV}(\mu)$. Then for every fixed $\xi\in(-\infty,+\infty)$ there exists a sequence $\{x+s_j+iv_{\xi}(s_j)\}_{j=1}^{+\infty}\in\Hp$, such that $\lim_{j\rightarrow +\infty}x+s_j+iv_{\xi}(s_j)=x$ and
\begin{align*}
\lim_{j\rightarrow +\infty}\chi_\LL(x+s_j+iv_{\xi}(s_j))&=\chi_{III\backslash IV}(\xi)\\
\lim_{j\rightarrow +\infty}\eta_\LL(x+s_j+iv_{\xi}(s_j))&=\eta_{III\backslash IV}(\xi).
\end{align*}
\end{Cor}

\begin{proof}
We may repeat the proof of Proposition \ref{GeoSing2} replacing a continuous path everywhere with a sequence $\{(x+s_j+v(s_j)\}_j$. The only difference is that since we are not assuming (\ref{C1})-(\ref{C6}) we may not conclude that there exists a solution to the equation $G_2(s,v;x)=\xi$ for every $s$ sufficiently small. However, since we are considering sequences instead of paths we can similarly to Corollary \ref{corGeoSing1} always find a sequence $s_j\rightarrow 0$ as $j\rightarrow +\infty$ such that $G_2(s_j,v;x)=\xi$. The rest of the proof remains the same.
\end{proof}
In general the equation $G_2(s,v;x)=\xi$ in Proposition \ref{GeoSing2} can of course not be solved explicitly. However there exists an important special case when $f(t)$ or $1-f(t)$ is convex in a neighborhood of the point $x$, when one can solve the equation $G_2(s,v;x)=\xi$ approximately.

\begin{Prop}
\label{propConvexDensity}
Let $x\in \mathcal{S}_{nt}^{sing,III}(\mu)$ and let $G_2(s,v;x)$ be the function defined in Proposition (\ref{GeoSing2}). Assume that there is an $\eps>0$ such that $f(t)$ is convex in $[x-\eps,x+\eps]$ and $f(x+2s)/f(x+s)$ is uniformly bounded for $\vert s\vert\leq \eps$. Fix $\xi>0$ and define 
\begin{align}
\label{eqConvexDensity1}
v(s)=\frac{\pi}{\xi}f(x+s)
\end{align}
for $\vert s\vert\leq \eps$. Then,
\begin{align}
\label{limConvexDensity1}
\lim_{s\to 0}G_2(s,v(s);x)=\xi.
\end{align}
Thus,
\begin{align}
\label{limConvexDensity2}
\lim_{s\to 0}\chi_\LL(x+s+iv(s))&=\chi_{III}(\xi)\\
\lim_{s\to 0}\eta_\LL(x+s+iv(s))&=\eta_{III}(\xi)\nonumber.
\end{align}
\end{Prop}

\begin{proof}
Consider the case $s>0$. Let 
\begin{align*}
I_1(s)&=\int_{x}^{x+2s}\frac{f(x+s)dt}{(x+s-t)^2+v(s)^2},\\
I_2(s)&=\int_{x+s}^{x+2s}\frac{(f(t)-f(x+s))dt}{(x+s-t)^2+v(s)^2}-\int_{x}^{x+s}\frac{(f(x+s)-f(t))dt}{(x+s-t)^2+v(s)^2}
\end{align*}
so that
\begin{align*}
G_2(x,v(s);x)=I_1(s)+I_2(s).
\end{align*}
Let $F(x)$ be a convex function on an interval $I$ and let $x,y,w\in I$ with $x<y<w$. Then
\begin{align}
\label{ineqConvex}
\frac{F(y)-F(x)}{y-x}\leq \frac{F(w)-F(x)}{w-x}
\end{align}
(see Proposition 1.25 in \cite{BSimon}). From this and $f(x)=0$ since $x\in \mathcal{S}_{nt}^{sing,III}(\mu)$ we see that $f(x+s)/s$ is an increasing function in $(0,\eps)$ and hence
\begin{align*}
a=\lim_{s\to 0^+}\frac{f(x+s)}{s}
\end{align*}
exists and is $\geq 0$. We must have $a=0$, since if $a>0$ then
\begin{align*}
\infty=\int_0^{\eps}\frac{a}{s}ds\leq \int_0^\eps\frac{f(x+s)}{s^2}ds\leq \int_\R\frac{f(t)}{(x-t)^2}dt,
\end{align*}
which contradicts $x\in \mathcal{S}_{nt}^{sing,III}(\mu)$. Thus, 
\begin{align}
\label{limConvex1}
\frac{v(s)}{s}=\frac{\pi}{\xi}\frac{f(x+s)}{s}\to 0
\end{align}
as $s\to 0^+$. It follows that 
\begin{align*}
I_1(s)=\frac{\xi}{\pi}\int_{x}^{x+2s}\frac{v(s)dt}{(x+s-t)^2+v(s)^2}=\frac{2\xi}{\pi}\arctan\frac{s}{v(s)}\to \xi
\end{align*}
as $s\to 0^+$. Hence, to prove (\ref{limConvexDensity1}) it remains to show that $I_2(s)\to 0^+$ as $s\to 0^+$. Notice that we can write 
\begin{align}
\label{intConvex}
I_2(s)&=\int_{0}^{s}(f(x+s+t)+f(x+s-t)-2f(x+s))\frac{dt}{t^2+v(s)^2}.
\end{align}
Since $f$ is convex in $[x-\eps,x+\eps]$ we see that for $0\leq s\leq \eps/2$,
\begin{align*}
\frac{1}{2}(f(x+s+t)+f(x+s-t))\geq f\bigg(\frac{x+s+t+x+s-t}{2}\bigg)=f(x+s)
\end{align*}
and consequently $I_2(s)\geq 0$. It follows from (\ref{ineqConvex}) that 
\begin{align*}
\frac{f(x+s)-f(x)}{s}\geq\frac{f(x+s-t)-f(x)}{s-t} 
\end{align*}
for $t\in[0,s)$ and since $f(x)=0$ we see that
\begin{align*}
f(x+s)\geq \frac{s}{s-t}f(x+s-t)\geq f(x+s-t).
\end{align*}
From (\ref{ineqConvex}) we also see that 
\begin{align*}
\frac{f(x+s+t)-f(x+s)}{t}\leq\frac{f(x+2s)-f(x+s)}{s}\leq \frac{f(x+2s)}{s} 
\end{align*}
for $t\in[0,s]$. Thus, 
\begin{align}
\label{ineqConvex2}
 f(x+s+t)+f(x+s-t)-2f(x+s)&=t\bigg(\frac{f(x+s+t)-f(x+s)}{t}\bigg)-(f(x+s)-f(x+s-t)) \nonumber \\&\leq t\frac{f(x+2s)}{s}\leq Ct\frac{f(x+s)}{s},
\end{align}
for some constant $C$. In the last estimate we used our assumption that $f(x+2s)/f(x+s)$ is uniformly bounded for $s\in[0,\eps]$. If we use (\ref{limConvex1}) together with the estimate (\ref{ineqConvex2}) in (\ref{intConvex}), we see that we have proved that 
\begin{align*}
0\leq I_2(s)\leq \frac{Cf(x+s)}{s}\int_0^s\frac{tdt}{t^2+v(s)^2}=\frac{C\xi v(s)}{\pi s}\int_{0}^s\frac{t dt}{t^2+v(s)^2}=\frac{C\xi v(s)}{2\pi s}\log\bigg(1+\frac{s^2}{v(s)^2}\bigg)\to 0
\end{align*}
as $s\to 0^+$ by (\ref{limConvex1}). This proves (\ref{limConvexDensity1}) and (\ref{limConvexDensity2}) follows as in the proof of Proposition \ref{GeoSing2}.
\end{proof}

\begin{rem}
In particular we note that the assumption that $\frac{f(x+2s)}{f(x+s)}$ is uniformly bounded in $s$ holds if $f(x+s)\sim g(s)\vert s\vert^{\alpha}$, for some positive and bounded function $g(s)$ and some $\alpha>0$ such that $g(s)\vert s\vert ^{\alpha}$ is convex in a neighborhood of $0$.
\end{rem}
In Propositions \ref{GeoSing1}-\ref{GeoSing2} we determined subsets of $\dv \mathcal{L}(x)$ when $x\in\bigcup_{A=I}^{IV}\mathcal{S}_{nt}^{sing,A}(\mu)$. We now want to show that under some additional assumptions on the density $f$, these sets are in fact all of $\dv \mathcal{L}(x)$.

\begin{Thm}
\label{GeoThm1}
Assume that $x\in\mathcal{S}_{nt}^{sing,III}(\mu)\cup\mathcal{S}_{nt}^{sing,IV}(\mu)$ and that the assumptions of Proposition \ref{GeoSing2} are satisfied. Furthermore,  assume that there exists sequences $\{r_n\}_n\subset G$ and $\{l_n\}_n\subset G$ of regular points such that $r_n>x$ and $l_n<x$ for all n and such that $\lim_{n\to \infty}r_n=\lim_{n\to \infty}l_n=x$. Finally assume that 
\begin{align*}
\max\{\sup_n\vert m_{\Hil f}(r_n)\vert,\sup_n\vert m_{\Hil f}(l_n)\vert\}<+\infty.
\end{align*}
Then, if $x\in \mathcal{S}_{nt}^{sing,III}(\mu)$
\begin{align*}
\partial \mathcal{L}(x)&=\overline{\{(\chi_{III}(\xi),\eta_{III}(\xi)):\xi\in(0,+\infty)\}},
\end{align*}
and 
if $x\in \mathcal{S}_{nt}^{sing,IV}(\mu)$
\begin{align*}
\partial \mathcal{L}(x)&=\overline{\{(\chi_{IV}(\xi),\eta_{IV}(\xi)):\xi\in(0,+\infty)\}},
\end{align*}
where the functions $\chi_{III/IV}(\xi)$ and $\eta_{III/IV}(\xi)$ are defined in Proposition \ref{GeoSing2}. In particular the assumptions holds if $x\in \mathscr{L}_{m_{\Hil f}}$ by a modification of the proof of Lemma \ref{Generic1}, and especially if $f\in \Hc$.
\end{Thm}

\begin{proof}
Let $x\in\mathcal{S}_{nt}^{sing,III}(\mu)$. We know from Proposition \ref{GeoSing2} that $A=\overline{\{(\chi_{III}(\xi),\eta_{III}(\xi):0<\xi<\infty\}}\subset \dv \LL(x)$ and we want to prove that equality holds. Let $w_n=u_n+iv_n\in \Hp$, $n\geq 0$, be any sequence such that $w_n\to x$ as $n\to \infty$. We want to show that all limit points of $(\chi_{\LL}(w_n),\eta_{\LL}(w_n))$ belong to $A$. By taking subsequences we can assume that $(\chi_{\LL}(w_n),\eta_{\LL}(w_n))$ converges. Set 
\begin{align*}
\xi_n=\int_{x}^{x+2(u_n-x)}\frac{f(t)dt}{(u_n-t)^2+v_n^2}.
\end{align*}
By taking a further subsequence we can assume that $\xi_n\to \xi\in[0,\infty]$ as $n\to \infty$ $u_n>x$ for all $n$. If $\xi\in[0,\infty)$, a repetition of the arguments of the proof of Proposition \ref{GeoSing2} gives
\begin{align*}
\lim_{n\to \infty}(\chi_{\LL}(w_n),\eta_{\LL}(w_n))=\bigg(x+\frac{1-e^{-\pi\mathcal{H}f(x)}}{\xi-\pi(\mathcal{H}f)'(x)},1-\frac{e^{\pi\mathcal{H}f(x)}+e^{-\pi\mathcal{H}f(x)}-2}{\xi-\pi(\mathcal{H}f)'(x)}\bigg).
\end{align*}
It remains to consider the case $\xi=\infty$. We want to show that in this case $(\chi_{\LL}(w_n),\eta_{\LL}(w_n))\to (x,1)$ as $n\to \infty$. Since $\xi_n\to \infty$, for every $\xi>0$ there is an $N(\xi)$ such that 
\begin{align}
\label{IneqXi1}
\int_{x}^{x+2(u_n-x)}\frac{f(t)dt}{(u_n-t)^2+v_n^2}>\xi
\end{align}
whenever $n>N(\xi)$. Let $v_{\xi}(s)$ be the continuous function defined in Proposition \ref{GeoSing2}. Then the inequality (\ref{IneqXi1}) above implies that
\begin{align*}
G_2(u_n-x,v_n;x)>G_2(u_n-x,v_{\xi}(u_n-x);x).
\end{align*}
Since the function $G_2(s,v;x)$ is monotonically decreasing in $v$ this implies that $v_n<v_{\xi}(u_n-x,v_n;x)$ for all $n>N(\xi)$. This implies that the sequence $\{w_n\}_n$ is trapped inside the set
\begin{align*}
\{(u,v)\in\Hp: x\leq u<u_{N(\xi)}, 0<v<v_{\xi}(u-x)\}
\end{align*}
whenever $n>N(\xi)$. In particular for every $n$ there exists an $r_{k_n}\in G$ such that $r_{k_n}>u_n$ and $\lim_{n\to\infty}r_{k_n}=x$. Let $X_n^{(\xi)}$ be the open set
\begin{align*}
X_n^{(\xi)}=\{(u,v)\in \R^2: x<u< r_n, 0< v<v_{\xi}(u-x)\}.
\end{align*}
In particular $w_n$ belongs to $X_n^{(\xi)}$ for every $n>N(\xi)$. 
Let $T_n$ be the closed set, whose boundary equals
\begin{align*}
\dv T_n&=\{(t,1):x\leq t\leq r_{k_n}\}\bigcup \{W_{\mathcal{L}}^{-1}(r_{k_n}+it):0<t\leq v_{\xi}(r_{k_n}-x)\}
\bigcup \{W_{\mathcal{L}}^{-1}(t+iv_{\xi}(t)):0<t\leq r_{k_n}-x\}\\&\bigcup\bigg\{\bigg(x+\frac{1-e^{-\pi\mathcal{H}f(x)}}{\xi'-\pi(\mathcal{H}f)'(x)},1-\frac{e^{\pi\mathcal{H}f(x)}+e^{-\pi\mathcal{H}f(x)}-2}{\xi'-\pi(\mathcal{H}f)'(x)}\bigg): \xi'\geq \xi\bigg\}\\
&:=\dv T_n^1\cup\dv T_n^2\cup\dv T_n^3\cup\dv T_n^4.
\end{align*}
We now show that $\overline{W_{\mathcal{L}}^{-1}(X_n^{(\xi)})}\subset T_n$. Since $x\in \Sn^{\circ}$, it follows from Lemma \ref{TangTop} that all points of $\dv T_n^1\subset \dv W_{\mathcal{L}}^{-1}(X_n^{(\xi)})$. Since $r_{k_n}\in G$, it follows that $\lim_{v\to 0^+}(\chi_{\LL}(r_{k_n}+iv),\eta_{\LL}(r_{k_n}+iv))=(r_{k_n},1)$. Hence $\dv T_n^2\subset \dv W_{\mathcal{L}}^{-1}(X_n^{(\xi)})$. By Proposition \ref{GeoSing2}, $(\dv T_n^3\cup \dv T_n^4)\subset \dv W_{\mathcal{L}}^{-1}(X_n^{(\xi)})$. On the other hand, by Lemma \ref{TangTop} and the fact that $W_{\mathcal{L}}^{-1}$ is a homeomorphism, we have that $W_{\mathcal{L}}^{-1}(X_n^{(\xi)})\subset T_n^{\circ}$. Thus $\overline{W_{\mathcal{L}}^{-1}(X_n^{(\xi)})}\subset T_n$. 
This fact follows from Lemma \ref{TangTop} and Proposition \ref{GeoSing2} and the assumption that $r_{k_n}$ is a regular point in $\Sn^{\circ}$. In particular, $(\chi_{\LL}(w_n),\eta_{\LL}(w_n))\in T_n$ for every $n>N(\xi)$. The trapping regions $T_n$ are illustrated in figure \ref{Sequence2}. We will now show that $\limsup_{n\to+\infty}d((x,1),\dv T_n)\leq C/\xi$ for some positive constant $C$ independent of $\xi$, which implies that $\lim_{n\to+\infty}(\chi_\LL(w_n),\eta_\LL(w_n))=(x,1)$. Recall that
\begin{align*}
d((x,1),\dv T_n)=\sup_{(x',y')\in \dv T_n}d((x,1),(x',y')).
\end{align*}
Clearly,
\begin{align*}
d((x,1),\dv T_n^1)=r_{k_n}-x\rightarrow 0
\end{align*}
as $n\to +\infty$. Similarly, from the proof of Proposition \ref{GeoSing2} it follows that
\begin{align*}
\lim_{n\to\infty}d((x,1),\dv T_n^3)=\bigg\vert\bigg(\frac{1-e^{-\pi\mathcal{H}f(x)}}{\xi-\pi(\mathcal{H}f)'(x)},\frac{e^{\pi\mathcal{H}f(x)}+e^{-\pi\mathcal{H}f(x)}-2}{\xi-\pi(\mathcal{H}f)'(x)}\bigg)\bigg\vert
\end{align*}
and 
\begin{align*}
d((x,1),\dv T_n^4)=\bigg\vert\bigg(\frac{1-e^{-\pi\mathcal{H}f(x)}}{\xi-\pi(\mathcal{H}f)'(x)},\frac{e^{\pi\mathcal{H}f(x)}+e^{-\pi\mathcal{H}f(x)}-2}{\xi-\pi(\mathcal{H}f)'(x)}\bigg)\bigg\vert.
\end{align*}
We now estimate $d((x,1),\dv T_n^2)$. This is the most subtle part of the proof, and here the choice of the sequence $\{r_{k_n}\}_n$ is critical. By assumption the sequence $\{m_{\mathcal{H}f}(r_{k_n})\}_n$ is bounded and hence by estimate (\ref{nontangMax}) $\vert H_{v}f(r_{k_n})\vert $, is uniformly bounded. Thus, there is a constant $C'$ independent of $\xi$ such that for all $n$ and $0<v<v_{\xi}(r_{k_n}-x)$, 
\begin{align*}
v\bigg\vert \frac{1-e^{-\pi H_{v}f(r_{k_n})}\cos(\pi P_{v}f(r_{k_n}))}{\sin(\pi P_{v}f(r_{k_n}))}\bigg\vert \leq v\frac{C'}{\vert \sin(\pi P_{v}f(r_{k_n}))\vert}
\end{align*} 
and 
\begin{align*}
v\bigg\vert \frac{e^{\pi H_{v}f(r_{k_n})}+e^{-\pi H_{v}f(r_{k_n)}}-2\cos(\pi P_{v}(r_{k_n}))}{\sin(\pi P_{v}f(r_{k_n}))}\bigg\vert \leq v\frac{C'}{\vert \sin(\pi P_{v}f(r_{k_n}))\vert}
\end{align*} 
for all $n$. In addition,
\begin{align*}
\frac{v}{\vert \sin( \pi P_{v}f(r_{k_n}))\vert}&\leq \frac{v}{\min\{P_{v}f(r_{k_n}), P_{v}(1-f)(r_{k_n}))\}}\\
&\leq \frac{1}{\min\{v^{-1}P_{v}f(r_{k_n}),v^{-1}P_{v}(1-f)(r_{k_n}))\}}\\
&\leq \frac{1}{\min\{v_{\xi}(r_{k_n}-x)^{-1}P_{v_{\xi}(r_{k_n}-x)}f(r_{k_n}),v_{\xi}(r_{k_n}-x)^{-1}P_{v_{\xi}(r_{k_n}-x)}(1-f)(r_{k_n}))\}}
\end{align*} 
for all $0<v<v_{\xi}(r_{k_n}-x)$, by the monotonicity of the function $v^{-1}\pi P_{v} f(r_{k_n})$. By the same argument that was used to control (\ref{propGeoSing2Int}) in the proof of Proposition \ref{GeoSing2} we see that
\begin{align*}
\int_{\R} \frac{f(t)dt}{(r_{k_n}-t)^2+v_{\xi}(r_{k_n}-x)^2}&=\int_{\vert r_{k_n}-t\vert>r_{k_n}-x} \frac{f(t)dt}{(r_{k_n}-t)^2+v_{\xi}(r_{k_n}-x)^2}+G_2(r_{k_n}-x,v_{\xi}(r_{k_n}-x);x)\\
&\to \xi+\int_{\R}\frac{f(t)dt}{(x-t)^2},
\end{align*} 
as $n\to \infty$. Furthermore, 
\begin{align*}
\int_{\R} \frac{(1-f(t))dt}{(r_n-t)^2+v_{\xi}(r_n-x)^2}&=\int_{\R} \frac{dt}{(r_n-t)^2+v_{\xi}(r_n-x)^2}-\xi+O(1)\\
&=\frac{\pi}{v_{\xi}(r_n-x)}-\xi+O(1)\geq \xi 
\end{align*} 
whenever $n$ is sufficiently large. Hence,
\begin{align*}
\limsup_{n\to\infty}d((x,1),\dv T_n^2)\leq \frac{C}{\xi}.
\end{align*}
Combining our estimates we have proved that there is a constant $C$ such that 
\begin{align*}
\limsup_{n\to\infty}d\Big((x,1),\dv T_n)\Big)\leq \frac{C}{\xi}.
\end{align*}
Since $\xi\in[0,\infty)$ was arbitrary, the result follows. The case when $x\in \mathcal{S}_{nt}^{sing,IV}(\mu)$ is analogus. 
\end{proof}
\begin{figure}[h]
\centering
\begin{tikzpicture}
\draw [green, thick, domain=-10:-4] plot (\x, {0.1*(\x+10)*(\x+10)}) node[black,right] {$v_\xi(s)$};
\draw [blue] (-6,{0.1*(-6+10)*(-6+10)})--(-6,0) node[black,below] {\small$r_n$};
\draw [blue] (-8,{0.1*(-8+10)*(-8+10)})--(-8,0) node[black,below] {\small$r_{n+1}$};
\draw[thick,->] (-10,0)--(-5,0);
\draw[thick,->] (-10,0)--(-10,4);
\draw (-10,-0.5) node {\small$(x,0)$};
\draw[thick,->] (-2,2)--(5,2);
\filldraw[black] (-7,0.3) circle (0.03cm); 
\filldraw[black] (-7.5,0.4) circle (0.03cm); 
\filldraw[black] (-6.3,1) circle (0.03cm); 
\filldraw[black] (-5.5,0.5) circle (0.03cm); 
\filldraw[black] (-8.3,0.1) circle (0.03cm); 
\draw[->] (-4,1)--(-3,1);
\draw (-3.5,1.3) node {$W_{\LL}^{-1}$};
\draw[red] (0,2)--(1,-1);
\draw[green] (0.20,{2-0.5}) to [out=0,in=135] (3,0);
\draw[blue] (2.9,0.1) to [out=45,in=180] (2,2);
\draw[blue] (1.2,1.32) to [out=45,in=180] (1,2);
\draw (0,2.3) node {\small$(x,1)$};
\filldraw[black] (3.2,1.6) circle (0.03cm); 
\filldraw[black] (2.5,1) circle (0.03cm); 
\filldraw[black] (1.8,1.6) circle (0.03cm); 
\filldraw[black] (1.2,1.9) circle (0.03cm); 
\filldraw[black] (0.7,1.85) circle (0.03cm); 
\draw (2.2,2.2) node {\tiny$(r_n,1)$};
\draw (1.2,2.2) node {\tiny$(r_{n+1},1)$};
\draw[->] (2,-0.4) to [out=45,in=270] (1,1.2);
\draw (2,-0.6) node {\scriptsize$W_{\LL}^{-1}(s+iv_{\xi}(s))$};
\draw[->] (4,1) to [out=170,in=60] (3,0.9);
\draw (5,0.9) node {\scriptsize$W_{\LL}^{-1}(r_n+is)$};
\end{tikzpicture}

\caption{This figure illustrates the trapping regions $T_n$. The black dots represent the positions of the sequence $\{u_n+iv_n\}_{n=1}^{+\infty}$}
\label{Sequence2}
\end{figure}
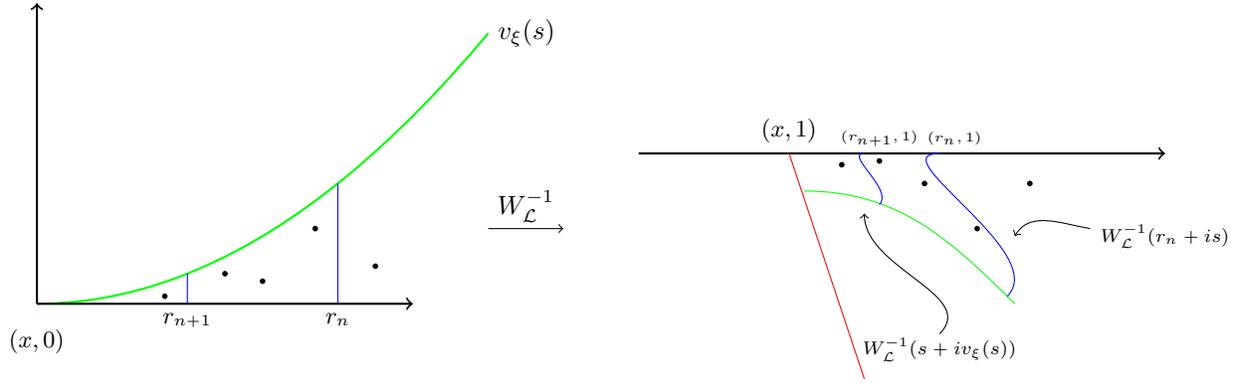

\begin{Thm}
\label{GeoThm2}
Let $x\in\mathcal{S}_{nt}^{sing,I}(\mu)\cup\mathcal{S}_{nt}^{sing,II}(\mu)$ and that the assumptions of Proposition \ref{GeoSing1} and Proposition \ref{GeoSing4} are satisfied. Furthermore,  assume that there exists sequences $\{r_n\}_n\subset G$ and $\{l_n\}_n\subset G$ of regular points such that $r_n>x$ and $l_n<x$ for all n and such that $\lim_{n\to \infty}r_n=\lim_{n\to \infty}l_n=x$. Finally assume that 
\begin{align*}
\max\{\sup_n\vert m_{\Hil \varphi}(r_n)\vert,\sup_n\vert m_{\Hil \varphi}(l_n)\vert\}=m<+\infty.
\end{align*}
Then,
\begin{align*}
\partial \mathcal{L}(x)&=\overline{\bigg\{\bigg(x,1+\frac{\delta e^{\pi \mathcal{H}\varphi(x)}}{1+\xi}\bigg):\xi\in[0,+\infty)\bigg\}}
\end{align*}
if $x\in\mathcal{S}_{nt}^{sing,I}(\mu)$
and
\begin{align*}
\partial \mathcal{L}(x)&=\overline{\bigg\{\bigg(x+\frac{\delta e^{-\pi \mathcal{H}\varphi(x)}}{1+\xi}\,1\frac{\delta e^{-\pi \mathcal{H}\varphi(x)}}{1+\xi}\bigg):\xi\in[0,+\infty)\bigg\}}
\end{align*}
if $x\in\mathcal{S}_{nt}^{sing,II}(\mu)$. In particular the assumptions holds if $x\in \mathscr{L}_{m_{\Hil \varphi}}$, and especially if $f\in \Hc$. 
\end{Thm}

\begin{proof}
The proof is similar to that of Theorem \ref{GeoThm1}. Let $x\in\mathcal{S}_{nt}^{sing,I}(\mu)$. We know from Proposition \ref{GeoSing1} that $A=\overline{\{(\chi_{I}(\xi),\eta_{I}(\xi):0<\xi<\infty\}}\subset \dv \LL(x)$ and we want to prove that equality holds. Let $w_n=u_n+iv_n\in \Hp$, $n\geq 0$, be any sequence such that $w_n\to x$ as $n\to \infty$. We want to show that all limit points of $(\chi_{\LL}(w_n),\eta_{\LL}(w_n))$ belong to $A$. By taking subsequences we can assume that $(\chi_{\LL}(w_n),\eta_{\LL}(w_n))$ converges. Set 
\begin{align*}
\xi_n=\int_{x}^{x+2(u_n-x)}\frac{(u_n-x)\varphi(t)dt}{(u_n-t)^2+v_n^2}.
\end{align*}
By taking a further subsequence we can assume that $\xi_n\to \xi\in[0,\infty]$ as $n\to \infty$ and that $u_n>x$ for all $n$. As in the proof of Theorem \ref{GeoThm1} it is enough to consider the case $\xi=\infty$. We want to show that in this case $(\chi_{\LL}(w_n),\eta_{\LL}(w_n))\to (x,1)$ as $n\to \infty$. Since $\xi_n\to \infty$, for every $\xi>0$ 
there exists an $N(\xi)$ such that 
\begin{align}
\label{IneqXi2}
\int_{x}^{x+2(u_{n}-x)}\frac{(u_n-x)\varphi(t)dt}{(u_{n}-t)^2+v_{n}^2}>\xi
\end{align}
whenever $n>N(\xi)$. Let $v_{\xi}(s)$ be the continuous function defined in Proposition \ref{GeoSing1}. Then inequality (\ref{IneqXi2}) implies that
\begin{align*}
G_1(u_n-x,v_n;x)>G_1(u_n-x,v_{\xi}(u_n-x);x).
\end{align*}
Since the function $G(s,v;x)$ is monotonically decreasing in $v$ this implies that $v_n<v_{\xi}(u_n-x,v_n;x)$ for all $n>N(\xi)$. This implies that the sequence $\{w_n\}_n$ is contained inside the set
\begin{align*}
\{(u,v)\in\Hp: x\leq u<u_{N(\xi)}, 0<v<v_{\xi}(u-x)\}
\end{align*}
whenever $n>N(\xi)$. In particular for every $n$ there exists an $r_{k_n}\in G$ such that $r_{k_n}>u_n$ and $\lim_{n\to\infty}r_{k_n}=x$. Let $X_n^{(\xi)}$ be the open set
\begin{align*}
X_n^{(\xi)}=\{(u,v)\in \R^2: x<u< r_{k_n}, 0< v<v_{\xi}(u-x)\}.
\end{align*}
In particular $w_n$ belongs to $X_n^{(\xi)}$ for every $n>N(\xi)$. A similar arguments to that in the proof of Theorem \ref{GeoThm1} shows $\dv W_{\mathcal{L}}^{-1}(X_n^{(\xi)})\subset T_n$, where $T_n$ is the closed set whose boundary equals
\begin{align*}
\dv T_n&=\{(t,1):x\leq t\leq r_{k_n}\}\bigcup \{W_{\mathcal{L}}^{-1}(r_{k_n}+it):0<t\leq v_{\xi}(r_{k_n}-x)\}
\bigcup \{W_{\mathcal{L}}^{-1}(t+iv_{\xi}(t)):0<t\leq r_{k_n}-x\}\\&\bigcup\bigg\{\bigg(x,1+\frac{\delta e^{\pi\mathcal{H}\varphi(x)}}{\xi'+1}\bigg):  \xi'\geq \xi\bigg\}\\
&:=\dv T_n^1\cup\dv T_n^2\cup\dv T_n^3\cup\dv T_n^4
\end{align*}
In particular, $(\chi(w_n),\eta(w_n))\in T_n$ for every $n>N(\xi)$. We will now show that $\lim_{n\to+\infty}d((x,1),\dv T_n)=C/\xi$ for some positive constant $C$ independent of $\xi$, which implies that $\lim_{n\to+\infty}(\chi(w_n),\eta(w_n))=(x,1)$. Clearly,
\begin{align*}
d((x,1),\dv T_n^1)=r_{k_n}-x\rightarrow 0
\end{align*}
as $n\to +\infty$. Similarly, from the proof of Proposition \ref{GeoSing1} it follows that
\begin{align*}
\limsup_{n\to\infty}d((x,1),\dv T_n^3)=\bigg\vert \bigg(x,1+\frac{\delta e^{\pi\mathcal{H}\varphi(x)}}{\xi+1}\bigg)\bigg\vert
\end{align*}
and 
\begin{align*}
d((x,1),\dv T_n^4)=\bigg\vert\bigg(\frac{1-e^{-\pi\mathcal{H}f(x)}}{\xi-\pi(\mathcal{H}f)'(x)},\frac{e^{\pi\mathcal{H}f(x)}+e^{-\pi\mathcal{H}f(x)}-2}{\xi-\pi(\mathcal{H}f)'(x)}\bigg)\bigg\vert.
\end{align*}
We now estimate $d((x,1),\dv T_n^2)$. As in Theorem \ref{GeoThm1}, this is the most subtle part of the proof, and here the choice of the sequence $\{r_{k_n}\}_n$ is critical. A computation gives
\begin{align*}
\pi H_vf(r_{k_n})=-\log\sqrt{(r_{k_n}-x)^2+v^2}+\log\sqrt{(r_{k_n}+\delta-x)^2+v^2}+H_v\varphi(r_{k_n})
\end{align*}
and
\begin{align*}
\pi P_vf(r_{k_n})&=-\arctan\bigg(\frac{r_{k_n}-x}{v}\bigg)+\arctan\bigg(\frac{r_{k_n}-x+\delta}{v}\bigg)+\pi P_v\varphi(r_{k_n})\\
&=\arctan\bigg(\frac{v}{r_{k_n}-x}\bigg)-\arctan\bigg(\frac{v}{r_{k_n}-x+\delta}\bigg)+\pi P_v\varphi(r_{k_n})
\end{align*}
Hence, using the trigonometric identity $\sin(\arctan{x}+y)=(1+x^2)^{-1/2}(x\cos y+\sin y)$ we get
\begin{align}
\label{Trig}
\sin(\pi P_v f(r_{k_n}))&=\frac{1}{\sqrt{1+\frac{(r_{k_n}-x)^2}{v^2}}} \bigg( \cos[\pi P_v\varphi(r_{k_n})-O(v)]+\frac{r_{k_n}-x}{v}\sin[\pi P_v\varphi(r_{k_n})-O(v)]\bigg)\\
&=\frac{1}{\sqrt{1+\frac{(r_{k_n}-x)^2}{v^2}}} \bigg( \cos[\pi P_v\varphi(r_{k_n})]+\frac{r_{k_n}-x}{v}\sin[\pi P_v\varphi(r_{k_n})]+O(v)+O(r_{k_n}-x)\bigg)\nonumber
\end{align}
By assumption the sequence $\{m_{\mathcal{H}\varphi}(r_{k_n})\}_n$ is bounded and hence by estimate (\ref{nontangMax}) $\vert H_{v}\varphi(r_{k_n})\vert $, is uniformly bounded. Thus there exists a constant $m$, independent of $\xi$ such that
\begin{align*}
\vert (\chi_\LL(r_{k_n},v)-r_{k_n},\eta_\LL(r_{k_n},v)-1)\vert &\leq \frac{\sqrt{(r_{k_n}+\delta)^2+v^2}}{\sqrt{(r_{k_n}-x)^2+v^2}}\frac{v\sqrt{5}(1+e^{m})}{\sin(\pi P_v \varphi(r_{k_n}))}\\
&\leq \frac{C}{\sqrt{1+\frac{(r_{k_n}-x)^2}{v^2}}}\frac{1}{\sin(\pi P_v \varphi(r_{k_n}))}
\end{align*}
where $C$ does not depend on $n$. Using (\ref{Trig}), one gets 
\begin{align*}
&\sqrt{1+\frac{(r_{k_n}-x)^2}{v^2}}\vert \sin(\pi P_v \varphi(r_{k_n}))\vert \geq \frac{r_{k_n}-x}{v}\sin[\pi P_v\varphi(r_{k_n})]-1\\&\geq \frac{r_{k_n}-x}{2v}\min\{P_v\varphi(r_{k_n}),P_v\ast(1-\varphi)(r_{k_n})\}-1\\
&=\frac{1}{2\pi}\min\bigg\{\int_{\R}\frac{(r_{k_n}-x)\varphi (t)dt}{(r_{k_n}-t)^2+v^2},\int_{\R}\frac{(r_{k_n}-x)(1-\varphi(t))dt}{(r_{k_n}-t)^2+v^2}\bigg\}-1.
\end{align*}
Using that
\begin{align*}
\int_{\R}\frac{(r_{k_n}-x)\varphi(t) dt}{(r_{k_n}-t)^2+v^2}=G_1(r_{k_n}-x,v;x)+\int_{\R\backslash[x,x+2(r_{k_n}-x)]}\frac{(r_n-x)\varphi(t)dt}{(r_{k_n}-t)^2+v^2},
\end{align*}
we can estimate the second term to get
\begin{align*}
&\bigg\vert\int_{\R\backslash[x,x+2(r_{k_n}-x)]}\frac{(r_{k_n}-x)\varphi(t)dt}{(r_{k_n}-t)^2+v^2}\bigg\vert\leq 2(r_{k_n}-x)\int_{(r_{k_n}-x)}^{+\infty}\frac{dt}{t^2}\leq 2
\end{align*}
Moreover, by monotonicity one has that $G_1(r_{k_n}-x,v;x)\geq G_1(r_{k_n}-x,v_{\xi}(r_{k_n}-x);x)$. This implies that for $v\leq v_{\xi}(r_{k_n}-x)$
\begin{align*}
&\min\bigg\{\int_{\R}\frac{(r_{k_n}-x)\varphi (t)dt}{(r_{k_n}-t)^2+v^2},\int_{\R}\frac{(r_{k_n}-x)(1-\varphi(t))dt}{(r_{k_n}-t)^2+v^2}\bigg\}\\&\geq \min\bigg\{G_1(r_{k_n}-x,v_{\xi}(r_{k_n}-x);x),\frac{2}{v}\arctan\bigg(\frac{r_{k_n}-x}{v}\bigg)-G_1(r_{k_n}-x,v_{\xi}(r_{k_n}-x);x)\bigg\}-2\geq \xi-2
\end{align*}
for $n$ sufficently large. In particular this implies that 
\begin{align*}
\vert (\chi_\LL(r_{k_n},v)-r_{k_n},\eta_{\LL}(r_{k_n},v)-1)\vert \leq \frac{C}{\sqrt{1+\frac{(r_{k_n}-x)^2}{v^2}}}\frac{1}{\xi-9}
\end{align*}
where $\xi>9$. Since $\xi>9$ was arbitrary we get
\begin{align*}
\limsup_{n\to\infty}d((x,1),\dv T_n^2)=0
\end{align*}
This implies that $\lim_{n\to \infty }(\chi_\LL(w_n),\eta_\LL(w_n))=\{(x,1)\}$. The case when $x\in \mathcal{S}_{nt}^{sing,II}(\mu)$ is analogus. 
\end{proof}

\begin{Prop}
\label{propHausdorff}
Let $\mathcal{H}^1$ denote the one dimensional Hausdorff measure. Then there exists a $\mu\in \mathcal{M}_{c,1}^{\lambda}(\R)$, such that $\mathcal{H}^1(\dv \mathcal{L})=+\infty.$
\end{Prop}

\begin{proof}
Let 
\begin{align*}
f_1(t)=t^2\sin^2(t^{-1})\chi_{[-\frac{2}{3\pi},\frac{2}{3\pi}]}(t)
\end{align*}
A Taylor expansion around $t=\frac{1}{k\pi}$ for $\vert k\vert=2,3,4,... $ shows that
\begin{align*}
\int_{\R}\frac{f_1(t)dt}{(\frac{1}{k\pi}-t)^2}<+\infty.
\end{align*}
Moreover,
\begin{align*}
\int_{\R}\frac{f_1(t)dt}{t^2}<C+\int_{-\frac{2}{3\pi}}^{\frac{2}{3\pi}}\frac{t^2dt}{t^2}<+\infty.
\end{align*}
We now show that in fact $\displaystyle\sup_{n\geq 2}\big\{\int_{\R}\frac{f_1(t)dt}{(\frac{1}{\pi n}-t)^2}\big\}<+\infty$. Write
\begin{align*}
\int_{-2/(3\pi)}^{2/(3\pi)}\frac{t^2\sin^2(t^{-1})dt}{(t-\frac{1}{\pi n})^2}&=\int_{-2/(3\pi)}^{1/(\pi n+\pi/2)}\frac{t^2\sin^2(t^{-1})dt}{(t-\frac{1}{\pi n})^2}+\int_{1/(\pi n+\pi/2)}^{1/(\pi n-\pi/2)}\frac{t^2\sin^2(t^{-1})dt}{(t-\frac{1}{\pi n})^2}+\int_{1/(\pi n-\pi/2)}^{2/(3\pi)}\frac{t^2\sin^2(t^{-1})dt}{(t-\frac{1}{\pi n})^2}\\
&:=I_1^{(n)}+I_2^{(n)}+I_3^{(n)}.
\end{align*}
We first estimate $I_1^{(n)}$. We get
\begin{align*}
I_1^{(n)}\leq \int_{-2/(3\pi)}^{1/(\pi n+\pi/2)}\frac{t^2dt}{(t-\frac{1}{\pi n})^2}&=\bigg[t-\frac{1}{(\pi n)^2}\frac{1}{(t-\frac{1}{\pi n})}+\frac{2}{\pi n}\log \bigg\vert t-\frac{1}{\pi n}\bigg\vert\bigg]_{-2/(3\pi)}^{1/(\pi n+\pi/2)}\\
&=\frac{2}{3\pi}+O(n^{-1}\log n).
\end{align*}
Similarly we get
\begin{align*}
I_3^{(n)}\leq \frac{2}{3\pi}+O(n^{-1}\log n).
\end{align*}
We now consider $I_2^{(n)}$. Using that $\vert \sin t\vert\leq \vert t\vert$ for all $t$ we get 
\begin{align*}
I_2^{(n)}=&\int_{1/(\pi n+\pi/2)}^{1/(\pi n-\pi/2)}\frac{t^2\sin^2(t^{-1})dt}{(t-\frac{1}{\pi n})^2}=(\pi n)^2\int_{-\pi /2}^{\pi/2} \frac{\sin^2(x)dx}{(x+\pi n)^2x^2}\leq (\pi n)^2\int_{-\pi /2}^{\pi/2} \frac{x^2dx}{(x+\pi n)^2x^2}\\
&\leq\frac{1}{(1-\frac{1}{2 n})^2}.
\end{align*}
Thus,
\begin{align}
\label{derivPoissonEst}
\sup_{n\geq 2}\bigg\{\int_{\R}\frac{f_1(t)dt}{(\frac{1}{\pi n}-t)^2}\bigg\}<+\infty.
\end{align}
We now consider $\mathcal{H}f_1((\pi n)^{-1})$. Note that $\mathcal{H}f_1(0)=0$, due to the symmetry of $f_1$ at $0$. We now consider the sign of $\mathcal{H}f_1((\pi n)^{-1})$ for $n=2,3,4,...$
\begin{align*}
\pi \mathcal{H}f_1((\pi n)^{-1})&=\int_{-2/(3\pi)}^0\frac{t^2\sin^2(t^{-1})dt}{\frac{1}{n\pi}-t}+\int_0^{2/(3\pi)}\frac{t^2\sin^2(t^{-1})dt}{\frac{1}{n\pi}-t}=\bigg\{t=\frac{1}{x}\bigg\}\\
&=\pi n\int_{-\infty}^{-3\pi/2}\frac{\sin^2(x)dx}{x^3(x-\pi n)}+\pi n\int_{3\pi/2}^{+\infty}\frac{\sin^2(x)dx}{x^3(x-\pi n)}\\
&=\pi n\int_{3\pi/2}^{+\infty}\frac{\sin^2(x)}{x^3}\bigg\{\frac{1}{x-\pi n}+\frac{1}{x+\pi n}\bigg\}dx\\
&=2\pi n\int_{3\pi/2}^{+\infty}\frac{\sin^2(x)dx}{x^2(x-\pi n)(x+\pi n)}:=2\pi nI(n),
\end{align*}
and
\begin{align*}
I(n)&=\int_{3\pi/2}^{+\infty}\frac{\sin^2(x)dx}{x^2(x-\pi n)(x+\pi n)}=\int_{3\pi/2}^{\pi n}\frac{\sin^2(x)dx}{x^2(x-\pi n)(x+\pi n)}+\int_{\pi n}^{2\pi n-3\pi/2}\frac{\sin^2(x)dx}{x^2(x-\pi n)(x+\pi n)}\\&+\int_{2\pi n-3\pi/2}^{+\infty}\frac{\sin^2(x)dx}{x^2(x-\pi n)(x+\pi n)}.
\end{align*}
Furthermore,
\begin{align*}
&\int_{3\pi/2}^{\pi n}\frac{\sin^2(x)dx}{x^2(x-\pi n)(x+\pi n)}+\int_{\pi n}^{2\pi n-3\pi/2}\frac{\sin^2(x)dx}{x^2(x-\pi n)(x+\pi n)}\\
&=\int_{3\pi/2-\pi n}^{0}\frac{\sin^2(t+\pi n)dt}{(t+\pi n)^2t(t+2\pi n)}+\int_{0}^{\pi n-3\pi/2}\frac{\sin^2(t+\pi n)dt}{(t+\pi n)^2t(t+2\pi n)}\\
&=\int_{0}^{\pi n-3\pi/2}\frac{\sin^2(t)dt}{(-t+\pi n)^2(-t)(-t+2\pi n)}+\int_{0}^{\pi n-3\pi/2}\frac{\sin^2(t+\pi n)dt}{(t+\pi n)^2t(t+2\pi n)}\\
&=\int_{0}^{\pi n-3\pi/2}\frac{\sin^2(t)}{t}\bigg\{\frac{1}{(t+\pi n)^2(t+2\pi n)}+\frac{1}{(t-\pi n)^2(t-2\pi n)}\bigg\}dt\\
&=\int_{0}^{\pi n-3\pi/2}\frac{\sin^2(t)}{t}\bigg\{\frac{10\pi^2 n^2t+2 t^3}{(t+\pi n)^2(t+2\pi n)(t-\pi n)^2(t-2\pi n)}\bigg\}dt=\{t=\pi n x\}\\
&=\frac{1}{(\pi n)^3}\int_{0}^{1-3/(2n)}\frac{\sin^2(\pi n x)}{x}\bigg\{\frac{10x+2 x^3}{(x+1)^2(x+2)(x-1)^2(x-2)}\bigg\}dx
\end{align*}
Let
\begin{align*}
R(x)=\frac{5+2x^2}{(x+1)^2(x+2)(x-1)^2(x-2)}.
\end{align*}
An elementary estimate gives
\begin{align*}
-\frac{7}{2}\frac{1}{(x-1)^2} \leq R(x)\leq-\frac{5}{32}\frac{1}{(x-1)^2} 
\end{align*}
valid for all $x\in [0,1)$. Hence
\begin{align*}
&\int_{0}^{1-3/(2n)}\sin^2(\pi n x)R(x)dx\leq -\frac{5}{32}\int_{0}^{1-3/(2n)}\frac{\sin^2(\pi n x)dx}{(x-1)^2}.
\end{align*}
Furthermore, choose $0<\eps<\pi/2$. Then 
\begin{align*}
&\frac{1}{\pi n}\int_{0}^{1-3/(2n)}\frac{\sin^2(\pi n x)dx}{(x-1)^2}=\int_{0}^{\pi n-3\pi/2}\frac{\sin^2(t)dt}{(t-\pi n)^2}\geq \sum_{k=1}^{n-1}\int_{(k-1)\pi}^{k\pi}\frac{\sin^2(t)dt}{(t-\pi n)^2}\\
&= \sum_{k=1}^{n-1}\int_{0}^{\pi}\frac{\sin^2(x)dx}{(x+\pi(k-1)-\pi n)^2}\geq \sin^2(\eps)  \sum_{k=1}^{n-1}\int_{\eps}^{\pi-\eps}\frac{dx}{(x+\pi(k-1)-\pi n)^2}\\
&=\sin^2(\eps)  \sum_{k=1}^{n-1}\frac{\pi-2\eps}{(\pi-\eps+\pi(k-1)-\pi n)(\eps+\pi(k-1)-\pi n)}\geq \frac{(\pi-2\eps)\sin^2(\eps)}{\pi^2 n}  \sum_{k=1}^{n-1}\frac{1}{(1-\frac{(k-1)}{n})^2}\frac{1}{n}\\
&\leq \frac{(\pi-2\eps)\sin^2(\eps)}{\pi^2 n} \int_{0}^{1-3/n}\frac{dx}{(1-x)^2}\geq  \frac{(\pi-2\eps)\sin^2(\eps)}{\pi^2 n}\frac{n}{3}=\frac{(\pi-2\eps)\sin^2(\eps)}{3\pi^2 }.
\end{align*}
Thus, 
\begin{align}
\label{HilbTransformEst1}
&\int_{0}^{1-3/(2n)}\sin^2(\pi n x)R(x)dx\leq -\frac{5}{32}\frac{(\pi-2\eps)\sin^2(\eps)n}{3\pi }.
\end{align}
In addition,
\begin{align}
\label{HilbTransformEst2}
&\int_{0}^{1-3/(2n)}\sin^2(\pi n x)R(x)dx\geq -\frac{7}{2}\int_{0}^{1-3/(2n)}\frac{dx}{(x-1)^2}\geq -\frac{7n}{3}
\end{align}
Finally, we have the estimate
\begin{align}
\label{HilbTransformEst3}
\int_{2\pi n-3\pi/2}^{+\infty}\frac{\sin^2(x)dx}{x^2(x-\pi n)(x+\pi n)}&=\frac{1}{(\pi n)^3}\int_{2-3/(2n)}^{+\infty}\frac{\sin^2(\pi n t)dt}{t^2(t-1)(t+1)}\nonumber \\
&\leq \frac{1}{(\pi n)^3}\int_{2-3/(2n)}^{+\infty}\frac{dt}{(t-1)^4}=\frac{1}{(\pi n)^3}\frac{1}{3(1-\frac{3}{2n})^3}\leq \frac{4^3}{3(\pi n)^3}.
\end{align}
Estimates (\ref{HilbTransformEst1})-(\ref{HilbTransformEst3}) gives 
\begin{align}
\label{FundamentaIneq}
-\frac{14}{3 \pi^2 n}-\frac{2\cdot 4^3}{3\pi^2 n^2}\leq \pi\mathcal{H}f_1((\pi n)^{-1})\leq -\frac{10}{32}\frac{(\pi-2\eps)\sin^2(\eps)}{3\pi^3 n }+\frac{2\cdot 4^3}{3\pi^2 n^2}.
\end{align}
(\ref{FundamentaIneq}) shows that there exists an $N>0$ such that $\mathcal{H}f_1((\pi n)^{-1})<0$ for all $n>N$. Moreover, (\ref{FundamentaIneq}) also shows that
\begin{align*}
\sup_{n>2}\{\vert \mathcal{H}f_1((\pi n)^{-1})\vert \}<+\infty.
\end{align*}

Let $f_{2}(t)=\chi_{[1/2,a]}(t)$ and choose $a$ so that $\int_{\R}(f_1(t)+f_2(t))dt=1$. Let $f=f_1+f_2$ be the density of the measure $\mu$. Since $\pi\mathcal{H}f_2(x)<0$ for all $x<1/2$, and $\pi\mathcal{H}f_2(x)$ is a continuous function on $[-\frac{2}{3\pi},\frac{2}{3\pi}]$ it follows that 
\begin{align}
\label{HilbTransformEst4}
\inf_{n>N}\{\vert \mathcal{H}f((\pi n)^{-1})\vert \}=m_2>0.
\end{align}
Furthermore,
\begin{align}
\label{derivPoissonEst2}
\sup_{n\geq 2}\bigg\{\int_{\R}\frac{f(t)dt}{(\frac{1}{\pi n}-t)^2}\bigg\}:=m_1<+\infty.
\end{align}
Thus we have shown that 
\begin{align*}
\bigg\{\frac{1}{\pi n}:n>N\bigg\}\bigcup\{0\}\subset \mathcal{S}_{nt}^{sing,III}(\mu).
\end{align*}
Moreover, every $x\in \{(\pi n)^{-1}:n>N\}$ satisfies the assumptions of Proposition \ref{GeoSing2}. In addition, $f$ is Lipschitz continuous on $(-\frac{2}{3\pi},\frac{2}{3\pi})$. Therefore, the function $H_vf(u)$ defined on $\{u+iv\in \Hp: u\in [-\frac{2}{3\pi}+\delta,\frac{2}{3\pi}-\delta]\}$, for some $\delta>0$ small enough, have a continuous extension to the domain $\{u+iv\in \C: u\in [-\frac{2}{3\pi}+\delta,\frac{2}{3\pi}-\delta], v\geq 0\}$ with boundary value $\mathcal{H}f(u)$. Proposition \ref{GeoSing2} and (\ref{HilbTransformEst4}) and (\ref{derivPoissonEst2}) implies that there exists a $d>0$ such that
\begin{align*}
\mathcal{H}^1(\dv \LL((\pi n)^{-1}))>\vert 1-\eta_{\G}^{III}((\pi n)^{-1}))\vert>d
\end{align*}
for all $n>N$. This immediately implies that $\mathcal{H}^1(\dv \LL)=+\infty$.
\end{proof}

\section{Appendix}
\subsection{Additional Results}
\begin{Lem}
\label{lemLBoundary}
\begin{align}
\dv \LL=\dv \LL(\infty)\bigcup\bigg(\bigcup_{x\in \R}\dv\LL(x)\bigg).
\end{align}
\end{Lem}
\begin{proof}
Let $\omega_x=\{w_n\}_n\subset \Hp$ be a sequence such that $\lim_{n\to \infty}w_n=x$. Then $\{W_{\LL}^{-1}(w_n)\}\subset \LL$. Since $\LL\subset \mathcal{P}$, $\overline{\LL}$ is compact.
By Heine-Borel thoerem, it follows that $\dv \LL_{[\omega]}(x)\neq \varnothing$. Assume that there exists a point $(\chi',\eta')\in \LL\bigcap\dv \LL_{[\omega]}(x)$. Then there exists a subsequence $\{w_{n_k}\}_k$ such that $(\chi_\LL(w_{n_k}),\eta_\LL(w_{n_k}))\to (\chi',\eta')$.  However, since $W_\LL$ is a homeomorphism, it follows that $\lim_{k\to \infty}w_{n_k}=w'=W_\LL((\chi',\eta'))$, a contradiction. Hence $\dv \LL_{[\omega]}(x)\subset \dv \LL$. Since this holds for every such sequence $\omega=\omega_x$, it follows that 
\begin{align*}
\dv \LL(x)=\bigcup_{[\omega]\in S_x}\dv \LL_{[\omega]}(x)\subset \dv \LL.
\end{align*}
In particular this holds for every $x\in \R$. Thus,
\begin{align*}
\bigcup_{x\in \R}\dv \LL(x)\subset \dv \LL.
\end{align*}
Finally, Lemma 2.1 in \cite{Duse14a} proves that for any sequence $\{w_n\}_n\subset \Hp$, such that $\lim_{n\to\infty}\vert w_n\vert=\infty$, $\lim_{n\to \infty}(\chi_\LL(w_{n}),\eta_\LL(w_{n}))=(\frac{1}{2}+\int_{\R}xd\mu(x),0)\in \dv \LL$. This shows that
\begin{align*}
\dv \LL(\infty)\bigcup\bigg(\bigcup_{x\in \R}\dv\LL(x)\bigg)\subset \dv \LL. 
\end{align*}
We now show the reverse inclusion. Let $(\chi',\eta')\in \dv \LL$. Then there exists a sequence $\{(\chi_n,\eta_n)\}_n\subset \LL$ such that $\lim_{n\to \infty}(\chi_n,\eta_n)=(\chi',\eta')$. Let $w_n=W_{\LL}((\chi_n,\eta_n))$. Assume that the sequence $\{w_n\}_n$ is unbounded. Then it contains a subsequence $\{w_{n_k}\}_k$ such that $\lim_{k\to \infty}\vert w_{n_k}\vert=\infty$. Then Lemma 2.1 in \cite{Duse14a} shows that $(\chi',\eta')=(\frac{1}{2}+\int_{\R}xd\mu(x),0)$. However, this implies that $\lim_{n\to \infty}\vert w_n\vert=\infty$. Thus, we may assume that the sequence $\{w_n\}_n$ is bounded in $\Hp$. Consider the set of limit points of $\{w_n\}_n$, that is $\overline{\{w_n\}_n}\backslash \{w_n\}_n$. Assume that $w'\in \overline{\{w_n\}_n}\backslash \{w_n\}_n\bigcap \Hp$. Then there exists a subsequence $\{w_{n_k}\}_k$ such that $\lim_{k\to\infty}w_{n_k}=w'$. However since $W_{\LL}$ is a homeomorphism, this implies that $\lim_{k\to \infty}W_{\LL}^{-1}(w_{n_k})=W_{\LL}^{-1}(w')\neq (\chi',\eta')$, a contradiction. Thus $\overline{\{w_n\}_n}\backslash \{w_n\}_n\subset \R$. This shows that
\begin{align*}
\dv \LL\subset \dv \LL(\infty)\bigcup\bigg(\bigcup_{x\in \R}\dv\LL(x)\bigg). 
\end{align*}
\end{proof}
\begin{Prop}
\label{DiniConv}
Let $f\in L^{p}(\R)$ where $p>1$. Assume that 
\begin{align}
\label{Dini}
\int_{x-1}^{x+1}\frac{\vert f(x)-f(t) \vert dt}{\vert x-t \vert}<+\infty.
\end{align}
Then for every non-tangential convergent sequence $\{u_n+iv_n\}_n$ to x,
\begin{align}
\label{HilbDini}
\lim_{n\to \infty} H_{v_n}f(u_n)=\mathcal{H}f(x).
\end{align}
Moreover, $x\in \mathscr{L}_f$.
\end{Prop}
\begin{proof}
We first note that (\ref{Dini}) implies that 
\begin{align*}
\lim_{h\to 0^+}\int_{x-h}^{x+h}\frac{\vert f(x)-f(t) \vert dt}{\vert x-t \vert}\geq \lim_{h\to 0^+}\int_{x-h}^{x+h}\frac{\vert f(x)-f(t) \vert dt}{2h}=0.
\end{align*}
Thus, $x\in \mathscr{L}_f$. We now show that $\mathcal{H}f(x)$ exists. We have for every $R>0$ sufficiently large
\begin{align*}
\int_{\vert x-t\vert>\eps}\frac{f(t)dt}{x-t}&=\int_{\eps<\vert x-t\vert<R}\frac{(f(t)-f(x))dt}{x-t}+f(x)\underbrace{\int_{\eps<\vert x-t\vert<R}\frac{dt}{x-t}}_{=0}+\int_{\vert x-t\vert>R}\frac{f(t)dt}{x-t}\\
&=J_1+J_2.
\end{align*}
We first estimate $I_2$. Since $f\in L^p(\R)$ we have by H\"older's inequality
\begin{align*}
\vert J_2\vert\leq \Vert f\Vert_p\bigg(\int_{\vert x-t\vert>R}\frac{dt}{\vert x-t\vert^q}\bigg)^{1/q}=\frac{ \Vert f\Vert_p2^{1/q}}{(q-1)^{1/q}R^{(q-1)/q}},
\end{align*}
where $q=\frac{p}{p-1}>1$.
Moreover, since
\begin{align*}
\lim_{\eps\to 0^+}\frac{(f(x)-f(t))\chi_{\vert t\vert >\eps}}{x-t}=\frac{f(x)-f(t)}{x-t}
\end{align*}
for all $t\neq x$, and
\begin{align*}
\frac{\vert f(t)-f(x)\vert \chi_{\vert t\vert >\eps}}{\vert x-t\vert}\leq \frac{\vert f(t)-f(x) \vert}{\vert x-t \vert},
\end{align*}
it follows by (\ref{Dini}) and Lebesgue dominated convergence theorem that
\begin{align*}
\lim_{\eps\to 0^+}\int_{\eps<\vert x-t\vert <R}\frac{f(t)-f(x)}{x-t}dt=\int_{\vert x-t\vert <R}\frac{f(t)-f(x)}{x-t}dt.
\end{align*}
Since $R>0$ was arbitrary and $f\in L^p(\R)$ it follows that
\begin{align*}
\pi \mathcal{H}f(x)=\lim_{R\to \infty}\int_{\vert x-t\vert <R}\frac{f(x)-f(t)}{x-t}dt=\int_{\R}\frac{f(x)-f(t)}{x-t}dt
\end{align*}
exists. Now consider a non-tangentially convergent sequence $\{u_n+iv_n\}_n$ to $x$. Then $\{u_n+iv_n\}_n\subset \G_k(x)$ for some $k>0$. We may assume that $u_n-x\geq 0$. Then
\begin{align*}
\label{HilbDini}
-\pi H_{v_n}f(u_n)&=\int_{\R}\frac{-(u_n-t)f(t)dt}{(u_n-t)^2+v_n^2}=\int_{\R}\frac{(u_n-t)(f(x)-f(t))dt}{(u_n-t)^2+v_n^2}\\&=\int_{x}^{x+2(u_n-x)}\frac{(u_n-t)(f(x)-f(t))dt}{(u_n-t)^2+v_n^2}+\int_{\R\backslash [x,x+2(u_n-x)]}\frac{(u_n-t)(f(x)-f(t))dt}{(u_n-t)^2+v_n^2}\\
&:=I_1^{(n)}+I_2^{(n)},
\end{align*}
We first consider $I_2^{(n)}$. 
By Lemma \ref{lemIneqRational}
\begin{align*}
\frac{\vert u_n-t\vert \chi_{\R\backslash [x,x+2(u_n-x)]}(t)}{(u_n-t)^2+v_n^2}\leq \frac{2}{\vert x-t\vert}
\end{align*}
for all $t$, and 
\begin{align*}
\lim_{n\to \infty}\frac{(u_n-t)(f(x)-f(t)) \chi_{\R\backslash [x,x+2(u_n-x)]}(t)}{(u_n-t)^2+v_n^2}=\frac{f(x)-f(t)}{x-t}
\end{align*}
for all $t\neq x$, we get form (\ref{Dini}) and Lebesgue's dominated convergence theorem that 
\begin{align*}
\lim_{n\to \infty}\int_{\R\backslash [x,x+2(u_n-x)]}\frac{(u_n-t)(f(x)-f(t))dt}{(u_n-t)^2+v_n^2}=-\pi \mathcal{H}f(x).
\end{align*}
We now consider $I_1^{(n)}$. Since $\{u_n+iv_n\}_n\subset \G_k(x)$, 
\begin{align*}
\vert I_1^{(n)}\vert&\leq \int_{x}^{u_n+(u_n-x)}\frac{\vert (u_n-t)(f(x)-f(t))\vert dt}{(u_n-t)^2+v_n^2}\leq \frac{1}{v_n}\int_{x}^{u_n+(u_n-x)}\vert f(x)-f(t)\vert dt\\
&\leq \frac{(u_n-x)}{v_n}\frac{1}{(u_n-x)}\int_{x-2(u_n-x)}^{x+2(u_n-x)}\vert f(x)-f(t)\vert dt\leq k\frac{1}{(u_n-x)}\int_{x-2(u_n-x)}^{x+2(u_n-x)}\vert f(x)-f(t)\vert dt\to 0
\end{align*}
as $n\to \infty$, since $x\in \mathscr{L}_f$.
 \end{proof}

\begin{Lem}
\label{MinMon}
Let $f,g:[0,+\infty)\to[0,+\infty)$ be monotonically decreasing continuous functions. Then $\min\{f,g\}:[0,+\infty)\to[0,+\infty)$ is also a continuous monotonically decreasing function.
\end{Lem}
 
\begin{proof}
Using the identity
\begin{align}
\min\{ f(x),g(x)\}=\frac{1}{2}\big( f(x)+g(x)-\vert f(x)-g(x)\vert \big),
\end{align}
the continuity of $\min\{f,g\}$ is immediate. Let $y>x$. Then by (6.4) and the triangle inequality
\begin{align*}
&2(\min\{f(y),g(y)\}-\min\{f(x),g(x)\})\leq \big( f(y)+g(y)-\vert f(y)-g(y)\vert \big)-\big( f(x)+g(x)-\vert f(x)-g(x)\vert \big)\\
&\leq f(y)-f(x)+g(y)-g(x)+\vert f(x)-f(y)+g(y)-g(x)\vert\\
&\leq f(y)-f(x)+\vert f(y)-f(x)\vert+g(y)-g(x)+\vert g(y)-g(x)\vert.
\end{align*}
By monotonicity, $f(y)-f(x)\leq 0$ and $g(y)-g(x)\leq 0$, which implies 
\begin{align*}
2(\min\{f(y),g(y)\}-\min\{f(x),g(x)\})\leq f(y)-f(x)-(f(y)-f(x))+g(y)-g(x)-(g(y)-g(x))\leq 0.
\end{align*}
Hence $\min\{f,g\}$ is monotonically decreasing.
\end{proof}
 
\begin{Lem}
\label{lemIneqRational}
Assume $u_n>x$. Then
\begin{align*}
\frac{\vert u_n-t\vert\chi_{[x,x+2(u_n-x)]}(t)}{(u_n-t)^2+v_n^2}\leq \frac{1}{\vert u_n-t\vert}\leq \frac{2}{\vert x-t\vert}
\end{align*}
for all $t$.
\end{Lem}

\begin{proof}
By translation invariance we may assume that $x=0$.
Clearly, $\vert u_n-t\vert>\vert t\vert$ when $t<0$ which implies 
\begin{align*}
 \frac{1}{\vert u_n-t\vert}\leq \frac{1}{\vert t\vert}\leq  \frac{2}{\vert t\vert}
\end{align*}
Now assume that $t>u_n>0$. Then the equation
\begin{align*}
 \frac{1}{\vert u_n-t\vert}= \frac{2}{\vert t\vert}
\end{align*}
has the unique solution $t=2u_n$. This immediately implies that 
\begin{align*}
 \frac{1}{\vert u_n-t\vert}\leq \frac{2}{\vert t\vert}
\end{align*}
for $t\geq 2u_n$.
\end{proof}

\begin{Lem}
\label{lemIncExcMeasure}
Let $a,b\in \R$, and $a<b$. Assume that $X_1,X_2\subset[a,b]$ are measurable and that $\lambda(X_1)=m_1$ and $\lambda(X_2)=m_2$ are such that $m_1+m_2>b-a$.
Then
\begin{align*}
\lambda(X_1\bigcup X_2)\geq m_1-m_2-b+a.
\end{align*}
\end{Lem}
\begin{proof}
By the inclusion-exclusion principle
\begin{align*}
\lambda(X_1\bigcup X_2)=\lambda(X_1)+\lambda(X_2)-\lambda(X_1\bigcap X_2).
\end{align*}
Since $\lambda(X_1\bigcup X_2)\leq b-a$ this gives the inequality
\begin{align*}
b-a\geq m_1+m_2-\lambda(X_1\bigcap X_2).
\end{align*}
which is equivalent to
\begin{align*}
\lambda(X_1\bigcap X_2)\geq m_1+m_2-b+a.
\end{align*}
\end{proof}

\subsection{Examples}
In this section we provide some details of the examples given earlier.

\begin{ex} 
Let $f(t)=\vert \sin t^{-1}\vert\chi_{[-1/2,1/2]}(t)+\chi_{[1/2,a]}(t)$, where $a$ is chosen so that $\int_{\R}f(t)dt=1$. Then aprior $f$ is not defined at $x=0$. However, we note that 
\begin{align*}
Mf(0,h)&=\frac{1}{2h}\int_{-h}^h\vert \sin(t^{-1})\vert dt=\frac{1}{h}\int_{1/h}^{+\infty}\frac{\vert \sin(s)\vert}{s^2}ds\\
&=\frac{1}{h}\int_{1/h}^{\pi m}\frac{\vert \sin(s)\vert}{s^2}ds+\frac{1}{h}\int_{\pi m}^{+\infty}\frac{\vert \sin(s)\vert}{s^2}ds
\end{align*}
for $h$ sufficiently small, where $m$ is the smallest integer such that $1/h\leq \pi m$. Hence,
\begin{align*}
Mf(0,h)&<\frac{1}{h}\int_{1/h}^{\pi m}\frac{1}{h^{-2}}ds+\sum_{n=m}^{+\infty}\frac{1}{h}\int_{\pi n}^{\pi(n+1)}\frac{\vert \sin(s)\vert}{\pi^2 n^2}ds\\
&= h(\pi m-h^{-1})+\sum_{n=m}^{+\infty}\frac{2}{h\pi^2n^2}\\
&\leq \pi h+\sum_{n=m}^{+\infty}\frac{2m}{\pi n^2}=\pi h+\frac{2m}{\pi}\sum_{k=0}^{+\infty}\frac{1}{(k+m)^2}\leq \pi h+\frac{2m}{\pi}\int_{m-1}^{+\infty}\frac{dx}{x^2}\\
&=\pi h+\frac{2}{\pi}\frac{m}{m-1},
\end{align*}
where we have used that $\pi(m-1)\leq h^{-1}\leq \pi m$, which implies that $h(\pi m-h^{-1})\leq h(\pi m-\pi (m-1))=\pi h$. Similarly,
\begin{align*}
Mf(0,h)&>\sum_{n=m}^{+\infty}\frac{1}{h}\int_{\pi n}^{\pi(n+1)}\frac{\vert \sin(s)\vert}{\pi^2 (n+1)^2}ds
=\sum_{n=m}^{+\infty}\frac{2}{h\pi^2(n+1)^2}\geq\sum_{n=m}^{+\infty}\frac{2(m-1)}{\pi (n+1)^2}\\&=\frac{2(m-1)}{\pi}\sum_{k=0}^{+\infty}\frac{1}{(k+m+1)^2}\geq \frac{2m}{\pi}\int_{m+2}^{+\infty}\frac{dx}{x^2}\\
&=\frac{2}{\pi}\frac{m-1}{m+2}.
\end{align*}
Hence, since $m\to \infty$ and $h\to 0^+$
\begin{align*}
\lim_{h\rightarrow 0^+}Mf(0,h)=\frac{2}{\pi}.
\end{align*}
Therefore, if $0$ is to be a Lebesgue point of $f$ we have to define $f(0):=\frac{2}{\pi}$. However,
\begin{align*}
&\frac{1}{2h}\int_{-h}^h\bigg\vert\vert \sin(t^{-1})\vert -\frac{2}{\pi}\bigg\vert dt=\frac{1}{h}\int_{1/h}^{+\infty}\bigg\vert\vert \sin(t)\vert -\frac{2}{\pi}\bigg\vert \frac{dt}{t^2},
\end{align*}
and
\begin{align*}
\vert \sin(t)\vert -\frac{2}{\pi}>0 \quad \text{for $t\in\bigcup_{k=0}^{+\infty}\bigg(\arcsin\frac{2}{\pi}+\pi k,\pi-\arcsin\frac{2}{\pi}+\pi k\bigg)$ and $t>0$.}
\end{align*}
Therefore, 
\begin{align*}
&\int_{\pi k}^{\pi (k+1)}\bigg\vert\vert \sin(t)\vert -\frac{2}{\pi}\bigg\vert \frac{dt}{t^2}\geq\int_{\pi k+\arcsin(2/\pi)}^{\pi-\arcsin(2/\pi)+\pi k}\bigg(\sin(t) -\frac{2}{\pi}\bigg)\frac{dt}{\pi^2(k+1)^2}\\&=\frac{1}{\pi^2(k+1)^2}\bigg[-\cos(t)-\frac{2t}{\pi}\bigg]_{\pi k+\arcsin(2/\pi)}^{\pi-\arcsin(2/\pi)+\pi k}=\frac{2}{\pi^2(k+1)^2}\bigg[\sqrt{1-\bigg(\frac{2}{\pi}\bigg)^2}-1+\frac{2}{\pi}\arcsin \bigg(\frac{2}{\pi}\bigg)\bigg].
\end{align*}
Let $m$ be the smallest integer such that $1/h<\pi m$. Then by the above there exist a positive constant $c>0$ such that
\begin{align*}
\frac{1}{h}\int_{1/h}^{+\infty}\bigg\vert\vert \sin(t)\vert -\frac{2}{\pi}\bigg\vert \frac{dt}{t^2}\geq \pi(m-1)\sum_{k=m}^{+\infty}\frac{c}{(k+1)^2}\geq \frac\pi c(m-1)\int_{m}^{+\infty}\frac{1}{(x+2)^2}dx=\frac{\pi c(m-1)}{(m+2)}\nrightarrow 0
\end{align*}
as $m \rightarrow +\infty$. Note that the last term does not converge to 0 as $h\to 0^+$. Therefore $0\neq \mathscr{L}_f$. Since, $f^{+}(0)=f^{-}(0)=\frac{2}{\pi}<1$, it follows from Proposition \ref{NLpointReg} that $0$ is a regular point.
However note that some of the the points $\Big\{\frac{1}{\frac{\pi}{2}+\pi k}: k\in\Z, \vert k\vert\geq 2\Big\}$ may be singular.
\end{ex}

\begin{ex}
Let $I_n=(2^{-(n+1)},2^{-n}]$ and let
\begin{align*}
f(t)=\sum_{k=1}^{+\infty}\bigg(1-\frac{1}{2k}\bigg)(\chi_{I_{2k}}(t)+\chi_{I_{2k}}(-t))+\sum_{k=1}^{+\infty}\frac{1}{2k+1}(\chi_{I_{2k+1}}(t)+\chi_{I_{2k+1}}(-t))+\chi_{(1/2,a]}(t),
\end{align*}
where $a$ is chosen so that $\int_{\R}f(t)dt=1$. One can show that
\begin{align*}
\limsup_{h\rightarrow 0^+}M_Rf(0,h)&=\limsup_{h\rightarrow 0^+}M_Lf(0,h)
=\lim_{k\rightarrow +\infty}\frac{1}{2^{-2k}}\int_{0}^{2^{-2k}}f(t)dt\\
&= \lim_{k\rightarrow +\infty}2^{2k}\bigg(\int_0^{2^{-2k-2}}f(t)dt+\int_{2^{-2k-2}}^{2^{-2k-1}}f(t)dt+\int_{2^{-2k-1}}^{2^{-2k}}f(t)dt\bigg)\\
&\leq  \lim_{k\rightarrow +\infty}2^{2k}\bigg(2^{-2k-2}+\frac{1}{2k+1}(2^{-2k-1}-2^{-2k-2})+ \bigg(1-\frac{1}{2k}\bigg)(2^{-2k}-2^{-2k-1})\bigg)\\
&=\frac{3}{4}<1.
\end{align*}
Similarly, one can also show that
\begin{align*}
\liminf_{h\rightarrow 0^+}M_Rf(0,h)&=\liminf_{h\rightarrow 0^+}M_Lf(0,h)
=\lim_{k\rightarrow +\infty}\frac{1}{2^{-2k-1}}\int_{0}^{2^{-2k-1}}f(t)dt\\
&= \lim_{k\rightarrow +\infty}2^{2k+1}\bigg(\int_0^{2^{-2k-3}}f(t)dt+\int_{2^{-2k-3}}^{2^{-2k-2}}f(t)dt+\int_{2^{-2k-2}}^{2^{-2k-1}}f(t)dt\bigg)\\
&\geq  \lim_{k\rightarrow +\infty}2^{2k+1}\bigg(\bigg(1-\frac{1}{2k+2}\bigg)(2^{-2k-2}-2^{-2k-3})+ \frac{1}{2k+1}(2^{-2k}-2^{-2k-1})\bigg)\\
&=\frac{1}{4}>0.
\end{align*}
Hence, $x=0$ does not belong to the Lebesgue set of $f$. However we see that $f$ satisfies the conditions of Proposition \ref{NLpointReg}. This implies that $0$ is a regular point. 
\end{ex}


\begin{thebibliography}{10}


\bibitem{BSimon}
{\sc Barry Simon.}
Convexity: An Analytic Viewpoint
{\emph{Cambridge University Press}}
2011.


\bibitem{Duse14a}
{\sc E. Duse and K. Johansson and A. Metcalfe.}
Asymptotic geometry of discrete interlaced patterns: Part I
{\em arXiv:1507.0047}
(to appear in International Journal of Mathematics)

\bibitem{Duse15c}
{\sc E. Duse and A. Metcalfe.}
Universal edge fluctuations of discrete interlaced particle systems.
(estimated 2015).


\bibitem{Duse15d}
{\sc E. Duse and K. Johansson and A. Metcalfe.}
The Cusp-Airy Process
(estimated 2015).



\bibitem{Garnett}
{\sc John Garnett}
Bounded Analytic Functions
{\emph{Springer}}
2006.



\bibitem{Garf}
{\sc Loukas Garfakos}
Classical Fourier Analysis, second edition
{\emph{Springer}}
2009.



\bibitem{Ken07}
{\sc R. Kenyon and A. Okounkov.}
Limit shapes and the complex Burgers equation.
{\em Acta Mathematica.}
199,2 (2007) 263--302



\bibitem{Met13}
{\sc A. Metcalfe.}
Universality properties of Gelfand-Tsetlin patterns.
{\em Probability Theory and Related Fields.}
155,1-2 (2013) 303--346.



\bibitem{Musk}
{\sc I.N. Muskhelishivili}
Singular Integral Equations
{\emph{Dover Books on Physics}}
2008.



\bibitem{Pet12}
{\sc L. Petrov.}
Asymptotics of Random Lozenge Tilings via Gelfand-Tsetlin Schemes.
{\em http://arxiv.org/abs/1202.3901.}
(2012)

\bibitem{Pet15}
{\sc L. Petrov.}
Asymptotics of uniformly random lozenge tilings of polygons. Gaussian free field.
{\em Annals of Probability}
Vol. 43 (2015), No. 1, 1-43.



\bibitem{Rudin}
{\sc W. Rudin.}
Well-distributed measurable sets.
{\em Amer. Math. Monthly.}
90 (1983), no. 1, 41-42.




\bibitem{Sharp}
{\sc Bennet, DeVore and Sharpley}
Weak-$L^{\infty}$ and BMO.
{\em Ann. of Math.}
113:601-611.



\bibitem{Stein}
{\sc Elias M. Stein.}
Singular Integrals and Differentiability Properties of Functions
{\emph{Princeton University Press}.}
1971.



\bibitem{SteinW1}
{\sc Elias M. Stein and Guido Weiss.}
Fourier Analysis on Euclidean Spaces
{\emph{Princeton University Press}.}
1971.




\bibitem{SteinW2}
{\sc Elias M. Stein}
Harmonic Analysis
{\emph{Princeton University Press}}
1993.






\end{thebibliography}
\end{document}